\documentclass{theoretics}

\usepackage{adaptive}
\usepackage{verbatim}

\newcommand{\gr}{\mathrm{GG}}
\newcommand{\cc}{\mathrm{CC}}

\newcommand{\rmed}{r_{\mathrm{med}}}
\newcommand{\rhat}{r_{\mathrm{apx}}}

\newcommand{\rlow}{r_{\mathrm{low}}}
\newcommand{\rhigh}{r_{\mathrm{high}}}
\newcommand{\rlowt}{r^{\mrm{term}}_{\mathrm{low}}}
\newcommand{\rhight}{r^{\mrm{term}}_{\mathrm{high}}}

\newcommand{\rlhc}{C_{\mathrm{range}}}
\newcommand{\crep}{C_{\mathrm{rep}}}
\newcommand{\spcomp}{\mathrm{space}}
\newcommand{\tcomp}{\mathrm{time}}
\newcommand{\chigh}{\mc{C}_{\mathrm{high}}}
\newcommand{\clow}{\mc{C}_{\mathrm{low}}}
\newcommand{\cprob}{c_{\mathrm{prob}}}
\newcommand{\ru}{\rho_u}
\newcommand{\rc}{\rho_c}
\newcommand{\xhat}{\hat{x}}
\newcommand{\tres}{\mc{T}_{\mrm{res}}}

\newcommand{\rone}{\rho_1}
\newcommand{\rtwo}{\rho_2}
\newcommand{\rthree}{\rho_3}
\newcommand{\rfour}{\rho_4}
\newcommand{\rrep}{\rho_{\mathrm{rep}}}
\newcommand{\nrep}{n_{\mathrm{rep}}}

\DeclareMathOperator*{\conv}{\mathrm{Conv}}

\newcommand{\comprmed}{\mathrm{CompRmed}}
\newcommand{\conpart}{\mathrm{ConstructPartition}}
\newcommand{\fixedscaleinst}{\mathrm{FixedScaleInstantiation}}
\newcommand{\fixedscalequery}{\mathrm{FixedScaleQuery}}
\newcommand{\multiscaleinst}{\mathrm{MultiScaleInstantiation}}
\newcommand{\multiscalequery}{\mathrm{MultiScaleQuery}}
\newcommand{\locn}{\mc{N}_\mathrm{loc}}

\newcommand{\xt}[1]{x^{(#1)}}
\newcommand{\vt}[1]{v^{(#1)}}
\newcommand{\ut}[1]{u^{(#1)}}
\newcommand{\At}[1]{A^{(#1)}}

\newcommand{\tv}{\wt{v}}

\newcommand{\Size}{\mathrm{Size}}

\newcommand{\fail}{\mathrm{FAIL}}

\newcommand{\stepsizec}{c_{\mathrm{step}}}
\newcommand{\stepsize}{\frac{\stepsizec \cdot \eps^3}{\log^3 n}}

\DeclareMathOperator*{\vol}{Vol}

\def\authornotes{1pt}

\ifdim\authornotes=1pt
\newcommand{\ynote}[1]{\footnote{\color{ForestGreen}Yeshwanth: #1}}
\newcommand{\jnote}[1]{\footnote{\color{Orange}Jelani: #1}}
\else
\newcommand{\ynote}[1]{}
\newcommand{\jnote}[1]{}
\fi

\newcommand{\RestateRemark}[1]{{\normalfont\bfseries #1}}
\newcommand{\RestateInit}[1]{\newcommand{#1}{}}
\newcommand{\RestateGo}[1]{\renewcommand{#1}{(Restated)}}

\title{Terminal Embeddings in Sublinear Time}

\ThCSauthor{Yeshwanth Cherapanamjeri}{yesh@mit.edu}[https://orcid.org/0000-0001-8718-8250]
\ThCSaffil{CSAIL, M.I.T}
\ThCSauthor{Jelani Nelson}{minilek@berkeley.edu}[https://orcid.org/0000-0001-7370-3733]
\ThCSaffil{UC Berkeley}

\ThCSshortnames{Y.\ Cherapanamjeri and J.\ Nelson}
\ThCSshorttitle{Terminal Embeddings in Sublinear Time}

\ThCSthanks{A preliminary version of this article appeared at FOCS 2021. The first author was supported by a Microsoft Research BAIR Commons Research Grant. The second author was supported by NSF
  award CCF-1951384, ONR grant N00014-18-1-2562, ONR DORECG award N00014-17-1-2127, and a Google Faculty Research Award.}

\ThCSyear{2024}
\ThCSarticlenum{6}
\ThCSdoicreatedtrue
\ThCSreceived{Mar 4, 2022}
\ThCSrevised{Nov 16, 2023}
\ThCSaccepted{Jan 28, 2024}
\ThCSpublished{Mar 14, 2024}

\ThCSkeywords{dimensionality reduction, terminal embeddings, adaptive data structures, nearest neighbor search, sublinear algorithms}

\begin{document}
\maketitle

\begin{abstract}  
  Recently (Elkin, Filtser, Neiman 2017) introduced the concept of a {\it terminal embedding} from one metric space $(X,d_X)$ to another $(Y,d_Y)$ with a set of designated terminals $T\subset X$. Such an embedding $f$ is said to have distortion $\rho\ge 1$ if $\rho$ is the smallest value such that there exists a constant $C>0$ satisfying
  \begin{equation*}
    \forall x\in T\ \forall q\in X,\ C d_X(x, q) \le d_Y(f(x), f(q)) \le C \rho d_X(x, q) .
  \end{equation*}
  When $X,Y$ are both Euclidean metrics with $Y$ being $m$-dimensional, recently (Narayanan, Nelson 2019), following work of (Mahabadi, Makarychev, Makarychev, Razenshteyn 2018), showed that distortion $1+\eps$ is achievable via such a terminal embedding with $m = O(\eps^{-2}\log n)$ for $n := |T|$. This generalizes the Johnson-Lindenstrauss lemma, which only preserves distances within $T$ and not to $T$ from the rest of space. The downside of prior work is that evaluating their embedding on some $q\in \R^d$ required solving a semidefinite program with $\Theta(n)$ constraints in~$m$ variables and thus required some superlinear $\poly(n)$ runtime. Our main contribution in this work is to give a new data structure for computing terminal embeddings. We show how to pre-process $T$ to obtain an almost linear-space data structure that supports computing the terminal embedding image of any $q\in\R^d$ in sublinear time $O^* (n^{1-\Theta(\eps^2)} + d)$. To accomplish this, we leverage tools developed in the context of approximate nearest neighbor search.
\end{abstract}


\section{Introduction}
\label{sec:intro}

A {\it distortion-$\rho$ terminal embedding} $f:X\rightarrow Y$ for two metric spaces $(X, d_X)$, $(Y, d_Y)$ and given terminal set $T\subseteq X$ is such that
\begin{equation*}
    \forall x\in T\ \forall q\in X,\ C d_X(x, q) \le d_Y(f(x), f(q)) \le C \rho d_X(x, q).
\end{equation*}
Recently Elkin, Filtser, and Neiman \cite{ElkinFN17} showed the existence of such $f$ when $X=\R^d, Y = \R^m$ with distortion arbitrarily close to $\sqrt{10}$ for $m = O(\log|T|)$. Following \cite{mmmr18}, the recent work of \cite{nn19} showed distortion $1+\eps$ is achievable with $m = O(\eps^{-2}\log|T|)$, thus yielding a strict generalization of the Johnson-Lindenstrauss (JL) lemma \cite{jl}.

One family of motivating applications for dimensionality-reducing terminal embeddings is to high-dimensional computational geometry static data structural problems. To make things concrete, one example to keep in mind is (approximate) nearest neighbor search in Euclidean space: we would like to pre-process a database $D = \{x_1,\ldots,x_n\}\subset \R^d$ into a data structure to later answer approximate nearest neighbor queries for some other $q\in\R^d$. Given a terminal embedding $f$ with terminal set $T = D$, we can build the data structure on $f(D)$ then later answer queries by querying the data structure on $f(q)$. In this way, we can save on memory and potentially query costs by working over a lower-dimensional version of the problem.

Prior to the introduction of terminal embeddings, the typical approach to applying dimensionality reduction in the above scenario was to observe that the JL lemma actually provides a distribution over (linear embeddings) $x\mapsto \Pi x$ for $\Pi$ randomly drawn from a distribution independent of $D$ such that $\|\Pi z\|_2 =(1\pm\eps) \|z\|_2$ with high probability for any fixed $z\in\R^d$. Thus if one wants to make a sequence of queries $q_1,\ldots,q_N$, one could set the failure probability to be $\ll 1/(nN)$ to then have by a union bound that $z_{i,j} = x_i - q_j$ has its norm preserved by~$\Pi$ for all $i,j$ simultaneously with good probability. This approach though does not generally provide correctness guarantees when the sequence of queries is chosen {\it adaptively}, i.e.\ the $j$th query depends upon the responses to queries $q_1,\ldots,q_{j-1}$. This is because since those query responses depend on $\Pi$ (and perhaps also the randomness of the data structure itself, if it is randomized), future adaptive queries are not independent of $\Pi$ and the internal randomness of the data structure. Terminal embeddings circumvent this problem and can be used for adaptive queries, since if $f$ is a terminal embedding it is guaranteed to preserve distances to {\it any} future query $q$, even one chosen adaptively.

At this point we observe the gap in the above motivation of terminal embeddings and the results of prior work: we want to speed up data structures for high-dimensional problems by allowing querying for $f(q)$ (a ``simpler'' query as it is lower-dimensional) instead of $q$, but the optimal terminal embedding of \cite{nn19} require solving a semidefinite program with $\Theta(n)$ constraints and $m$ variables to compute $f(q)$. Thus the bottleneck becomes computing $f(q)$ itself, which may be even slower than exactly solving the original data structural problem in the higher-dimensional space (note nearest neighbor search can be solved exactly in linear $O(nd)$ time!). Even the $\sqrt{10}$-distortion terminal embedding construction of \cite{ElkinFN17} (and all subsequent work) requires computing the nearest neighbor of $q$ in $D$ in order to compute $f(q)$, and thus clearly cannot be used to speed up the particular application of approximate nearest neighbor search.

\paragraph{Our Contribution.} We give a new terminal embedding construction for Euclidean space into dimension $O(\eps^{-2}\log n)$, where $n$ is the number of terminals, together with a Monte Carlo procedure for computing $f(q)$ given any query $q\in\R^d$. Specifically, we compute $f(q)$ correctly with high probability, even if $q$ is chosen adaptively, and our running time to compute $f(q)$ is {\it sublinear} in $n$. Specifically, we can compute $f(q)$ in time $O^*(n^{1 - \Theta(\eps^2)} + d)$. The space complexity of our data structure to represent the terminal embedding is at most $O^*(nd)$. Our techniques also yield faster query times if more space is allowed (see \cref{thm:term_embedding}). 

\medskip

\begin{remark}
Our main theorem statement (\cref{thm:term_embedding}) for computing terminal embeddings does not provide the guarantee that if $q$ is queried twice, the same image $f(q)$ will be computed each time. If this property is desired, it can be achieved by increasing the query time to~$O^*(n^{1-\Theta(\eps^2)}d)$.
\end{remark}

\paragraph{Notation:} Through the paper, for $x, y \in \mb{R}^d$, $\norm{x}$ and $\inp{x}{y}$ will denote the standard Euclidean norm and inner product. For $X = \{x_i\}_{i = 1}^n \subset \mb{R}^d$, we will frequently use calligraphic letters, say $\mc{S}$, to denote subsets of the powerset of $X$. We use $\mc{D}$ and (sub/super)scripted versions to denote data structures and use $\spcomp(\mc{D})$ to denote the space complexity of $\mc{D}$ and $\tcomp(\mc{D}(q))$ to denote the time taken by $\mc{D}$ to process query $q$. We use $\mb{B}(x, r)$ to denote the set $\{y: \norm{y - x} \leq r\}$ and for $A \subset \mb{R}^d$, $\vol(A)$ denotes its volume and $\conv(A)$ denotes its convex hull. For a probabilistic event $A$, $\bm{1} \lbrb{A}$ will denote the indicator function of the event. For $n \in \mb{N}$, we use $[n]$ to denote the set $\{1, \dots, n\}$. Finally, $O^*(\cdot)$ and $o^* (\cdot)$ hide factors of $\poly (1 / \eps) \cdot (nd)^{o(1)}$ and $\poly (\eps)$ respectively while $\Ot(\cdot)$ and $\wt{\Omega}(\cdot)$ hide poly-logarithmic factors in $n$ and $d$. 

\paragraph{Organization:} In \cref{sec:techniques}, we provide an overview of our procedure to compute terminal embeddings and the key ideas involved in the analysis. Then, in \cref{sec:term}, we formally state and prove the main theorem of our paper and in \cref{sec:aann}, we develop a data structure for adaptive approximate nearest neighbor search which we use as a subroutine to prove our main result. \cref{sec:med_jl} provides a dimensionality reduction technique tailored to the adaptive setting allowing further speedups in query time, \cref{sec:cons_tree} details a recursive point partitioning data structure crucial to our construction and \cref{sec:misc_res} contains supporting results used in earlier sections.

\section{Our Techniques}
\label{sec:techniques}

The starting point of our work is the construction of \cite{mmmr18}, tightly analyzed in \cite{nn19}. Before we describe the construction and our generalization, we recall some definitions pertaining to the construction of terminal embeddings. The first is the precise parametrization of a terminal embedding used in our paper:
\begin{definition}[Terminal Embedding]
    \label{def:term_embed}
    Given $X \subset \R^d$ and $\eps \in (0, 1)$, we say $f: \R^d \to \R^k$ is an $\eps$-\emph{terminal embedding} for $X$ if:
    \begin{equation*}
        \forall x \in X, y \in \R^d: (1 - \eps) \norm{y - x} \leq \norm{f(y) - f(x)} \leq (1 + \eps) \norm{y - x}.
    \end{equation*}
\end{definition}
Next, we recall the notion of an Outer Extension \cite{ElkinFN17,mmmr18}:
\begin{definition}[Outer Extension]
    \label{def:outer_ext}
    Given $X \subset \mb{R}^d$ and $f: \mb{R}^d \to \mb{R}^k$, we say that $g: \mb{R}^d \to \mb{R}^{k^\prime}$ for $k^\prime > k$ is an \emph{outer extension} of $f$ on $X$ if:
    \begin{equation*}
        \forall x \in X: g(x) = (f(x), 0, \dots, 0).
    \end{equation*}
\end{definition}
All previous approaches \cite{ElkinFN17,mmmr18,nn19} as well as ours all construct terminal embeddings by extending a standard distance preserving embedding on $X$ by a single coordinate. Therefore, in all our subsequent discussions, we restrict ourselves to the case where $k^\prime = k + 1$. However, \cite{nn19} require the stronger property that the distance preserving embedding being extended satisfies $\eps$-convex hull distortion, allowing them to obtain the optimal embedding dimension of~$O(\eps^{-2}\log n)$:
\begin{definition}
    \label{def:conv_hull}
    Given $X = \{x_i\}_{i = 1}^n \subset \mb{R}^d$ and $\eps > 0$, we say that a matrix $\Pi \in \mb{R}^{k \times d}$ satisfies $\eps$-convex hull distortion for $X$ if
        $\forall z \in \conv (T): \abs{\norm{\Pi z} - \norm{z}} \leq \eps$,
    where $T = \lbrb{\frac{x - y}{\norm{x - y}}: x, y \in X}$.
\end{definition}
Furthermore, \cite{nn19} show that a matrix with i.i.d.\ sub-Gaussian entries satisfies this property with high probability. We now formally describe the construction analyzed in \cite{nn19}. Given query~$q$ and $\Pi \in \mb{R}^{k\times d}$ satisfying $\eps$-convex hull distortion for $X$, they construct a terminal embedding for~$q$ by first finding $v \in \mb{R}^k$ satisfying the following constraints:
\begin{gather*}
    \norm{v - \Pi \xhat} \leq \norm{q - \xhat} \\
    \forall x \in X: \abs{\inp{v - \Pi \xhat}{\Pi \lprp{x - \xhat}} - \inp{q - \xhat}{x - \xhat}} \leq \eps \norm{q - \xhat} \norm{x - \xhat} \tag{Prog} \label{eq:pprog}
\end{gather*}
where $\xhat = \argmin_{x \in X} \norm{x - q}$; that is, the closest neighbor of $q$ in $X$. It is shown in \cite{nn19}, building upon \cite{mmmr18}, that such a point indeed exists and furthermore, the above set of constraints are convex implying the existence of polynomial time algorithms to find $v$. Given $v^*$ satisfying the above constraints, it is then shown that one can set $f(q) := z_q = (v^*, \sqrt{\norm{q - \xhat}^2 - \norm{v^* - \Pi \xhat}^2})$.

While \ref{eq:pprog} is \emph{not} the convex program we solve in our work, it is still instructive to analyze it to construct an algorithm with fast query time albeit with large memory. We do this in two steps:
\begin{enumerate}
    \item We assume access to a separating oracle for the set $K \coloneqq \{v \in \R^k: v \text{ satisfies \ref{eq:pprog}}\}$ and analyze the complexity of an optimization algorithm making queries to the oracle.
    \item We then construct a separating oracle $\mc{O}$ for \ref{eq:pprog}.
\end{enumerate}
First, observe that there exists a choice of $k = \Theta(\eps^{-2}\log n)$, as was shown in \cite{nn19}. Now let $\mc{O}$ be a separating oracle for the convex set $K = \{v \in \mb{R}^k: v \text{ satisfies \ref{eq:pprog}}\}$; that is, given $v \in \mb{R}^k$, $\mc{O}$ either reports that $v \in K$ or outputs $u \neq 0$ such that for all $y \in K$, we have $\inp{y - v}{u} \geq 0$. Then standard results on the Ellipsoid Algorithm \cite{ny83,bertsekas} imply that one can find a feasible $v$ by making $O^*(1)$ calls to $\mc{O}$ with each call incurring additional time $O^* (1)$. Hence, we may restrict ourselves to the task of designing a fast separating oracle for $K$.

To implement a fast separation oracle, first note that the first constraint in \ref{eq:pprog} can be checked explicitly in time $O(d)$ and if the constraint is violated, the oracle may return $\Pi \hat{x} - v$. If the first constraint is satisfied, consider the following re-arrangement of the second constraint:
\begin{equation*}
    \forall x \in X: \abs*{\inp*{\frac{v - \Pi \xhat}{\norm{q - \xhat}}}{\Pi \lprp{\frac{x - \xhat}{\norm{x - \xhat}}}} - \inp*{\frac{q - \xhat}{\norm{q - \xhat}}}{\frac{x - \xhat}{\norm{x - \xhat}}}} \leq \eps.
\end{equation*}
By denoting $\tv_{y} = \frac{\lprp{q - y,-(v - \Pi y)}}{\norm{\lprp{q - y,-(v - \Pi y)}}}$ and $\tv_{y,z} = \frac{\lprp{y - z,\Pi (y - z)}}{\norm{\lprp{y - z, \Pi (y - z)}}}$, we see that the above constraint essentially tries to enforce that $\inp{\tv_{\hat{x}}}{\tv_{x, \hat{x}}}$ is close to $0$. Note that from the constraints that we have already checked previously, we have that $\norm{\tv_{x, \hat{x}}}, \norm{\tv_{\hat{x}}} = 1$ and hence, the condition  $\abs{\inp{\tv_{\hat{x}}}{\tv_{x, \hat{x}}}} \approx 0$ is equivalent to $\norm{\tv_{\hat{x}} \pm \tv_{x, \hat{x}}} \approx \sqrt{2}$. Conversely, $\abs{\inp{\tv_{\hat{x}}}{\tv_{x, \hat{x}}}} \gg C \eps$ implies $\norm{\tv_{\hat{x}} \pm \tv_{x, \hat{x}}} \leq \sqrt{2} - C\eps$ for some signing. Since, $\tv_{x,y}$ for $x,y \in X$ are independent of $q$, we may build a fast separating oracle by building $n$ nearest neighbor data structures, one for each $x \in X$, with the point set $\{\pm \tv_{y, x}\}_{y \in X}$ and at query time, constructing $\tv_{\hat{x}}$ and querying the data structure corresponding to~$\hat{x}$. Despite this approach yielding a fast separating oracle, it has three significant shortcomings:
\begin{enumerate}
    \item Computing exact nearest neighbor in high dimensions is inefficient
    \item Queries may be adaptive, violating guarantees of known randomized approximate nearest neighbor data structures
    \item The space complexity of the separating oracle is rather large.
\end{enumerate}

The first two points are easily resolved: the correctness guarantees can be straightforwardly extended to the setting where one works with an approximate nearest neighbor and adopting the framework from \cite{adaptiveds} in tandem with known reductions from Approximate Nearest Neighbor to Approximate Near Neighbor yield an adaptive approximate nearest neighbor algorithm. 

For the third point, note that in the approach just described, we construct $n$ data structures with $n$ data points each. Hence, even if each data structure can be implemented in $O(n)$ space, this still yields a data structure of at least quadratic space complexity. The rest of our discussion is dedicated to addressing this difficulty. 

We first generalize \ref{eq:pprog}, somewhat paradoxically, by adding more constraints:
\begin{gather*}
    \forall y \in X: \norm{v - \Pi y} \leq (1 + \eps) \norm{q - y}\\
    \forall x, y \in X: \abs{\inp{v - \Pi y}{\Pi \lprp{x - y}} - \inp{q - y}{x - y}} \leq \eps \norm{q - y} \norm{x - y}. \tag{Gen-Prog} \label{eq:gen_prog}
\end{gather*}
Despite these constraints (approximately) implying those in \ref{eq:pprog}, they are also implied by those from \ref{eq:pprog}. Hence, in some sense the two programs are equivalent but it will be more convenient to describe our approach as it relates to the generalized set of constraints. These constraints may be interpreted as a multi-centered characterization of the set of constraints in \ref{eq:pprog}. While \ref{eq:pprog} only has constraints corresponding to a centering of $q$ with respect to $\xhat$ and its projection, \ref{eq:gen_prog} instead requires $v$ to satisfy similar constraints irrespective of centering.

The first key observation behind the construction of our oracle is that it is not \emph{necessary} to find a point satisfying all the constraints \ref{eq:gen_prog}. It suffices to construct an oracle, $\mc{O}$, satisfying:
\begin{enumerate}
    \item If $\mc{O}$ outputs $\mrm{FAIL}$, $v$ can be extended to a terminal embedding for $q$; that is $v$ may be appended with one more element to form a valid distance-preserving embedding for $q$.
    \item Otherwise, $\mc{O}$ outputs a separating hyperplane for \ref{eq:gen_prog}.
\end{enumerate}
This is weaker than the one just constructed for \ref{eq:pprog} in two significant ways: the oracle may output $\mrm{FAIL}$ even if $v$ does \emph{not} satisfy \ref{eq:pprog}, and the expanded set of constraints allow a greater range of candidate separating hyperplanes, hence, making the work of the oracle ``easier''. A technical benefit of introducing the program \ref{eq:gen_prog} is that it provides the definition of a ``nice'' convex body for which this oracle outputs a separating hyperplane when it does not output $\mrm{FAIL}$, hence, allowing use to inherit the convergence bounds from the use of the Ellipsoid method for convex optimization.

The second key observation underlying the design of our oracle is one that allows restricting the set of relevant constraints for each input substantially. Concretely, to ensure that $v$ can be extended to a point preserving its distance to any $x \in X$, it is sufficient to satisfy the following two constraints for all $x \in X$:
\begin{gather*}
    \norm{v - \Pi y} \leq (1 + \eps) \norm{q - y} \\
    \abs{\inp{v - \Pi y}{\Pi (x - y)} - \inp{q - y}{x - y}} \leq \eps \norm{q - y} \norm{x - y}
\end{gather*}
for any $y \in X$ satisfying $\norm{q - y} = O (\norm{q - x})$. In particular, the point $y$ may be much farther from $q$ than $\xhat$ but still constitutes a good ``centering point'' for $x$. Therefore, we simply need to ensure that our oracle checks at least one constraint involving a valid centering point for $x$ for all $x \in X$. However, at this point, several difficulties remain:
\begin{enumerate}
    \item Which constraints should we satisfy for any input query $q$?
    \item How do we build a succinct data structure quickly checking these constraints?
\end{enumerate}

These difficulties are further exacerbated by the fact that the precise set of constraints may depend on the query $q$ which may be chosen adaptively. In the next two subsections, we address these issues with the following strategy where we still make use of nearest neighbor data structures over $\{\tv_{x,y}\}_{x,y \in X}$:
\begin{enumerate}
    \item Each nearest neighbor data structure consists of points $\{\tv_{y,x}\}_{y \in S}$ for some small $S \subset X$
    \item Use a smaller number of data structures
    \item Only a small number of relevant data structures are queried when presented with query~$q$
\end{enumerate}
For each of these three choices, we will exploit recent developments in approximate near neighbor search \cite{alrw}. In \cref{ssec:fixed_scale_oracle}, we describe our approach to construct an oracle for a fixed scale setting where the distance to the nearest neighbor is known up to a polynomial factor and finally, describe the reduction to the fixed scale setting in \cref{ssec:gen_red}.

\subsection{Fixed Scale Oracle}
\label{ssec:fixed_scale_oracle}
In this subsection, we outline the construction of a suitable separation oracle for all $q$ satisfying $\wt{r} \leq \norm{q - \xhat} \leq \poly (n) \wt{r}$ for some \emph{known} $\wt{r}$. For the sake of illustration, in this subsection we only consider the case where our oracle has space complexity $O^*(nd)$ though our results yield faster oracles if more space is allowed (see \cref{thm:term_embedding}). In order to decide which points to use to construct our nearest neighbor data structures, we will make strong use of the following randomized partitioning procedures. Given $X \subset \R^d$, these data structures construct a set of subsets $\mc{S} = \{S_i \subseteq X\}_{i = 1}^m$ and a hash function $h: \R^d \to 2^{[m]}$ such that for a typical input point, $x$, sets in $h(x)$ contain points that are close to $x$ and exclude points far from $x$. The data structure is formally defined below recalling that $\spcomp (\mc{D})$ and $\tcomp (\mc{D})$ denote the space and time complexities of a data structure $\mc{D}$:
\begin{definition}
    \label{def:na_ann_part}
    We say a randomized data structure is an $(\rho_u, \rho_c)$-Approximate Partitioning (AP) if instantiated with a set $X = \{x_i\}_{i = 1}^n \subset \mb{R}^d$ and $r > 0$, produces, $\mc{D} = (h, \mc{S} = \{S_i\}_{i = 1}^m)$, with $S_i \subseteq X$, $\abs{S_i} > 0$ and $h: \mb{R}^d \to 2^{[m]}$ satisfying:
    \begin{gather*}
        \mb{P} \lbrb{\spcomp (\mc{D}) \leq O^*(n^{1 + \rho_u} d) \text{ and } \sum_{i = 1}^m \abs{S_i} \leq O^*(n^{1 + \rho_u}) } = 1 \\
        \forall x \in \mb{R}^d: \mb{P} \lbrb{\abs{h(x)} \leq O^*(n^{\rho_c}) \text{ and } \tcomp(h (x)) \leq O^*(n^{\rho_c} d)} = 1 \\
        \forall x \in X: \mb{E} \lbrb{\sum_{i = 1}^m \bm{1}\lbrb{x \in S_i}} \leq O^*(n^{\rho_u})\\
        \forall x \in \mb{R}^d: \mb{E} \lbrb{\sum_{\substack{y \in X \\ \norm{x - y} \geq 2r}} \sum_{i \in h(x)} \bm{1} \lbrb{y \in S_i}} \leq O^*(n^{\rho_c}) \\
        \forall x \in \mb{R}^d, y \in X \text{ such that } \norm{x - y} \leq r: \mb{P} \lbrb{\sum_{i \in h(x)} \bm{1} \lbrb{y \in S_i} \geq 1} \geq 0.99
    \end{gather*}
    where the probability is taken over the random decisions used to construct $\mc{D}$.
\end{definition}

The first condition in the above definition restricts the space complexity of the data structure and the sum of the number of points stored in all of the sets in $\mc{S}$ while the third condition states that each point is replicated very few times across all the sets in $\mc{S}$ in expectation. The second condition states that for any input $x$, $h$ is computable quickly and maps $x$ to not too many sets in $\mc{S}$. Finally, the last two conditions ensure that points far from $x$ are rarely in the sets $x$ maps to and that points close to $x$ are likely to be found in these sets. Essentially, these data structures will allow us to partition our space such that most points in each partition are close to each other and hence, the constraints corresponding set from \ref{eq:gen_prog} can be checked with a small number of centering points while the few far away points may be checked explicitly. 

Data structures satisfying the above definition have been a cornerstone of LSH-based approaches for approximate nearest neighbor search which yield state-of-the-art results and nearly optimal time-space tradeoffs \cite{him,nearoptHashHighDim,alrw}. The conditions with probability $1$ can be ensured by truncating the construction of the data structure if its space complexity grows too large or by truncating the execution of $h$ on $x$ if its runtime exceeds a certain threshold. Finally, the events with probability $0.99$ can be boosted from arbitrary constant probability by repetition. In this subsection's setting of almost linear memory, any $(0, c)$-AP data structure suffices for any $c < 1$, and \cite{alrw} show that an $(0, 7/16)$-AP data structure exists. 

For the sake of exposition, assume that the second, third and fourth conditions in the above definition hold deterministically; that is, assume that each data point is only replicated in $O^*(n^{\ru})$ many of the $S_i$, for each $x$ only $O^*(n^{\rc})$ many points farther away from $x$ than $2r$ are present in the sets mapped to by $h$ and each point in the dataset within $r$ of $x$ is present in one of the sets that $x$ maps to. In our formal proof, we show that any issues caused due to the randomness are easily addressed. We now instantiate $O \lprp{\log (n) / \gamma}$ independent $(0, 7/16)$-AP data structures, $\mc{D}_i = (h_i, \mc{S}_i = \{S_j\}_{j = 1}^{m_i})$ for the point set $X$ with $r_i = (1 + \gamma)^i \wt{r}$ and $\gamma \approx 1 / \log^3 n$. Note that this only results in $O^* (1)$ data structures in total. Now, for each $i$ and $S \in \mc{S}_i$, we pick $l \approx \log n$ random points from $S$, $\mc{Z}_{i,S} = \lbrb{z_j}_{j = 1}^l$, and instantiate nearest neighbor data structures for the points $\lbrb{\pm \tv_{x, z_j}}_{x \in S}$ and assign any point $y \in S$ within $4r_i$ of $z_j$ to $z_j$ for $S$. Note that these assignments are only for the set $S$; for a distinct set $S^\prime$ such that $y \in S^\prime$, $y$ need not be assigned to any point. Intuitively, this assignment will be used to determine which of the nearest-neighbor data structures are queried on an input query $q$ for which a terminal embedding is required to be computed as we only instantiate these data structures for some points in $S_i$ -- in fact, we will choose the points in $\mc{Z}_{i, S}$ to which the nearest neighbor of $q$ in $X$ has been assigned. Each of these data structures only stores $O^*(n)$ points in total and the existence of near neighbor data structures with space complexity $O^*(dn)$ (and query time $O^*(dn^{1 - \wt{c} \eps^2})$) \cite{alrw} complete the bound on the space complexity of a data structure at a single scale. Since we only instantiate $O^* (1)$ many such data structures, this completes the bound on the space complexity of the data structure. By choosing the points randomly in this way, one can show that for any $x$, the sum total of the number of unassigned points over all the sets in $h_i(x)$ (including potentially duplicated points) is at most $O^*(n^{\rc})$ by partitioning $X$ into a set with points close to $x$ and those far away and analyzing their respective probabilities of being assigned in each set. Furthermore, note the total number of points stored in all the $h_i(x)$ is trivially at most $O^*(n)$.

At query time, suppose we are given query $q$, its nearest neighbor in $X$, $\xhat$, and a candidate~$v$ and we wish to check whether $v$ can be extended to a valid terminal embedding for $q$. While having access to the exact nearest neighbor is an optimistic assumption, extending the argument to use an approximate nearest neighbor is straightforward. We query the data structure as follows, we query each data structure $\mc{D}_i$ with $\xhat$ and for each $S \in h_i (\xhat)$, we check whether each unassigned point, $x$, satisfies $\abs{\inp{\tv_{x, \xhat}}{\tv_{\xhat}}} \approx 0$. Then, for each point in $z_j \in \mc{Z}_{i, S}$, we query its nearest neighbor data structure with $\wt{v}_{z_j}$. If any of the data structures report a point significantly violating the inner product condition, we return that data point as a violator to our set of constraints.

We now prove the correctness and bound the runtime of the oracle. We start by bounding the runtime of this procedure. For a single scale, $r_i$, we query at most $O^*(n^{\rc})$ unassigned points and their contribution to the runtime is correspondingly bounded. Intuitively, this is true because $h(x)$ contains at most $O^*(n^{\rc})$ points far from $x$ and if there are more than $O^*(n^{\rc})$ points close to $x$, they tend to be assigned. For assigned points, there are at most $O^*(n)$ many of them (repetitions included) spread between $O^*(n^{\rc})$ many nearest neighbor data structures (as $\abs{h_i (\xhat)} \leq O^*(n^{\rc})$) and each of these data structures has query time $O^*(l^{1 - \wt{c} \eps^2})$ where $l$ is the number of points assigned to the data structure. A simple convexity argument shows that the time taken to query all of these is at most $O^*(n^{1 - \wt{c} \eps^2})$. This bounds the query time of the procedure. 

To establish its correctness, observe that when the algorithm outputs a hyperplane, the correctness is trivial. Suppose now that the oracle outputs $\mrm{FAIL}$. Note that any point, $x \in X$, \emph{very} far away from $\xhat$ may be safely ignored (say, those $\poly(n) \wt{r}$ away from $\xhat$) as an embedding that preserves distance to $\xhat$ also preserves distance to $x$ by the triangle inequality. We now show that for any other $x$, it satisfies $\abs{\inp{\tv_{x,y}}{\tv_y}} \approx 0$ for some $y$ with $\norm{q - y} \leq C\norm{q - x}$. Any such~$x$ must satisfy $\norm{x - \xhat}\leq 2 \norm{q - x}$ by the triangle inequality and the fact that $\xhat$ is the nearest neighbor. As a consequence, there exists~$i$ such that $\norm{x - \xhat} \leq r_i$ and $\norm{x - q} \geq 0.5 r_{i - 1}$. For this~$i$, there exists $S \in h_i(\xhat)$ containing $x$. In the case that $x$ is not assigned, we check $\abs{\inp{\tv_{x,\xhat}}{\tv_{\xhat}}} \approx 0$ and correctness is trivial. In case $x$ is assigned, it is assigned to $y$ with $\norm{y - x} \leq 4r_i$ and we have:
\begin{equation*}
    \norm{y - q} \leq \norm{x - q} + \norm{y - x} \leq \norm{x - q} + 4r_i \leq 10 \norm{x - q}
\end{equation*}
where the final inequality follows from the fact that $\norm{x - q} \geq 0.5 r_{i-1}$. This proves that the inner product condition for $x$ is satisfied with respect to $y$ with $\norm{q - y} \leq 10 \norm{x - q}$. This concludes the proof in the second case where the oracle outputs $\mrm{FAIL}$. 

The argument outlined in the last two paragraphs concludes the construction of our weaker oracle when an estimate of $\norm{q - \xhat}$ is known in advance. The crucial property provided by the existence of the $(\ru, \rc)$-AP procedure is that there are at most $O^*(n)$ many points used to construct the near neighbor data structures for the points $\lbrb{\wt{v}_{x,y}}$ (as opposed to $n^2$ for the previous construction). This crucially constrains us to having either a large number of near neighbor data structures with few points or a small number with a large number of points but not both. However, the precise choice of how the algorithm trades off these two competing factors is dependent on the set of data points and the scale being considered. The savings in query time follow from the fact that at most $O^*(n^{c})$ of these data structures are consulted for any query for some $c < 1$.

\subsection{Reduction to Fixed Scale Oracle}
\label{ssec:gen_red}

We reduce the general case to the fixed scale setting from the previous subsection.\footnote{Note that \cref{ssec:fixed_scale_oracle} already allows the construction of a separating oracle if one is willing to incur runtimes depending \emph{logarithmically} in the aspect ratio of the dataset. However, the results in the paper allow for a construction with \emph{no} such dependence.} To define our reduction, we will need a data structure we will refer to as a Partition Tree that has previously played a crucial part in reductions from the Approximate Nearest Neighbor to Approximate Near Neighbor \cite{him}. We show that the same data structure also allows us to reduce the oracle problem from the general case to the fixed scale setting. Describing the data structure and our reduction, requires some definitions from \cite{him}:
\RestateInit{\restategpcc}
\begin{restatable}{definition}{gpcc}
    \label{def:gp_cc}\RestateRemark{\restategpcc}
    Let $X = \{x_i\}_{i = 1}^n \subset \mb{R}^d$ and $r > 0$. We will use $\gr (X, r)$ to denote the graph with nodes indexed by $x_i$ and an edge between $x_i$ and $x_j$ if $\norm{x_i - x_j} \leq r$. The connected components of this graph will be denoted by $\cc (X, r)$; that is, $\cc (X, r) = \{C_j\}_{j = 1}^m$ is a partitioning of $X$ with $x \in C_j$ if and only if $\norm{x - y} \leq r$ for some $y \in C_j \setminus \{x\}$.
\end{restatable}
Note from the above definition that $\cc (X, r)$ results in increasingly fine partitions of~$X$ as~$r$ decreases. This notion is made precise in the following straightforward definition:
\RestateInit{\restatepartref}
\begin{restatable}{definition}{partref}
    \label{def:part_ref}\RestateRemark{\restatepartref}
    For a data set $X = \{x_i\}_{i = 1}^n \subset \mb{R}^d$, we say that a partition $\mc{C}$ \emph{refines} a partition~$\mc{C}^\prime$ if for all $C \in \mc{C}$,  $C \subseteq C^\prime$ for some $C^\prime$ in $\mc{C}^\prime$. This will be denoted by $\mc{C}^\prime \sqsubseteq \mc{C}$. 
\end{restatable}

Next, define $\rmed(X)$ as:
\begin{equation*}
    \rmed(X) = \min \{r > 0: \exists C \in \cc (X, r) \text{ with } \abs{C} \geq n / 2\}.
\end{equation*}

We are now ready to define a Partition Tree:
\RestateInit{\restateparttree}
\begin{restatable}{definition}{parttree}
    \label{def:part_tree}\RestateRemark{\restateparttree}
    Given $X = \{x_i\}_{i = 1}^n \subset \mb{R}^d$, a Partition Tree of $X$ is a tree, $\mc{T}$, whose nodes are labeled by $\lprp{Z, \lbrb{\mc{T}_C}_{C \in \clow}, \mc{T}_{\mathrm{rep}}, \clow, \chigh, \crep, \rhat}$ where $Z,\crep \subset X$, $\{\mc{T}_{C}\}_{C \in \clow} \cup \{\mc{T}_{\mrm{rep}}\}$ represent its children, $\clow, \chigh$ are partitions of $Z$ and $\rhat > 0$ satisfying the following conditions:
    \begin{gather*}
        \cc (Z, 1000n^2\rhat) \sqsubseteq \chigh \sqsubseteq \cc (Z, \rhat) \sqsubseteq \cc (Z, \rmed) \sqsubseteq \cc \lprp{Z, \frac{\rhat}{10n}} \sqsubseteq \clow \sqsubseteq \cc \lprp{Z, \frac{\rhat}{1000n^3}} \\
        \forall C \in \chigh: \exists!\, z \in \crep \text{ with } z \in C.
    \end{gather*}
    For the sake of notational simplicity, we will use $\mc{T}^\prime \in \mc{T}$ both to refer to a node in the tree as well as the subtree rooted at that node and $\Size(\mc{T}^\prime)$ to refer to the sum of the number of points stored in the subtree $\mc{T}^\prime$. The above condition implies $\Size(\mc{T}) \leq O(n\log n)$ \cite{him}.
\end{restatable}

While deterministic data structures with the same near linear runtime are also known \cite{hp}, we include, for the sake of completeness, a simple probabilistic algorithm to compute a partition tree with probability $1 - \delta$ in time $O(nd\log^2 n / \delta)$. Having defined the data structure, we now describe how it may be used to define our reduction. 

At a high level, the reduction traverses the data structure starting at the root and at each step either terminating at the node currently being explored or proceeds to one of its children. By the definition of the data structure, the number of points in the node currently being explored drops by at least a factor of $2$ in each step. Therefore, the procedure explores at most $\ceil{\log n}$ nodes. For any node $\mc{T}^\prime \in \mc{T}$, with associated point set $Z$, currently being traversed, we will aim to enforce the following two conditions:
\begin{enumerate}
    \item An approximate nearest neighbor of $q$ in $Z$ is also an approximate nearest neighbor in $X$
    \item A terminal embedding of $q$ for $Z$ is also valid for $X$.
\end{enumerate}
For simplicity, we assume the existence of near neighbor data structures; that is, data structures which when instantiated with $(X, r)$ and given a query $q$, output a candidate $y \in X$ such that~$\norm{y - q} \leq r$ if $\min_{x \in X} \norm{x - q} \leq r$ (we refrain from assuming access to nearest neighbor data structures here as this reduction will also be used to construct our nearest neighbor data structures). For each node $\mc{T}^\prime = \lprp{Z, \lbrb{\mc{T}_C}_{C \in \clow}, \mc{T}_{\mrm{rep}}, \clow, \chigh, \crep, \rhat}$, we first decide two thresholds $\rlow = \rhat / \poly(n)$ and $\rhigh = \poly (n)\rhat$ and interpolate the range with roughly $m \approx (\log \rhigh / \rlow) / \gamma$ many near neighbor data structures, $\{\mc{D}_i\}_{i = 0}^m$, with $r_i = (1 + \gamma)^i \rlow$ for the point set $Z$. Note, we set $\gamma \approx 1 / \log^3 n$ which implies that we instantiate at most $O^* (1)$ many near neighbor data structures.

At query time, suppose we are at node $\mc{T}^\prime = \lprp{Z, \lbrb{\mc{T}_C}_{C \in \clow}, \mc{T}_{\mrm{rep}}, \clow, \chigh, \crep, \rhat}$ with associated near neighbor data structures, $\mc{D}_i$. We query each of the data structures $\mc{D}_i$ with $q$ and we have three possible cases:
\begin{enumerate}
    \item The nearest neighbor to $q$ is within a distance of $\rlow$
    \item The nearest neighbor to $q$ is beyond $ \rhigh$ of it
    \item The nearest neighbor to $q$ is between $\rlow$ and $\rhigh$ of $q$
\end{enumerate}
The first case occurs when $\mc{D}_0$ returns a candidate nearest neighbor, the second when none of the near neighbor data structures return a candidate and the third when, $\mc{D}_i$ succeeds but $\mc{D}_{i - 1}$ fails for some $i$. If the third case occurs, the reduction is complete. If the second case occurs, let $x \in C \in \chigh$ and $\wt{x} \in \crep \cap C$. We have by the triangle inequality and \cref{def:gp_cc,def:part_tree} $\norm{x - \wt{x}} \leq \poly (n) \rmed \ll \rhigh$ and hence $\norm{q - \wt{x}} \approx \norm{q - x}$ and hence we recurse in $\crep$ still satisfying the two conditions stated above. For the first case, let $\wt{x} \in C \in \clow$ such that $\norm{q - \wt{x}} \leq \rlow$. From \cref{def:gp_cc,def:part_tree}, any $x \notin C$ satisfies $\norm{x - \wt{x}} \geq \rmed / \poly (n) \gg \rlow$ and hence, both conditions are again maintained as the nearest neighbor of $q$ in $Z$ is in $C$ and a terminal embedding for $q$ in $C$ is a terminal embedding in $Z$ by the triangle inequality.

\subsection{Runtime Improvements and Adaptivity}
\label{ssec:runtime_adapt}

We conclude with other technical considerations glossed over in the previous discussion. As remarked before, we assumed access to exact nearest and near neighbor data structures which perform correctly even when faced with adaptive queries. While the arguments outlined in the previous two subsections extend straightforwardly to the setting where approximate nearest neighbors are used, the guarantees provided by previous approaches to the approximate near neighbors are not robust in the face of adaptivity. 
\begin{definition}
    \label{def:na_ann}
    For $\rho_u, \rho_c \geq 0$ and $c > 1$, we say a randomized data structure is an $(\rho_u, \rho_c, c)$-Approximate Near Neighbor (ANN) data structure for Approximate Near Neighbor if instantiated with a set of data points $X = \{x_i\}_{i = 1}^n$ and $r > 0$ constructs, $\mc{D}$, satisfying:
    \begin{gather*}
        \mb{P} \lbrb{\spcomp (\mc{D}) \leq O^*(d n^{1 + \rho_u})} = 1 \\
        \forall q \in \mb{R}^d: \mb{P} \lbrb{\tcomp(\mc{D} (q)) \leq O^*(dn^{\rho_c})} = 1 \\
        \forall q \in \mb{R}^d: \mb{P} \lbrb{\mc{D} (q) \text{ returns } x \in X \text{ with } \norm{q - x} \leq cr \text{ if } \exists y \in X \text{ with } \norm{q - y} \leq r} \geq 0.99
    \end{gather*}

    where the probability is taken over \emph{both} the random decisions used to construct $\mc{D}$ and those used by $\mc{D}$ to answer the query $q$. Additionally, $q$ is assumed to be independent of $\mc{D}$. 
\end{definition}
A simple repetition argument can then be used to devise adaptive near neighbor data structure with similar guarantees \cite{adaptiveds}. Used in tandem with the reduction outlined in the previous subsection yields the adaptive nearest neighbor data structures used in \cref{ssec:fixed_scale_oracle}.

Finally, the argument outlined previously enabled computing terminal embeddings in time $O^*(dn^{\rho})$ for some $\rho < 1$ which while being sublinear in $n$ is suboptimal in its interaction with the dimension. This is in contrast to state-of-the-art approaches to nearest neighbor search which yield runtimes scaling as $O^*(dn + n^{\rho})$. However, all these approaches deduce this result by first projecting onto a lower dimensional space using a Johnson-Lindenstrauss (JL) projection and building the data structure in the lower dimensional space. This is not feasible in our setting as JL projections are not robust to adaptivity and we require an alternate strategy. 

To improve our runtime, suppose $\Pi^\prime \in \R^d \to \R^k$ with $k = O^* (1)$ satisfies:
\begin{equation}
    \label{eq:pip_ip_pres}
    \forall x,y,z \in X \cup \{q\}: \abs{\inp{\Pi^\prime (x - z)}{\Pi^\prime (y - z)} - \inp{x - z}{y - z}} \leq o^*(1) \cdot \norm{x - z} \norm{y - z}.
\end{equation}
Then, to construct a terminal embedding, we may construct the vectors $\tv_{x,y}$ and $\tv_x$ by using the vectors $(\Pi^\prime x, \Pi^\prime y, \Pi^\prime q)$ instead of the vectors in the high dimensional space. Assuming the projections for $x \in X$ are pre-computed, we may do this projection in $O^* (d)$ time and the rest of the procedure to compute terminal embeddings takes time $O^*(n^{\rho})$. When $q$ is independent of the data structure, a standard JL-sketch satisfies \cref{eq:pip_ip_pres} with high probability but this is not true when $q$ depends on $\mc{D}$. Thankfully, we show that if one draws $O^* (d)$ many $JL$-sketches, at least $95\%$ satisfy \cref{eq:pip_ip_pres} for any query $q$ with high probability. Note that the order of the quantifiers in the statement make the proof more challenging than previous work using such ideas \cite{adaptiveds} and requires a careful gridding argument which we carry out in \cref{sec:med_jl}.

\section{Terminal Embeddings}
\label{sec:term}

In this section, prove the main theorem of the paper. Note that in the following theorem the query $q$ can be chosen with full knowledge of the data structure. That is, conditioned on the successful creation of the data structure, the randomized construction of $z_q$ is only over the random decisions taken at \emph{query-time} and not during the creation of the data structure. The main theorem of the paper is stated below:
\begin{theorem}
    \label{thm:term_embedding}
    Let $\eps \in (0, 1)$, $\rone, \rtwo, \rthree, \rfour, \rrep > 0$. Then, there is a randomized procedure which when instantiated with a dataset $X = \{x_i\}_{i = 1}^n \subset \mb{R}^d$, a $(\rthree, \rfour)$-Approximate Partitioning data structure, a $(\rone, \rtwo, (1 + \eps^\dagger))$-Approximate Near Neighbor data structure for $\eps^\dagger = c\eps$ for some small enough $c > 0$ and parameter $\rrep$ constructs a data structure, $(\mc{D}, \Pi \in \mb{R}^{k \times d})$, satisfying the following guarantees:
    \begin{enumerate}
        \item $\Pi$ has $\eps$-convex hull distortion for $X$
        \item Given $q \in \mb{R}^d$, $\mc{D}$ produces with probability at least $1 - 1 / \poly(n)$ over the randomness of $\mc{D}$ at \emph{query} time, a vector $z_q \in \mb{R}^{k + 1}$  such that:
        \begin{equation*}
            \forall x \in X: (1 - O(\eps)) \norm{q - x} \leq \norm{z_q - (\Pi x, 0)} \leq (1 + O(\eps)) \norm{q - x}
        \end{equation*}
        \item On any $q \in \mb{R}^d$, the runtime of $\mc{D}$ is $O^*(d + n^{\rtwo} + n^{\rfour} + n^{\rfour + (1 + \rthree - \rfour - \rrep) \rtwo})$
        \item The space complexity of $\mc{D}$ is $O^*(dn^{\rrep + (1 + \rone)} + dn^{\rthree + (1 + \rone)})$
    \end{enumerate}
    with probability at least $1 - 1 / \poly(n)$ over the randomness during the \emph{instantiation} of $\mc{D}$.
\end{theorem}

Before we proceed, we will instantiate the above theorem in three specific cases. We note that the state-of-the-art algorithms for approximate nearest neighbor search are implemented in terms of approximate partitioning schemes as defined here \cite{alrw} who show for any $(\ru, \rc)$ satisfying:
\begin{equation*}
    c^2 \sqrt{\rc} + (c^2 - 1) \sqrt{\ru} = \sqrt{2c^2 - 1},
\end{equation*}
there exist both a $(\ru, \rc, c)$-Approximate Near Neighbor data structure and specifically, a $(\ru, \rc)$-Approximate Partitioning scheme when $c = 2$. By instantiating an Approximate Partitioning data structure (by setting $c = 2$) and applying the results of \cref{thm:term_embedding}, we get a data structure to compute $\eps$-terminal embeddings for (possibly different) universal $\wt{c}, \wt{C} > 0$:
\begin{itemize}
    \item Query time $O^*(d + n^{1 - \wt{c}\eps^2})$ and space complexity $O^*(nd)$ with $\rrep = \rone = \rthree = 0$
    \item Query time $O^*(d + n^{1 - \wt{c} \eps})$ and space complexity $O^*(n^{2}d)$ with $\rrep = \rthree = 0, \rone = 1$
    \item Query time $O^*(d)$ and space complexity $dn^{\wt{C} / \eps^2}$ with $\rfour = \rtwo = 0$.
\end{itemize}
The previous three results capture a range of potential time space tradeoffs for data structures computing terminal embeddings. We now move on to the proof of \cref{thm:term_embedding}. As discussed in \cref{sec:techniques}, our algorithm operates by constructing a weak separating oracle for a convex program. In \cref{ssec:gen_char}, we define the convex program for which we define our oracle, \cref{ssec:fixed_scale_terminal} defines our weak oracle for a fixed scale and in \cref{ssec:multi_scale_term} we describe the data structure which enables the reduction of the general case to the fixed scale scenario and establish that it provides an appropriate separation oracle for the convex program introduced in \cref{ssec:gen_char}. 

A key data structure that will be extensively utilized and will, henceforth, be referred to as a $(\ru, \rc, c)$-Adaptive Approximate Nearest Neighbor (AANN) data structure is described in the following theorem. It supplies a data structure supporting several important functions. It enables approximate nearest neighbor queries on the dataset \emph{even} for queries which are adaptively chosen based on the instantiation of the data structure (but not on the fresh randomness drawn at query time) and furthermore, constructs a partition tree, $\mc{T}$, with the property that for any query $q$, it suffices to construct a valid terminal embedding for the data points in node in the partition tree that the data structure returns, $\tres$. The guarantees for all other points in the dataset are guaranteed as a consequence of the first claim of the theorem. The key use of this data structure is that it now effectively suffices to construct a valid terminal embedding for the node in the partition tree that is returned in response to the query. Furthermore, the distance bounds on the distance of $q$ to its approximate nearest neighbor~$\hat{x}$, allows us to use a small number of \emph{fixed-scale} data structures for each node thus directly reducing the multi-scale setting to the fixed-scale setting. The proof of the theorem itself will be deferred to \cref{sec:aann}.
\RestateInit{\restateaannmain}
\begin{restatable}{theorem}{aannmain}
    \label{thm:aann_main}\RestateRemark{\restateaannmain}
    Let $c > 1$ and $\ru, \rc > 0$. Then, there is a randomized procedure which when instantiated with a dataset $X = \{x_i\}_{i = 1}^n \subset \mb{R}^d$ and a $(\ru,\rc, c)$-Approximate Near Neighbor data structure (\cref{def:na_ann}) produces a data structure, $(\mc{D}, \mc{T})$, satisfying:
    \begin{enumerate}
        \item Given any $q \in \mb{R}^d$, $\mc{D}$ produces $(\xhat \in X, \tres \in \mc{T})$ satisfying:
        \begin{enumerate}
            \item $\norm{q - \xhat} \leq \min_{x \in X} (1 + o^*(1))c \norm{q - x}$
            \item $\xhat \in \tres$
            \item Furthermore, let $\mc{Y} = \{y_i\}_{i = 1}^n \subset \mb{R}^k$ satisfying for some $\eps^\dagger \in \lprp{\frac{1}{\sqrt{d}}, 1}$:
                \begin{equation*}
                    \forall i, j \in [n]: (1 - \eps^\dagger) \norm{x_i - x_j} \leq \norm{y_i - y_j} \leq (1 + \eps^\dagger) \norm{x_i - x_j}
                \end{equation*}
                and for $\tres = \lprp{Z, \lbrb{\mc{T}_C}_{C \in \clow}, \mc{T}_{\mathrm{rep}}, \clow, \chigh, C_{\mathrm{rep}}, \rhat}$, let $y \in \mb{R}^k$ satisfy for $\eps^\ddagger \in [\eps^\dagger, 1)$: 
                \begin{equation*}
                    \forall x_i \in Z: (1 - \eps^\ddagger) \norm{q - x_i} \leq \norm{y - y_i} \leq (1 + \eps^\ddagger) \norm{q - x_i}.
                \end{equation*}
                Then:
                \begin{equation*}
                    \forall x_i \in X: \lprp{1 - \lprp{1 + o^*(1)}\eps^\ddagger} \norm{q - x_i} \leq \norm{y - y_i} \leq \lprp{1 + \lprp{1 + o^*(1)} \eps^\ddagger} \norm{q - x_i}.
                \end{equation*}
                and if $\abs{Z} > 1$:
                \begin{equation*}
                    \Omega \lprp{\frac{1}{(nd)^{10}} \rhat} \leq \norm{q - \xhat} \leq O((nd)^{10} \rhat)
                \end{equation*}
        \end{enumerate}
        with probability at least $1 - 1 / \poly(n)$
        \item $\mc{T}$ is a valid Partition Tree of $X$ (\cref{def:part_tree})
        \item The space complexity of $\mc{D}$ is $O^*(dn^{1 + \ru} \log 1 / \delta)$ 
        \item The runtime of $\mc{D}$ on any $q \in \mb{R}^d$ is at most $O^*(d + n^{\rc})$
    \end{enumerate}
    with probability $1 - \delta$.
\end{restatable} 

\subsection{Generalized Characterization of Terminal Embeddings}
\label{ssec:gen_char}

To start, we recall a key lemma from \cite{nn19}:

\begin{lemma}
    \label{lem:conv_hull}
    Let $X = \{x_i\}_{i = 1}^n \subset \mb{R}^d$, $\eps \in \lprp{\frac{1}{\sqrt{n}}, 1}$ and $T = \lbrb{\frac{x - y}{\norm{x - y}}: x \neq y \in X} \bigcup \{0\}$. Then for $\Pi \in \mb{R}^{k \times d}$ with $k = \Omega \lprp{\frac{\log n + \log 1 / \delta}{\eps^2}}$ with $\Pi_{i,j} \thicksim \mc{N}(0, 1/k)$,
    \begin{equation*}
        \Pr(\Pi\text{ satisfies }\eps\text{-convex hull distortion for } T) \ge 1-\delta .
    \end{equation*}
\end{lemma}

Finally, the convex program for which we will construct our oracle is defined in the following generalization of \cite{nn19,mmmr18}. As remarked in \cref{sec:techniques}, we generalize the convex program to an expanded set of constraints but crucially do not attempt to satisfy \emph{all} of them in our algorithm. For the convergence guarantees for our weak oracle to hold (\cref{lem:weak_ellipsoid}), we first need to show the feasibility of the convex program. Before, we proceed we require the following simple lemma which is proved in \cref{ssec:misc_tech}.

\RestateInit{\restateippres}
\begin{restatable}{lemma}{ippres}
    \label{lem:ip_pres}\RestateRemark{\restateippres}
    Let $X = \lbrb{x_i}_{i = 1}^n$, $0 < \eps < 1$ and $T = \lbrb{\frac{x - y}{\norm{x - y}}: x \neq y \in X} \bigcup \{0\}$. Furthermore, suppose $\Pi \in \mb{R}^{k \times d}$ has $\eps$-convex hull distortion for $X$. Then, we have:
    \begin{equation*}
        \forall x, y \in \conv (T): \abs*{\inp{\Pi x}{\Pi y} - \inp{x}{y}} \leq 6\eps.
    \end{equation*}
\end{restatable}

We now establish feasibility in the following lemma.

\begin{lemma}
    \label{lem:multi_point_nn}
    Suppose $X = \{x_i\}_{i = 1}^n \subset \mb{R}^d$ and $\Pi \in \mb{R}^{k \times d}$ satisfies $\eps$-convex hull distortion for $X$. Then, for any $q \in \mb{R}^d$, there exists $z \in \mb{R}^k$ such that:
    \begin{gather*}
        \forall x, y \in X: \abs{\inp{z - \Pi x}{\Pi (y - x)} - \inp{q - x}{y - x}} \leq 15\eps \norm{q - x} \norm{y - x} \\
        \forall x \in X: \norm{z - \Pi x} \leq (1 + 8\eps) \norm{q - x}.
    \end{gather*}
\end{lemma}
\begin{proof}
    Let $x^* = \argmin_{x \in X} \norm{q - x}$. As in \cite{nn19,mmmr18}, we consider a bilinear game where $\lambda$ is essentially selects which constraint is worst-violated and $w$ is a candidate terminal embedding:
    \begin{multline*}
        \max_{\norm{\lambda}_1 \leq 1} \min_{\norm{w} \leq \norm{q - x^*}} \sum_{x, y \in X} \lambda_{x,y} \lprp{\frac{\inp{w}{\Pi (y - x)} - \inp{q - x^*}{y - x}}{\norm{q - x^*} \norm{y - x}}} \\
        = \min_{\norm{w} \leq \norm{q - x^*}} \max_{\norm{\lambda}_1 \leq 1} \sum_{x,y \in X}\lambda_{x,y} \lprp{\frac{\inp{w}{\Pi (y - x)} - \inp{q - x^*}{y - x}}{\norm{q - x^*} \norm{y - x}}}
    \end{multline*}
    where the exchange of the min and max follows from von Neumann's minimax theorem. Considering the first formulation, let $\norm{\lambda}_1 \leq 1$. We have for such a $\lambda$:
    \begin{align*}
        &\min_{\norm{w} \leq \norm{q - x^*}} \sum_{x,y \in X} \lambda_{x,y} \lprp{\frac{\inp{w}{\Pi (y - x)} - \inp{q - x^*}{y - x}}{\norm{q - x^*} \norm{y - x}}} \\
        &=  \min_{\norm{w} \leq \norm{q - x^*}} \lprp{\inp*{\frac{w}{\norm{q - x^*}}}{\Pi \lprp{\sum_{x, y \in X} \lambda_{x,y} \cdot \frac{y - x}{\norm{y - x}}}} - \inp*{\frac{q - x^*}{\norm{q - x^*}}}{\lprp{\sum_{x,y \in X} \lambda_{x,y} \cdot \frac{y - x}{\norm{y - x}}}}} \\
        &\leq \min_{\norm{\wt{w}} \leq 1} \inp*{\wt{w}}{\Pi \lprp{\sum_{x,y \in X} \lambda_{x,y} \cdot \frac{y - x}{\norm{y - x}}}} + \norm*{\sum_{x,y \in X} \lambda_{x,y} \frac{y - x}{\norm{y - x}}} \\
        &= -\norm*{\Pi \lprp{\sum_{x,y \in X} \lambda_{x,y} \cdot \frac{y - x}{\norm{y - x}}}} + \norm*{\sum_{x,y \in X} \lambda_{x,y} \cdot \frac{y - x}{\norm{y - x}}} \leq \eps
    \end{align*}
    where the first inequality follows from Cauchy-Schwarz and the final inequality follows from the fact that $\Pi$ has $\eps$-convex hull distortion. From the previous result, we may conclude that there exists $w \in \mb{R}^k$ satisfying:
    \begin{gather*}
        \norm{w} \leq \norm{q - x^*} \\
        \forall x, y \in X: \abs{\inp{w}{\Pi (y - x)} - \inp{q - x^*}{y - x}} \leq \eps \norm{q - x^*} \norm{y - x}.
    \end{gather*}
    Now consider the vector $z = w + \Pi x^*$. For this $z$, we have for any $x,y \in X$:
    \begin{align*}
        &\abs{\inp{z - \Pi x}{\Pi (y - x)} - \inp{q - x}{y - x}} \\
        &\quad \leq \abs{\inp{z - \Pi x^*}{\Pi (y - x)} - \inp{q - x^*}{y - x}} + \abs{\inp{\Pi (x^* - x)}{\Pi (y - x)} - \inp{x - x^*}{y - x}} \\
        &\quad \leq \eps \norm{q - x^*}\norm{y - x} + \norm{x^* - x}\norm{y - x} \abs*{\inp*{\Pi {\frac{x^* - x}{\norm{x^* - x}}}}{\Pi {\frac{y - x}{\norm{y - x}}}} - \inp*{\frac{x - x^*}{\norm{x - x^*}}}{\frac{y - x}{\norm{y - x}}}} \\
        &\quad \leq \eps \norm{q - x^*}\norm{y - x} + 6\eps \norm{x^* - x} \norm{y - x} \\
        &\quad = \eps \norm{y - x} \lprp{\norm{q - x^*} + 6 \norm{x^* - x}} \\
        &\quad \leq \eps \norm{y - x} \lprp{\norm{q - x^*} + 6 (\norm{x^* - q} + \norm{q - x})} \\
        &\quad \leq \eps \norm{y - x} \lprp{\norm{q - x^*} + 12 \norm{q - x})} \\
        &\quad \leq 15 \eps \norm{y - x} \norm{q - x}
    \end{align*}
    where the second inequality is due to the condition on $z$, the third inequality is a consequence of \cref{lem:ip_pres} and the second to last inequality follows from the fact that $\norm{q - x^*} \leq \norm{q - x}$. This establishes the first claim of the lemma. For the second claim, we have for any $x \in X$:
    \begin{align*}
        \norm{z - \Pi x}^2 - \norm{q - x}^2 &= \lprp{\norm{z - \Pi x^* + \Pi (x^* - x)}^2 - \norm{q - x^* + (x^* - x)}^2} \\
        &= \norm{z - \Pi x^*}^2 + 2\inp{z - \Pi x^*}{\Pi (x^* - x)} + \norm{\Pi (x^* - x)}^2 \\
        &\qquad - \norm{q - x^*}^2 - 2\inp{q - x^*}{x^* - x} - \norm{x^* - x}^2 \\
        &\leq 2 (\inp{z - \Pi x^*}{\Pi (x^* - x)} - \inp{q - x^*}{x^* - x}) + \norm{\Pi (x^* - x)}^2 - \norm{x^* - x}^2 \\
        &\leq 2 (\inp{z - \Pi x^*}{\Pi (x^* - x)} - \inp{q - x^*}{x^* - x}) + ((1 + \eps)^2 - 1) \norm{x^* - x}^2 \\
        &\leq 2 (\inp{z - \Pi x^*}{\Pi (x^* - x)} - \inp{q - x^*}{x^* - x}) + 3\eps \norm{x^* - x}^2 \\
        &\leq 2 \eps \norm{q - x^*} \norm{x - x^*} + 3\eps \norm{x^* - x}^2 \leq 4\eps \norm{q - x^*} \norm{q - x} + 12\eps \norm{q - x}^2 \\
        &\leq 16 \eps \norm{q - x}^2
    \end{align*}
    where the first inequality follows from the fact that $\norm{z - \Pi x^*} \leq \norm{q - x^*}$, the second inequality follows from the fact that $\Pi$ has $\eps$-convex hull distortion for $X$, the fourth from our condition on $z$ and the fifth follows from the triangle inequality and the fact that $\norm{q - x^*} \leq \norm{q - x}$. By re-arranging the above inequality and taking square roots, we get:
    \begin{equation*}
        \norm{z - \Pi x}^2 \leq (1 + 8\eps) \norm{q - x}
    \end{equation*}
    concluding the proof of the lemma.
\end{proof}

\subsection{Fixed Scale Violator Detection}
\label{ssec:fixed_scale_terminal}

Note that as a consequence of \cref{thm:aann_main}, we may assume that the function can instantiate $(\rone, \rtwo, (1 + \eps^\dagger))$-Adaptive Approximate Nearest Neighbor (AANN) data structures and $(\rthree, \rfour)$-Approximate Partitioning (AP) data structures. Our data structure for constructing an oracle at a fixed scale is constructed in \cref{alg:fixed_scale_inst} and the query procedure is outlined in \cref{alg:fixed_scale_query}. \cref{alg:fixed_scale_inst} takes as input a set of data points $X$, a projection matrix $\Pi$, a memory parameter $\rrep$ and a failure probability $\delta$. When we instantiate this data structure, the point sets used to construct it will be the subsets of points corresponding to the nodes of the Partition Tree provided by \cref{thm:aann_main}. \cref{alg:fixed_scale_query} takes as input the data structure constructed by \cref{alg:fixed_scale_inst}, the query point $q$, a candidate solution to the convex program $v$, an approximate nearest neighbor of $q$ in $X$ and another tolerance parameter whose role will become clear when incorporating this data structure into the reduction. 

To state the correctness guarantees for \cref{alg:fixed_scale_inst}, we start by introducing some notation. For $x \in X$ and $r > 0$, define the local neighborhood of $x$ as follows: $\locn(x, r) = \{y \in X: \norm{y - x} \leq 2r\}$. For $x \in X$ such that $\abs{\locn (x, r)} \geq n^{1 - \rrep}$, our success event is simply that there exists $z \in \mc{Z}$ such that $\norm{z - x} \leq 2r$ and that $\mc{D}_z$ is instantiated successfully. For $x \in X$ such that $\abs{\locn (x, r)} \leq n^{1 - \rrep}$, the success event is more complicated. Informally, we will require that most of AP data structures produce appropriate partitions for $x$ and furthermore, that each $y \in X$ with $\norm{y - x} \leq r$ is well represented in these data structures. This is formally described in the following lemma:

\begin{lemma}
    \label{lem:fixed_scale_inst}
    Given $X = \{x_i\}_{i = 1}^n \subset \mb{R}^d$, $r > 0$, $\rrep \in [0, 1]$ and $\delta \in (0, 1)$, \cref{alg:fixed_scale_inst} produces a data structure $\mc{D}$ with the following guarantees:
    \begin{enumerate}
        \item For $x \in X$ such that $\abs{\locn (x, r)} \geq n^{1 - \rrep}$, we have that there exists $z \in \mc{Z}$ such that $\norm{x - z} \leq 2r$. That is, $x$ is assigned in the first stage of the algorithm. 
        \item For $x \in X$ such that $\abs{\locn (x, r)} < n^{1 - \rrep}$, we have:
        \begin{gather*}
            \sum_{i = 1}^l \bm{1} \lbrb{ 
            \begin{gathered}
                \sum_{\substack{y \in X \\ \norm{y - x} \geq 2r}} \sum_{S \in h_i(x)} \bm{1} \lbrb{y \in S} \leq O^*(n^{\rfour}) \text{ and } \\ \sum_{\substack{y \in X \\ \norm{y - x} \leq 2r}} \sum_{S \in h_i(x)} \bm{1} \lbrb{y \in S} \leq O^*(n^{(1 - \rrep) + \rthree})
            \end{gathered}
            } \geq 0.98l \\
            \sum_{i = 1}^l \bm{1} \lbrb{\sum_{S \in h_i (x)} \abs{S \setminus A_{i, S}} \leq O^*(n^{\rfour})} \geq 0.98l \\
            \forall y \in X \text{ such that } \norm{y - x} \leq r: \sum_{i = 1}^l \bm{1} \lbrb{\exists S \in h_i (x) \text{ such that } y \in S} \geq 0.98l
        \end{gather*}
        \item All the AANN data structures instantiated in the algorithm are instantiated successfully; that is, $\{\mc{D}_{z}\}_{z \in \mc{Z}}$ and $\{\mc{D}_{i, S, w}\}_{i \in [l], S \in \mc{D}_i, w \in \mc{W}_{i,S}}$ are instantiated successfully. 
        \item Finally, the space complexity of the data structure is $O^*(dn^{1 + \rone}(n^{\rrep} + n^{\rthree}) \log^2 1 / \delta)$
    \end{enumerate}
     with probability at least $1 - \delta$.
\end{lemma}
\begin{proof}
    For the third claim, note that the total number of $\mc{D}_z$ instantiated is at most $\nrep$. Therefore, the probability that all the $\mc{D}_z$ are instantiated correctly is at least $1 - \delta / 16$. As for the $\mc{D}_{i, S, w}$, note that at most $O^*(lpn^{1 + \rthree})$ many of these are instantiated as there are $l$ AP data structures, each of which has at most $O^*(n^{1 + \rthree})$ subsets and each subset has at most $p$ AANN data structures. Therefore, again by the union bound the probability that each of these are instantiated correctly is at least $1 - \delta / 16$. Hence, the probability that all AANN data structures are instantiated correctly is at least $1 - \delta / 8$. 

    For the last claim, note that the space occupied by each of the $\mc{D}_z$ data structures is at most $O^*(n^{1 + \rone})$ and there are $\nrep$ of these. The space required to store each of the $l$ AP data structures is at most $O^*(n^{1 + \rthree})$. Finally, the space occupied by the $\mc{D}_{i, S, w}$ data structures is given by:
    \begin{equation*}
        \sum_{i = 1}^l \sum_{S \in \mc{S}_i} p \cdot O^*(\abs{S}^{1 + \rone} \log 1 / \delta) \leq O^*(n^{\rthree + (1 + \rone)} \log^2 1 / \delta)
    \end{equation*}
    by the fact that $\abs{S} \leq n$ for each $S$, $\sum_{S \in \mc{S}_i} \abs{S} \leq O^*(n^{1 + \rthree})$ for each $l$ and the convexity of the function $f(x) = x^{1 + \rthree}$. From the previously established bounds, the space complexity follows. 

    For the first claim, let $x \in X$ such that $\abs{\locn (x, r)} \geq n^{1 - \rrep}$ and $K_i = \bm{1} \lbrb{z_i \in \locn (x, r)}$. We have $\mb{P} (K_i = 1) \geq n^{-\rrep}$. Therefore, we have that $\mb{E} \lsrs{K \coloneqq \sum_{i = 1}^{\nrep} K_i} \geq C \log (n / \delta)$ for some large constant $C$. By noting that $\Var (K) \leq \mb{E} [K]$, we have by Bernstein's inequality that with probability at least $\delta / (16n)$ that there exists $K_i = 1$ for some $i \in [\nrep]$. Therefore, there exists $z \in \mc{Z}$ such that $\norm{z - x} \leq 2r$. By a union bound, this establishes the first claim with probability at least $1 - \delta / 16$.

    We now prove each of the three conclusions of the second claim of the lemma separately. First note that from \cref{def:na_ann_part}, the linearity of expectation and the fact that $\abs{\locn (x, r)} \leq n^{1 - \rrep}$, we have for any $i \in [l]$:
    \begin{gather*}
        \mb{P}\lbrb{\sum_{\substack{y \in X \\ \norm{y - x} \geq 2r}} \sum_{S \in h_i(x)} \bm{1} \lbrb{y \in S} \leq O^*(n^{\rfour})} \geq 0.99 \text{ and }\\
        \mb{E} \lbrb{\sum_{\substack{y \in X \\ \norm{y - x} \leq 2r}} \sum_{S \in h_i(x)} \bm{1} \lbrb{y \in S}} \leq O^*(n^{1 - \rrep + \rthree}).
    \end{gather*}
    By an application of Markov's Inequality on the second inequality, we have by the union bound:
    \begin{gather*}
        \mb{P}\lbrb{\sum_{\substack{y \in X \\ \norm{y - x} \geq 2r}} \sum_{S \in h_i(x)} \bm{1} \lbrb{y \in S} \leq O^*(n^{\rfour}) \text{ and } \sum_{\substack{y \in X \\ \norm{y - x} \leq 2r}} \sum_{S \in h_i(x)} \bm{1} \lbrb{y \in S} \leq O^*(n^{(1 - \rrep) + \rthree})} \geq 0.985.
    \end{gather*}
    Letting $L_i$ be the indicator random variable for the above event for $x$ in the data structure~$\mc{D}_i$, we have by the Chernoff bound that $\mb{P}\lbrb{\sum_{i = 1}^l L_i \geq 0.98l} \geq 1 - \delta / (16n)$. Therefore, the first conclusion of the second claim holds for $x$ with probability at least $1 - \delta / (16n)$. A union bound now establishes the lemma for all $x \in X$ with probability at least $1 - \delta / 16$.

    For the second conclusion of the second claim, let $i$ be such that $L_i = 1$ from the preceding discussion. Let $S \in h_i(x)$ be such that $\abs{\locn (x, r) \cap S} > \abs{(X \setminus \locn (x, r)) \cap S}$; that is, $S$ is a set in $h_i(x)$ with more points from the local neighborhood of $x$ than far away points. For any such $S$, the probability that $w_j \in \locn (x, r)$ is at least $1/2$. Therefore, we have by the definition of $p$ that with probability at least $1 - \delta^{\ddagger}$ that there exists $w \in \mc{W}_{i, S}$ such that $w \in \locn (x, r)$. Note that all the points in $\locn (x, r) \cap S$ are assigned to $w$ in this case. By the union bound, the probability that this happens for all such $S \in h_i (x)$ is at least $1 - \delta / (ln^{2 + (\rthree + \rfour)})$. Conversely, for $S$ such that $\abs{\locn (x, r) \cap S} \leq \abs{(X \setminus \locn (x, r)) \cap S}$, the total number of unassigned points is upper bounded by $2 \abs{(X \setminus \locn (x, r)) \cap S}$ and by summing over all such $S$ the conclusion follows from the definition of $L_i$. Therefore, by a union bound, we get that the second conclusion of the second claim holds for $x \in X$ for all $i$ such that $L_i = 1$ with probability at least $1 - \delta / (n^{2 + (\rthree + \rfour)})$. The conclusion for all $x \in X$ follows from another union bound with probability at least $1 - \delta / 16$.

    Finally, for the last conclusion of the second claim, let $x, y \in X$ such that $\norm{x - y} \leq r$ and $M_i = \bm{1} \{\exists S \in h_i (x) : y \in S\}$. From \cref{def:na_ann_part}, we have that $\mb{P} \lbrb{M_i = 1} \geq 0.99$. Therefore, we have by the Chernoff bound that with probability at least $1 - \delta / (16n^3)$ that the conclusion holds for specific $x, y \in X$ satisfying $\norm{x - y} \leq r$. Through a union bound, the conclusion holds for all $x, y \in X$ with $\norm{x - y} \leq r$ with probability at least $1 - \delta / (16n)$. 

    A final union bound over all the events described in the proof gives us the required guarantees on the running of \cref{alg:fixed_scale_inst} with probability at least $1 - \delta / 2$.
  \end{proof}

  \setcounter{AlgoLine}{0}
  \begin{algorithm}[H]
        \KwIn{Dataset $X = \{x_i\}_{i = 1}^n$, Projection $\Pi$, Scale $r > 0$, Repetition $\rrep$, Failure Probability $\delta$}
         $\nrep \gets \Theta (n^{\rrep} \log (n / \delta)), \delta^{\dagger} \gets \Theta \lprp{\delta / \nrep}$\;
         Pick $\nrep$ main points, $\mc{Z} = \lbrb{z_i}_{i = 1}^{\nrep}$, uniformly at random from $X$\;
         For each $z \in \mc{Z}$, instantiate independent $(\rone, \rtwo, (1 + \eps^\dagger))$-AANN data structures, $\mc{D}_z$, with point set $\lbrb{y_i = \wt{v}_{x_i, z}}$ and failure probability $\delta^{\dagger}$ and for $x \in X$, assign $x$ to $z$ if $\norm{x - z} \leq 2r$ and add $x$ to $A_z$\;
         $l \gets \Theta \lprp{\log (n / \delta)}, \delta^{\ddagger} \gets \Theta \lprp{\delta / ln^{4 + 2(\rthree + \rfour)}}$\;
         Instantiate $l$ independent $(\rthree, \rfour)$-AP data structures, $\{\mc{D}_i = \lprp{h_i, \mc{S}_i}\}_{i = 1}^l$ with pointset $X$\;
        \For {$i \in l$ and $S \in \mc{S}_i$}{
             $p \gets \Theta \lprp{\log 1 / \delta^{\ddagger}}$\;
             Pick $p$ points, $\mc{W}_{i, S} = \{w_j\}_{j = 1}^p$, uniformly at random from $S$\;
             For each $w \in \mc{W}_{i, S}$, instantiate $(\rone, \rtwo, (1 + \eps^\dagger))$-AANN data structure, $\mc{D}_{i, S, w}$, with point set $\lbrb{\wt{v}_{x,w}}_{x \in S}$ and failure probability $\delta^{\ddagger}/p$. Assign $x \in S$ to $w$ if $\norm{x - w} \leq 4r$ and add $x$ to $A_{i, S}$\;
        }
        \Return{$\mc{D} = \lprp{X, \mc{Z}, \lbrb{\mc{D}_z, A_z}_{z \in \mc{Z}}, \lbrb{\mc{D}_i}_{i = 1}^l, \{\mc{W}_{i, S}\}_{i \in [l], S \in \mc{D}_i}, \lbrb{\lbrb{\mc{D}_{i, S, w}}_{w \in \mc{W}_{i, S}}, A_{i, S}}_{i \in [l], S \in \mc{D}_i}}$}
    \caption{$\fixedscaleinst (X, \Pi, r, \rrep, \delta)$}
    \label{alg:fixed_scale_inst}
  \end{algorithm}

  We delay the analysis of the query procedure to the next subsection where we incorporate the fixed scale data structure into a multi-scale separating oracle and conclude the proof of \cref{thm:term_embedding}.

  \setcounter{AlgoLine}{0}
 \begin{algorithm}
        \KwIn{Data Structure $\mc{D}$, Query $q$, Candidate $v$, Approximate Nearest Nbr $\hat{x}$, Tolerance $\eps^{\dagger}$}
        \If {$\norm{v - \Pi \hat{x}} \geq (1 + 10 \eps^\dagger) \norm{q - \hat{x}}$}
            {\Return{$\hat{x}$}}
        \ElseIf {$\hat{x}$ is assigned to any $z \in \mc{Z}$}{
            \If {$\norm{v - \Pi z} \geq (1 + 10\eps^\dagger) \norm{q - z}$}
                {\Return{$z$}}
            \If {$\abs{\inp{\tv_{z, \xhat}}{\tv_{\xhat}}} \geq 20\eps^\dagger$} 
                {\Return{$(z, \hat{x})$}}
             Query $\mc{D}_z$ with $\tv_{z}$ and $-\tv_z$ to obtain solutions $y_i = \tv_{x_i, z}, y_j = \tv_{x_j, z}$\;
             $u = (x_i, z)$ if $\abs{\inp{\wt{v}_z}{y_i}} \geq 20\eps^\dagger$, $(x_j, z)$ if $\abs{\inp{\wt{v}_z}{y_j}} \geq 20\eps^\dagger$, $\fail$ otherwise \; 
             \Return{$u$}
             }
        \Else{ 
            \For {$i \in [l] : \sum_{S \in h_i (\hat{x})} \abs{S \setminus A_{i, S}} \leq O^*(n^{\rfour})$, $\sum_{S \in h_x (\hat{x})} \abs{A_{i, S}} \leq O^*(n^{(1 - \rrep) + \rthree})$ and $S \in h_i(\hat{x})$}{
                \For {$x \notin A_{i, S}$}{
                    \If {$\abs{\inp{\tv_{\xhat}}{\tv_{x, \xhat}}} \geq 20\eps^\dagger$}
                        {\Return{$(x, \hat{x})$}}
                }
                \For {$w \in \mc{W}_{i, S}$}{
                    \If {$\norm{v - \Pi w} \geq (1 + 10\eps^\dagger) \norm{q - w}$}
                        {\Return{$w$}}
                    \If {$\abs{\inp{\tv_{w, \xhat}}{\tv_{\xhat}}} \geq 20\eps^\dagger$} 
                        {\Return{$(w, \hat{x})$}}
                     Query $\mc{D}_{i, S, w}$ with $\tv_{w}$ and $-\tv_{w}$ to obtain solutions $y_i = \tv_{x_i, w}, y_j = \tv_{x_j, w}$\;
                     $u = (x_i, w)$ if $\abs{\inp{\tv_{x_i, w}}{\tv_w}} \geq 20\eps^\dagger$, $(x_j, w)$ if $\abs{\inp{\tv_{x_j, w}}{\tv_w}} \geq 20\eps^\dagger$, $\fail$ otherwise \;
                    \Return{$u$}
                }
            }
        }
        \Return{$\fail$}
    \caption{$\fixedscalequery (\mc{D}, q, v, \hat{x}, \eps^\dagger)$}
    \label{alg:fixed_scale_query}
\end{algorithm}

\subsection{Multi-Scale Reduction and Proof of Theorem \ref{thm:term_embedding}}
\label{ssec:multi_scale_term}

In this subsection, we wrap up the proof of \cref{thm:term_embedding} by incorporating the fixed scale violator detection method from \cref{ssec:fixed_scale_terminal} into a multi-scale procedure. We first define the auxiliary data structures that we will assume access to for the rest of the proof:
\begin{enumerate}
    \item From \cref{lem:conv_hull}, we may assume $\Pi \in \mb{R}^{k \times d}$ satisfies $\eps^\dagger$-convex hull distortion for~$X$ with $\eps^\dagger = \wt{c} \eps$ for some suitably small constant $\wt{c} > 0$ with $k = O(\eps^{-2}\log n)$
    \item From \cref{thm:aann_main}, we may assume access to a partition tree $\mc{T}$ satisfying \cref{def:part_tree}
    \item We may assume access to a $(\rthree, \rfour, (1 + \eps))$-Adaptive Approximate Nearest Neighbor data structure for $X$ built on $\mc{T}$ from \cref{thm:aann_main}. 
\end{enumerate}

Recall the definitions of a partition tree from \cref{ssec:gen_red}. The first is the partitioning of a data points into connected components of a graph constructed by thresholding the distances between points in the dataset.

\RestateGo{\restategpcc}
\gpcc*

The second is the notion of refinement between partitions.
\RestateGo{\restatepartref}
\partref*

Next, define $\rmed(X)$ as:
\begin{equation*}
    \rmed(X) = \min \{r > 0: \exists C \in \cc (X, r) \text{ with } \abs{C} \geq n / 2\}.
\end{equation*}

The last is the definition of the partition tree itself.
\RestateGo{\restateparttree}
\parttree*

We first describe the data structure which reduces the multi-scale violator detection problem to the fixed scale setting using the partition tree that is returned by \cref{thm:aann_main}. The data structure implementing the multi-scale violator detector is constructed in \cref{alg:multi_scale_inst} and the corresponding query procedure is described in \cref{alg:multi_scale_query}. \cref{alg:multi_scale_inst} takes as input the projection matrix $\Pi$, the Partition Tree $\mc{T}$, a failure probability $\delta$ and a repetition factor $\rrep$. The following straightforward lemma is the only guarantee we will require of \cref{alg:multi_scale_inst}:
\begin{lemma}
    \label{lem:multi_scale_inst}
    Let $X = \{x_i\}_{i = 1}^n \subset \mb{R}^d$, $\mc{T}$ be a valid Partition Tree of $X$, $\delta \in (0, 1)$ and $\rrep \in [0, 1]$. Then, the data structures output by \cref{alg:multi_scale_inst} on input $X, \mc{T}$, $\delta$ all satisfy the conclusions of \cref{lem:fixed_scale_inst} with probability at least $1 - \delta$.
\end{lemma}
\begin{proof}
    The proof follows from a union bound over the nodes of the tree, the $l$ data structures instantiated for each node, the setting of $\delta^\dagger$ in \cref{alg:multi_scale_inst} and \cref{lem:fixed_scale_inst}.
\end{proof}

\setcounter{AlgoLine}{0}
\begin{algorithm}[H]
        \KwIn{Projection $\Pi$, Partition Tree $\mc{T}$, Failure Probability $\delta$, Repetition $\rrep$}
        $\delta^{\dagger} \gets \Theta (\delta / (nd)^{10})$\;
        \For {$\mc{T}^\prime = \{Z, \{\mc{T}_C\}_{C \in \clow}, \mc{T}_{\mathrm{rep}}, \clow, \chigh, \crep, \rhat\} \in \mc{T}$}{
             $\rlowt \gets \Theta (\rhat / (nd)^{20}), \rhight \gets \Theta ((nd)^{20} \rhat), \gamma \gets \Theta (1 / \log^3 n)$\;
            \For {$i = 0, \dots, \ceil*{\frac{\log {\rhight / \rlowt}}{\gamma}} = p$}{
                 Instantiate, $\mc{D}_{\mc{T}^\prime, i} \gets \fixedscaleinst (Z, \Pi, (1 + \gamma)^i \rlowt, \rrep, \delta^{\dagger})$
            }
        }
        \Return{$\{\mc{D}_{\mc{T}^\prime, i}\}_{\mc{T}^\prime \in \mc{T},\, i \in \{0\} \cup [p]}$}
    \caption{$\multiscaleinst (\Pi, \mc{T}, \delta, \rrep)$}
    \label{alg:multi_scale_inst}
  \end{algorithm}


We are now ready to conclude the proof of \cref{thm:term_embedding}. Suppose $q \in \mb{R}^d$ and we are required to construct a valid terminal embedding for $q$ with respect to the point set $X$. We may assume access to $(\xhat, \mc{T}^\prime)$ for $\xhat$ being an approximate nearest neighbor of $q$ in $X$ and $\mc{T}^\prime \in \mc{T}$ satisfying the conclusion of \cref{thm:aann_main}. Also, it suffices to construct a valid terminal embedding for the set of points in $\mc{T}^\prime$. We may assume $\mc{T}^\prime$ has more than one element as in the one element case, any point on a sphere of radius $\norm{q - \hat{x}}$ around $(\Pi \hat{x}, 0)$ suffices from \cref{thm:aann_main}. Our algorithm is based on the following set of convex constraints where $Z$ is the set of points in $\mc{T}^\prime$:
\begin{gather*}
    \forall x, y \in Z: \abs{\inp{z - \Pi x}{\Pi (y - x)} - \inp{q - x}{y - x}} \leq 20\eps^\dagger \norm{q - x} \norm{y - x} \\
    \forall x \in Z: \norm{z - \Pi x} \leq (1 + 10\eps^\dagger) \norm{q - x} \tag{Req} \label{eq:term_const}.
\end{gather*}
Let $K$ denote the convex subset of $\mb{R}^k$ satisfying the above set of constraints. Let $\hat{r} = \norm{q - \xhat}$. By \cref{lem:multi_point_nn} and Cauchy-Schwarz, $K$ is non-empty and there exists $\wt{z} \in K$ with $\mb{B}(\wt{z}, \eps^\dagger \hat{r}) \subset K$ (there exists $\wt{z}$ satisfying the above constraints with the right-hand sides replaced by $15\eps^\dagger$ and $8\eps^\dagger$ respectively). Also, note that $K \subseteq \mb{B} (\Pi \xhat, 2 \hat{r})$ from the above set of constraints. While one can show a feasible point in $K$ can be used to construct a valid terminal embedding, we construct a slightly weaker oracle. We construct an oracle $\mc{O}$ satisfying:
\begin{enumerate}
    \item If $\mc{O} (v)$ outputs $\fail$, $v$ can be extended to a terminal embedding
    \item Otherwise $\mc{O} (v)$ outputs a valid separating hyperplane for $K$.
\end{enumerate}
To utilize the above oracle for constructing our terminal embeddings, we require the following guarantees on the ellipsoid algorithm which essentially states that for a convex set containing a ball of substantially large radius with an oracle that is guaranteed to output a correct separating hyperplane or $\mrm{FAIL}$, a point for which the oracle outputs $\mrm{FAIL}$ may be found quickly. In the context of the proof, $\mc{O}$ will be our oracle and the convex body will be the feasible points of \ref{eq:term_const}. The proof the statement is deferred to \cref{ssec:ellips}.
\RestateInit{\restateweakellipsoid}
\begin{restatable}{lemma}{weakellipsoid}
    \label{lem:weak_ellipsoid}\RestateRemark{\restateweakellipsoid}
    Suppose $\eps > 0$ and $K \subset \mb{R}^d$ be a closed convex set such that there exists $x^* \in K$ such that for all $y \in \mb{B}(x^*, \eps)$, $y \in K$. Furthermore, let $x \in \mb{R}^d$ and $R > 0$ be such that $K \subset \mb{B}(x, R)$. Suppose further that $\mc{O}$ satisfies for any input $z$, $\mc{O}$:
    \begin{enumerate}
        \item Outputs $v \neq 0$ such that for all $y \in K$, we have $\inp{v}{y - x} \geq 0$ or
        \item Outputs $\mrm{FAIL}$.
    \end{enumerate}
    Then, \cref{alg:ellipsoid} when instantiated with $x, R$ and $\mc{O}$, outputs $\hat{x}$ satisfying $\mc{O}(\hat{x}) = \mrm{FAIL}$. Furthermore, the number of iterations of the algorithm is bounded by $O\lprp{d^2 \log \frac{R}{\eps}}$ and hence the total computational complexity is bounded by $O\lprp{d^4 \log \frac{R}{\eps}}$.
\end{restatable}
From \cref{lem:weak_ellipsoid} and the fact that $K$ contains a ball of radius $\eps^\dagger \hat{r}$ and is contained in a ball of radius $2\hat{r}$, we see that there is a procedure which with $O^*(1)$ many queries to $\mc{O}$ and $O^*(1)$ total additional computation, outputs a point $v^*$ for which $\mc{O}$ outputs $\fail$. Our oracle implementing this property is defined simply by \cref{alg:multi_scale_query}. Through the rest of this subsection, we focus our attention on proving the correctness of the oracle.

Let $\mc{T}^\prime = \{Z, \{\mc{T}_C\}_{C \in \clow}, \mc{T}_{\mathrm{rep}}, \clow, \chigh, \crep, \rhat\}$ and $m = \abs{Z}$. First, we focus on the easy case when \cref{alg:multi_scale_query} does not output $\fail$. 

\begin{lemma}
    \label{lem:ms_not_fail}
    If \cref{alg:multi_scale_query} does not output $\fail$ on $v$, it outputs a separating hyperplane for $v$ from $K$.
\end{lemma}
\begin{proof}
    Note that this only happens when some $x$ or pair $(x,y)$ is returned by \cref{alg:fixed_scale_query}. Now, \cref{alg:fixed_scale_query} may return an $x$ in one of two ways:
    \begin{enumerate}
        \item[] \textbf{Case 1: } $\norm{v - \Pi x} \geq (1 + 10 \eps^\dagger) \norm{q - x}$ for $x$ returned by the algorithm. In this case, the correctness of the procedure trivially follows from \ref{eq:term_const}.
        \item[] \textbf{Case 2: } For $x, y \in Z$, the vectors $v_1 = \wt{v}_x$ and $v_2 = \wt{v}_{y,x}$ satisfy $\abs{\inp{v_1}{v_2}} \geq 20 \eps^\dagger$. We have:
        \begin{align*}
            \abs*{\inp{q - x}{y - x} - \inp{v - \Pi x}{\Pi (y - x)}} &\geq 20\eps^\dagger \norm{(q - x, -(v - \Pi x))} \norm{(y - x, \Pi (y - x))} \\
            &\geq 20 \eps^\dagger \norm{q - x} \norm{y - x}
        \end{align*}
        which is again a violator of \ref{eq:term_const} which proves correctness in this case as well. 
    \end{enumerate}
    This concludes the proof of the lemma.
\end{proof}

Next, we consider the alternate case where \cref{alg:multi_scale_query} outputs $\fail$, where we show that the input $v$ can be used to construct a valid terminal embedding for $q$ with respect to $Z$ and consequently for $X$ from the guarantees of \cref{thm:aann_main}. We recall the following crucial parameters from \cref{alg:multi_scale_inst} where $\rhat$ is defined in \cref{def:part_tree}
\begin{equation*}
    \rlowt = \Theta (\rhat / (nd)^{20}) \qquad \rhight = \Theta ((nd)^{20} \rhat). \tag{RSET} \label{eq:rset}
\end{equation*}
\begin{lemma}
    \label{lem:ms_fail}
    If \cref{alg:multi_scale_query} outputs $\fail$ on $v$, $z_q = (v, \sqrt{\max(0, \norm{q - \hat{x}}^2 - \norm{v - \Pi \hat{x}}^2)})$ satisfies:
    \begin{equation*}
        \forall z \in Z: (1 - 2000\eps) \norm{q - z} \leq \norm{z_q - (\Pi z, 0)} \leq (1 + 2000\eps) \norm{q - z}.
    \end{equation*}
\end{lemma}
\begin{proof}
Note that since \cref{alg:fixed_scale_query} returned $\fail$ for all fixed scale data structures, we must have $\norm{v - \Pi \hat{x}} \leq (1 + 10 \eps^\dagger) \norm{q - \hat{x}}$. As a consequence, we get that $\norm{z_q - (\Pi \hat{x}, 0)} \leq (1 + 10\eps^\dagger) \norm{q - \hat{x}}$. Now, we need to show that $z_q$ is a valid terminal embedding of $q$ for an arbitrary $x \in Z$. We first consider the case where $\norm{x - \hat{x}} \geq 0.5 \rhight$ as defined in \cref{alg:multi_scale_inst} for the node $\mc{T}^\prime$. For this point, we have:
\begin{equation*}
    \norm{q - x} \geq \norm{\hat{x} - x} - \norm{q - \hat{x}} \geq \lprp{1 - \frac{1}{(10nd)^{9}}} \norm{\hat{x} - x}
\end{equation*}
where the last inequality comes from the condition on $\norm{q - \hat{x}}$ from \cref{thm:aann_main} and the definition of $\rhight$. Now, we get:
\begin{align*}
    \abs{\norm{x - q} - \norm{z_q - (\Pi x, 0)}} &\leq \abs{\norm{\hat{x} - x} - \norm{\Pi (\hat{x} - x)}} + \norm{\hat{x} - q} + \norm{z_q - (\Pi \xhat, 0)} \\
    &\leq \eps^\dagger \norm{\hat{x} - x} + (2 + \eps) \norm{\hat{x} - q} \\
    &\leq 2 \eps^\dagger \norm{q - x} + (2 + \eps) \norm{\hat{x} - q} \\
    & \leq \norm{q - x} \lprp{2 \eps^\dagger + 3 \frac{\norm{\hat{x} - q}}{\norm{q - x}}}\leq \eps \norm{x - q}
\end{align*}
where the second inequality follows from the fact that $\Pi$ as $\eps^\dagger$-convex full distortion for $X$, the third from the previous display and last inequality from the fact that $\norm{\hat{x} - q} \leq O((nd)^{10} \rhat)$ (\cref{thm:aann_main}) and $\norm{q - x} \geq \norm{\hat{x} - x} / 2 \geq \rhight / 4$ and the values of $\rhat$ and $\rhight$ (\ref{eq:rset}). 

We now consider the alternative case where $\norm{x - \hat{x}} \leq 0.5 \rhight$. In this case, from the definition of $p$ in \cref{alg:multi_scale_inst}, let $i^*$ be the smallest $i \in [p]$ such that $\norm{x - \hat{x}} \leq (1 + \gamma)^{i^*} \rlowt$. Note that $i^*$ is finite from our condition on $\norm{x - \xhat}$. For this $i^*$, consider the fixed scale data structure, $\mc{D}_{\mc{T}^\prime, i^*} = \lprp{Z, \mc{Z}, \lbrb{\mc{D}_z, A_z}_{z \in \mc{Z}}, \lbrb{\mc{D}_i}_{i = 1}^l, \{\mc{W}_{i, S}\}_{i \in [l], S \in \mc{D}_i}, \lbrb{\lbrb{\mc{D}_{i, S, w}}_{w \in \mc{W}_{i, S}}, A_{i, S}}_{i \in [l], S \in \mc{D}_i}}$ constructed in \cref{alg:multi_scale_inst}. Letting $\wt{r} = (1 + \gamma)^{i^*} \rlowt$, we prove the lemma in two cases corresponding to the structure of $\mc{D}_{\mc{T}^\prime, i^*}$ as described in \cref{lem:fixed_scale_inst}:
\begin{enumerate}
    \item[] \textbf{Case 1: } $\hat{x} \in A_z$ for some $z \in \mc{Z}$. Note that this necessarily happens if $\abs{\locn(\hat{x}, \wt{r})} \geq m^{1 - \rrep}$. 
    \item[] \textbf{Case 2: } $\hat{x} \notin A_z$ for all $z \in \mc{Z}$. Recall, this implies $\abs{\locn(\hat{x}, \wt{r})} < m^{1 - \rrep}$.
\end{enumerate}
In either case, the following two simple claims will be crucial in bounding the error terms:
\begin{claim}
    \label{clm:distance_bound}
    For all $x \in Z$, we have:
    \begin{equation*}
        \norm{x - q} \geq \frac{1}{(2 + \eps + o^*(1))} \norm{\hat{x} - x}.
    \end{equation*}
\end{claim}
\begin{subproof}
    We have
    \begin{equation*}
        \norm{\hat{x} - x} \leq \norm{\hat{x} - q} + \norm{x - q} \leq (2 + \eps + o^*(1)) \norm{x - q}.
    \end{equation*}
    By re-arranging the above inequality, we obtain our result.
\end{subproof}

\begin{claim}
    \label{clm:q_dist_bnd}
    We have:
    \begin{equation*}
        \norm{q - x} \geq \frac{5}{12} \wt{r}.
    \end{equation*}
\end{claim}
\begin{subproof}
    If $i^* = 0$, we have:
    \begin{equation*}
        \norm{q - x} \geq (1 + \eps + o^*(1))^{-1} \norm{q - \xhat} \geq (10nd)^5 \rlowt
    \end{equation*}
    from our assumption on $\norm{q - \xhat}$ (\cref{thm:aann_main}). When $i^* > 0$, we have $(1 + \gamma)^{-1} \wt{r} \leq \norm{\xhat - x} \leq \wt{r}$ and the conclusion follows from \cref{clm:distance_bound}.
\end{subproof}
In the first case, we get that $\hat{x}$ is assigned to some $z \in \mc{Z}$ with $\norm{\hat{x} - z} \leq 2\wt{r}$ which implies by the triangle inequality $\norm{x - z} \leq 3\wt{r}$. 

We now prove another claim we will use through the rest of this proof:
\begin{claim}
    \label{clm:distance_bound_xhat_centered}
    For all $w \in Z$, let $\wt{w} = \wt{v}_{w, \hat{x}}$ and $\wt{v} = \wt{v}_{\hat{x}}$. If $\abs{\inp{\wt{w}}{\wt{v}}} \leq 20\eps^\dagger$:
    \begin{equation*}
        \norm{z_q - (\Pi w, 0)} - \norm{w - q} \leq 400 \eps^\dagger \norm{w - q}.
    \end{equation*}
\end{claim}
\begin{subproof}
    We have:
    \begin{align*}
        \abs{\norm{w - q}^2 - \norm{(\Pi w, 0) - z_q}^2} &\leq \abs{\norm{w - \hat{x}}^2 - \norm{\Pi (w - \hat{x})}^2} + \abs{\norm{q - \hat{x}}^2 - \norm{z_q - (\Pi \hat{x}, 0)}^2} \\
        &\quad + 2 \abs{\inp{w - \hat{x}}{q - \hat{x}} - \inp{\Pi (w - \hat{x})}{v - \Pi \hat{x}}} \\
        &\leq 3\eps^\dagger \norm{w - \hat{x}}^2 + 25\eps^\dagger \norm{q - \hat{x}}^2 + 100 \eps^\dagger \norm{w - \hat{x}} \norm{q - \hat{x}} \\
        &\leq 750 \eps^\dagger \norm{w - q}^2
    \end{align*}
    where the second inequality follows from the fact that $\Pi$ satisfies $\eps^\dagger$-convex hull distortion for~$X$, our claim on $\norm{z_q - (\Pi \hat{x}, 0)}$ and our condition on $\inp{\wt{w}}{\wt{v}}$ while the final inequality follows from \cref{clm:distance_bound} and the fact that $\hat{x}$ is a $(1 + \eps + o^*(1))$-Approximate Nearest Neighbor of $q$. By factorizing the LHS and dividing by $\norm{w - q}$, we get the desired result.
\end{subproof}
Returning to our proof, recall $\norm{x - z} \leq 3 \wt{r}$ and from \cref{clm:q_dist_bnd}, $\norm{q - x} \geq \frac{5}{12} \wt{r}$. We claim since $\mc{D}_{\mc{T}^\prime, i^*}$ returns $\fail$ that $\abs{\inp{\wt{v}_{x,z}}{\wt{v}_z}} \leq 4\eps$. To see this, suppose $\inp{\wt{v}_{x,z}}{\wt{v}_z} \geq 4\eps$. In this case, we have:
\begin{equation*}
    \norm{\wt{v}_{x,z} - \wt{v}_z}^2 \leq 2 - 8\eps \implies \norm{\wt{v}_{x,z} - \wt{v}_{z}} \leq (1 - 2\eps)\sqrt{2}.
\end{equation*}
Therefore, $\mc{D}_z$ on input $\wt{v}_z$ returns $y$ such that $\norm{y - \wt{v}_z} \leq (1 + \eps + o^*(1)) (1 - 2\eps) \sqrt{2}$ follows from the fact that $\mc{D}_z$ is successfully instantiated. From this, we get that $\norm{y - \wt{v}_z} \leq \sqrt{2} (1 - 0.8\eps)$. From this expression, we obtain:
\begin{equation*}
    \norm{y - \wt{v}_z}^2 = 2 - 2\inp{y}{\wt{v}_z} \implies \inp{y}{\wt{v}_z} \geq 0.8 \eps
\end{equation*}
which contradicts the assumption that $\mc{D}_{\mc{T}^\prime, i^*}$ returns $\fail$. The proof for $\inp{\wt{v}_z}{\wt{v}_{x,z}} \leq -4\eps$ is similar by replacing $\wt{v}_z$ by $-\wt{v}_z$ in the above proof and using the fact that we query $\mc{D}_{z}$ with $-\wt{v}_z$ as well. Hence, we have $\abs{\inp{\wt{v}_z}{\wt{v}_{x,z}}} \leq 4\eps$. Finally, bound the deviation of $\norm{z_q - (\Pi x, 0)}$ from $\norm{q - x}$:
\begin{align*}
    \abs{\norm{x - q}^2 - \norm{(\Pi x, 0) - z_q}^2} &\leq \abs{\norm{x - z}^2 - \norm{\Pi (x - z)}^2} + \abs{\norm{q - z}^2 - \norm{z_q - (\Pi z, 0)}^2} \\
    &\quad + 2 \abs{\inp{x - z}{q - z} - \inp{\Pi (x - z)}{v - \Pi z}} \\
    &\leq 3\eps^\dagger \norm{x - z}^2 + 750 \eps^\dagger \norm{q - z}^2 + 2 \cdot 4\eps \cdot \sqrt{5} \cdot \norm{x - z} \norm{q - z} \\
    &\leq 3\eps^\dagger \norm{x - z}^2 + 750 \eps^\dagger \norm{q - z}^2 + 20 \eps \norm{x - z} \norm{q - z} \\
    &\leq 27 \eps^\dagger \wt{r}^2 + 750 \eps^\dagger (\norm{q - x} + \norm{x - z})^2 \\
    &\qquad + 20 \eps \norm{x - z} (\norm{q - x} + \norm{x - z}) \\
    &\leq 27 \eps^\dagger \wt{r}^2 + 1500 \eps^\dagger (\norm{q - x}^2 + \norm{x - z}^2) + 60 \eps \wt{r} (\norm{q - x} + 3\wt{r}) \\
    &\leq 1500 \eps \norm{q - x}^2
\end{align*}
where the second inequality follows from the fact that $\Pi$ has $\eps^\dagger$-convex hull distortion for $X$, \cref{clm:distance_bound_xhat_centered} and the fact that $\abs{\inp{\wt{v}_z}{\wt{v}_{x,z}}} \leq 4\eps$, the fourth inequality follows from the fact that $\norm{x - z} \leq 3\wt{r}$ and the last and second-to-last inequalities follow from the fact that $(a + b)^2 \leq 2(a^2 + b^2)$ and \cref{clm:q_dist_bnd}. Factorizing the LHS now establishes the lemma in this case.

We now consider the alternate case where $\abs{\locn(\hat{x}, \wt{r})} < m^{1 - \rrep}$. For the fixed scale data structure, $\mc{D}_{\mc{T}^\prime, i^*}$, recall that we have from \cref{lem:fixed_scale_inst} that:
\begin{gather*}
    \sum_{i = 1}^l \bm{1} \lbrb{
    \begin{gathered}
        \sum_{\substack{y \in Z \\ \norm{y - \hat{x}} \geq 2\wt{r}}} \sum_{S \in h_i(\hat{x})} \bm{1} \lbrb{y \in S} \leq O^*(m^{\rfour}) \text{ and } \\ 
        \sum_{\substack{y \in Z \\ \norm{y - \hat{x}} \leq 2\wt{r}}} \sum_{S \in h_i(\hat{x})} \bm{1} \lbrb{y \in S} \leq O^*(m^{(1 - \rrep) + \rthree})    
    \end{gathered}
    } \geq 0.98l \\
    \sum_{i = 1}^l \bm{1} \lbrb{\sum_{S \in h_i (\hat{x})} \abs{S \setminus A_{i, S}} \leq O^*(m^{\rfour})} \geq 0.98l \\
    \sum_{i = 1}^l \bm{1} \lbrb{\exists S \in h_i (\hat{x}) \text{ such that } x \in S} \geq 0.98l.
\end{gather*}
In particular, we have by the union bound, there exists a set $\mc{J} \subset [l]$ with $\abs{\mc{J}} \geq 0.9l$ satisfying the indicators in all three of the above equations. For any $j \in \mc{J}$, from \cref{alg:fixed_scale_query}, all the subsets in $h_j(\hat{x})$ are fully explored. For some $j \in \mc{J}$, consider $S \in h_j(\hat{x})$ such that $x \in S$. Now, if $x \notin A_{j, S}$ or $x \in \mc{W}_{j, S}$, conclusion of the lemma follows from \cref{clm:distance_bound_xhat_centered}. Hence, the only case left to consider is the case where $x \in A_{j, S}$. Let $x$ be assigned to $w \in \mc{W}_{j, S}$ such that $\norm{x - w} \leq 4 \wt{r}$. In this case, we proceed identically to the previous case where $\hat{x}$ is assigned to some $z \in \mc{Z}$. Since $\mc{D}_{\mc{T}^\prime, i^*}$ returns $\fail$, $\abs{\inp{\wt{v}_w}{\wt{v}_{x,w}}} \leq 4\eps$. To see this, assume $\inp{\wt{v}_w}{\wt{v}_{x,w}} \geq 4\eps$ and we have:
\begin{equation*}
    \norm{\wt{v}_w - \wt{v}_{x,w}}^2 \leq 2 - 8\eps \implies \norm{\wt{v}_w - \wt{v}_{x,w}} \leq (1 - 2\eps)\sqrt{2}.
\end{equation*}
Therefore, $\mc{D}_{j, S, w}$ on input $\wt{v}_w$ returns $y$ such that $\norm{y - \wt{v}_w} \leq (1 + \eps + o^*(1)) (1 - 2\eps) \sqrt{2}$ follows from the fact that $\mc{D}_{j, S, w}$ is successfully instantiated. From this, we get $\norm{y - \wt{v}_w} \leq \sqrt{2} (1 - 0.8\eps)$. We obtain:
\begin{equation*}
    \norm{y - \wt{v}_w}^2 = 2 - 2\inp{y}{\wt{v}_w} \implies \inp{y}{\wt{v}_w} \geq 0.8 \eps
\end{equation*}
which contradicts the assumption that $\mc{D}_{\mc{T}^\prime, i^*}$ returns $\fail$. The proof when $\inp{\wt{v}_w}{\wt{v}_{x,w}} \leq -4\eps$ is similar by replacing $\wt{v}_w$ by $-\wt{v}_w$ and using the fact that we query $\mc{D}_{j, S, w}$ with $-\wt{v}_w$ as well. Hence, we have that $\abs{\inp{\wt{v}_w}{\wt{v}_{x,w}}} \leq 4\eps$. As before, we bound the deviation of $\norm{z_q - (\Pi x, 0)}$ from $\norm{q - x}$:
\begin{align*}
    \abs{\norm{x - q}^2 - \norm{(\Pi x, 0) - z_q}^2} &\leq \abs{\norm{x - w}^2 - \norm{\Pi (x - w)}^2} + \abs{\norm{q - w}^2 - \norm{z_q - (\Pi w, 0)}^2} \\
    &\quad + 2 \abs{\inp{x - w}{q - w} - \inp{\Pi (x - w)}{v - \Pi w}} \\
    &\leq 3\eps^\dagger \norm{x - w}^2 + 750 \eps^\dagger \norm{q - w}^2 + 2 \cdot 4\eps \cdot \sqrt{5} \cdot \norm{x - w} \norm{q - w} \\
    &\leq 3\eps^\dagger \norm{x - w}^2 + 750 \eps^\dagger \norm{q - w}^2 + 20 \eps \norm{x - w} \norm{q - w} \\
    &\leq 48 \eps^\dagger \wt{r}^2 + 750 \eps^\dagger (\norm{q - x} + \norm{x - w})^2 \\
    &\qquad + 20 \eps \norm{x - w} (\norm{q - x} + \norm{x - w}) \\
    &\leq 48 \eps^\dagger \wt{r}^2 + 1500 \eps^\dagger (\norm{q - x}^2 + \norm{x - w}^2) + 80 \eps \wt{r} (\norm{q - x} + 4\wt{r}) \\
    &\leq 2000 \eps \norm{q - x}^2
\end{align*}
where the second inequality follows from the fact that $\Pi$ has $\eps^\dagger$-convex hull distortion for $X$, \cref{clm:distance_bound_xhat_centered} and the fact that $\abs{\inp{\wt{v}_w}{\wt{v}_{x,w}}} \leq 4\eps$, the fourth inequality follows from the fact that $\norm{x - w} \leq 4\wt{r}$ and the last and second-to-last inequalities follow from the fact that $(a + b)^2 \leq 2(a^2 + b^2)$ and \cref{clm:q_dist_bnd}. Factorizing the LHS now establishes the lemma in this case concluding the proof of the lemma.
\end{proof}

Having proved the correctness of \cref{alg:multi_scale_query}, we finally bound its runtime. Before doing so, we require a data structure allowing a weak form of dimensionality reduction that allows a speed-up of our algorithm. Note that standard techniques for dimensionality reduction may not be used in our setting as the queries may be chosen based on the data structure itself. For example, for the Johnson-Lindenstrauss lemma, the queries may be chosen orthogonal to the rows of the projection matrix violating its correctness guarantees. We first define an approximate inner product condition below for a matrix $\Pi$ and $x, y, z \in \mb{R}^d$ and $\eps > 0$
\begin{equation*}
    \abs*{\inp{\Pi (x - z)}{\Pi (y - z)} - \inp{x - z}{y - z}} \leq \eps \norm{x - z} \norm{y - z} \tag{AP\text{-}IP}. \label{eq:ap_ip}
\end{equation*}
For $x,y,z \in \mb{R}^d$ and $\eps > 0$, we say a matrix $\Pi$ satisfies $\ref{eq:ap_ip}(\eps, x, y, z)$ if it satisfies \ref{eq:ap_ip}. For a dataset, $X \subset \mb{R}^d$, $\Pi$ satisfies $\ref{eq:ap_ip}(\eps, X)$ if it satisfies $\ref{eq:ap_ip}(\eps, x, y, z)$ for all $x,y,z \in X$. 

We utilize the following theorem which states that if one initializes $\wt{O} (d)$ many independent JL-sketches, for any input query, most of them satisfy the \ref{eq:ap_ip} condition for the dataset $X$ augmented with the input query. This essentially allows us to reduce the dimensionality of the search problem by instead verifying the inner product condition in \ref{eq:gen_prog} on a randomly drawn sketch from this set as opposed to the $d$-dimensional space containing the dataset.

\RestateInit{\restatejlrep}
\begin{restatable}{theorem}{jlrep}
    \label{thm:jl_rep}\RestateRemark{\restatejlrep}
    Let $X = \{x_i\}_{i = 1}^n \subset \mb{R}^d$, $\eps \in \lprp{\frac{1}{\sqrt{nd}}, 1} $ and $\delta \in (0, 1)$. Furthermore, let $\{\Pi^i\}_{i = 1}^m \subset \mb{R}^{k \times d}$ for $m \geq \Omega((d + \log 1 / \delta) \log nd)$ and $k \geq \Omega\lprp{\frac{\log n}{\eps^2}}$ with $\Pi^i_{j, l} \overset{iid}{\thicksim} \mc{N} (0, 1 / k)$. Then, we have:
    \begin{equation}
        \forall q \in \mb{R}^d : \sum_{i = 1}^m \bm{1} \lbrb{\Pi^i \text{ satisfies } \ref{eq:ap_ip}(\eps,X \cup \{q\})} \geq 0.95m \tag{JL-Rep} \label{eq:jl_rep}
    \end{equation}
    with probability at least $1 - \delta$.
\end{restatable}

We now use \cref{thm:jl_rep} to bound the runtime of \cref{alg:multi_scale_query}.

\begin{lemma}
    \label{lem:ms_qtime}
    \cref{alg:multi_scale_query} is implementable in time $O^*(d + n^{\rtwo} + n^{\rfour} + n^{\rfour + (1 + \rthree - \rfour - \rrep) \rtwo})$ with probability at least $1 - n^{-10}$. 
\end{lemma}
\begin{proof}
    Note that we may restrict ourselves to bounding the runtime of \cref{alg:fixed_scale_query} as we query at most $O^* (1)$ many of them. For a single fixed scale data structure, $\mc{D}_{\mc{T}^\prime, i}$ being queried, computing the sets $h_i (\xhat)$ takes time $O^*(dn^{\rfour})$. If $\xhat$ is assigned to a point $z$, the nearest neighbor procedure takes time $O^*(dn^{\rtwo})$. Otherwise, processing the unassigned points takes time $O^*(dn^{\rfour})$ and for the assigned points, there are at most $O^*(n^{1 - \rrep + \rthree})$ in at most $O^*(n^{\rfour})$ many sets. From the concavity of the function $f(x) = x^{\rtwo}$, we get that the maximum time taken to query all of these nearest neighbor data structures is at most $O^*(dn^{\rfour + (1 + \rthree - \rfour - \rrep) \rtwo})$. We now discuss how to decouple the dimensionality term from the term dependent on $n$. 
    
    As explained in \cref{ssec:runtime_adapt}, it suffices for our argument to use $\Pi^\prime (x - y)$ instead of $x - y$ as the first component in the construction of the vectors used to instantiate the data structures $\mc{D}_z$ and $\mc{D}_{i, S, w}$ in \cref{alg:fixed_scale_inst} for any $\Pi^\prime$ satisfying:
    \begin{equation*}
        \forall x, y, z \in X \cup \{q\}: \abs{\inp{\Pi^\prime (x - z)}{\Pi^\prime (y - z)} - \inp{x - z}{y - z}} \leq o^*(1) \cdot \norm{x - z} \norm{y - z}.
    \end{equation*}
    From \cref{thm:jl_rep}, we get that an instantiation of $m = \wt{\Theta} (d)$ JL sketches, $\{\Pi^i\}_{i \in [m]}$, with $\Theta (\log n \log \log n)$ rows each satisfies the above condition for at least $95\%$ of them with probability at least $1 - n^{-10}$. To construct our final violator detection subroutine, we instantiate $m$ copies of our violator detection algorithm for the projections of the data points with respect to each of the sketches $\Pi^i$. At query time, we simply sample $\Theta (\log n)$ many of these sketches for a possible violator. We then check the validity of each of the returned candidates which can be done in time $\Ot (d)$. Since, $95\%$ of the sketches satisfy the above condition, at least one of the sampled sketches will with probability at least $1 - n^{-10}$ and hence, satisfies the guarantees required of our oracle.
\end{proof}

\paragraph{Proof of \cref{thm:term_embedding}:} Hence, \cref{lem:ms_fail,lem:ms_not_fail,lem:ms_qtime} along with \cref{thm:aann_main,thm:jl_rep} and the preceding discussion conclude the proof of \cref{thm:term_embedding} via a union bound over the $O^*(1)$ rounds of the Ellipsoid algorithm. 
\qed

\setcounter{AlgoLine}{0}
\begin{algorithm}[H]
        \KwIn{Data Structure $\lbrb{\mc{D}_{\mc{T}^\prime, i}}_{\mc{T}^\prime \in \mc{T}, i \in \{0\}\cup [p]}$, Tree $\mc{T}$, Query $q$, Candidate $v$, AANN Output $(\hat{x}, \mc{T}^\prime)$, Tolerance $\eps^{\dagger}$}
        \For {$i = 0 \dots, p$}{
            \If {$x = \fixedscalequery (\mc{D}_{\mc{T}^\prime, i}, q, v, \hat{x}, \eps^\dagger) \neq \mrm{FAIL}$}{
                 \Return{$x$}
            }
        }
        \Return{$\mrm{FAIL}$}
    \caption{$\multiscalequery (\lbrb{\mc{D}_{\mc{T}^\prime, i}}_{\mc{T}^\prime \in \mc{T}, i \in \{0\}\cup [p]}, \mc{T}, q, v, (\hat{x}, \mc{T}^\prime), \eps^\dagger)$}
    \label{alg:multi_scale_query}
\end{algorithm}


\section{Adaptive Approximate Nearest Neighbor}
\label{sec:aann}

In this section, we prove the following theorem regarding the existence of adaptive algorithms for approximate nearest neighbor search based on adapting ideas from \cite{adaptiveds} to the nearest neighbor to near neighbor reduction. These data structures will play a crucial role in designing algorithms to compute terminal embeddings. Note again that the probability of success only depends on the random choices made by data structure at \emph{query-time} which may be made arbitrarily high by repetition assuming successful instantiation of the data structure. Also, the second property of the tuple returned by the data structure is irrelevant for computing an approximate nearest neighbor but plays a crucial part in our algorithm for computing terminal embeddings. 
\RestateGo{\restateaannmain}
\aannmain*

For to establish \cref{thm:aann_main}, we require a construction of a partition tree. Recall again the definitions relevant to a partition tree from \cref{ssec:gen_red}. The first is the partitioning of a data points into connected components of a graph constructed by thresholding the distances between points in the dataset.

\gpcc*

The second is the notion of refinement between partitions.

\partref*

Next, define $\rmed(X)$ as:
\begin{equation*}
    \rmed(X) = \min \{r > 0: \exists C \in \cc (X, r) \text{ with } \abs{C} \geq n / 2\}.
\end{equation*}

The last is the definition of the partition tree itself.
\parttree*

We refer to the size of a subtree $\mc{T}^\prime = \lprp{Z, \lbrb{\mc{T}_C}_{C \in \clow}, \mc{T}_{\mathrm{rep}}, \clow, \chigh, \crep, \rhat} \in \mc{T}$ by the following recursive formula:
\begin{equation*}
    \Size(\mc{T}) = \abs{Z} + \sum_{C \in \clow} \Size(\mc{T}_C) + \Size (\mc{T}_{\mathrm{rep}}).
\end{equation*} 

We state a result proved in \cref{sec:cons_tree} which allows the construction of a partition tree in time near-linear in the dataset size.
\RestateInit{\restateconstree}
\begin{restatable}{lemma}{constree}
    \label{lem:cons_tree}\RestateRemark{\restateconstree}
    Let $X = \{x_i\}_{i = 1}^n \subset \mb{R}^d$ and $\delta \in (0, 1)$. Then, \cref{alg:cons_tree} when given $X$, $\delta$ and $n$, runs in time $\wt{O} (nd \log (1 / \delta))$ and constructs, $\mc{T}$, satisfying:
    \begin{gather*}
        \forall \mc{T}^\prime = \lprp{Z, \lbrb{\mc{T}_C}_{C \in \clow}, \mc{T}_{\mathrm{rep}}, \clow, \chigh, \crep, \rhat} \in \mc{T}: \\
        \cc (Z, 1000n^2\rhat) \sqsubseteq \chigh \sqsubseteq \cc (Z, \rhat) \sqsubseteq \cc (Z, \rmed) \sqsubseteq \cc \lprp{Z, \frac{\rhat}{10n}} \sqsubseteq \clow \sqsubseteq \cc \lprp{Z, \frac{\rhat}{1000n^3}}
    \end{gather*}
    with probability at least $1 - \delta$. Furthermore, as a consequence we have for all $n \geq 3$:
    \begin{equation*}
        \Size(\mc{T}) \leq C n\log n,\ \forall C \in \clow \cup \{C_{\mathrm{rep}}\}: \abs{C} \leq \frac{\abs{Z}}{2} \text{ and } \rmed \leq \rhat \leq n\rmed.
    \end{equation*}
\end{restatable}

Through the rest of the section, we prove \cref{thm:aann_main}. In \cref{ssec:cons_aann}, we overview the construction of our data structure given a partition tree constructed by say, \cref{lem:cons_tree} and in \cref{ssec:query_aann}, we show how to query the data structure where we show that it is sufficient to construct terminal embeddings for the set of points in the node we terminate our traversal of the Partition Tree.

Intuitively, through a union bound and a repetition argument, one would expect to show that for $\wt{\Omega} (d)$ independently instantiated approximate near neighbor data structures, most of them (say $90\%$) will correctly answer a near neighbor query for \emph{any} (even perhaps, adaptively chosen based on the instantiations of the data structures) query $q \in \R^d$. However, care must be taken to choose the grid in the covering number argument to avoid dependence on the aspect ratio of the dataset. Additionally, we discretize the input to one of the points in the chosen grid as we assume no continuity properties on the near neighbor data structures. Furthermore, we require stronger additional properties that enable these results to be used to construct the terminal embeddings in \cref{sec:term}. These technical considerations complicate the proof as it requires interleaving the choice of the grid with the design of the partition tree and its use in the nearest neighbor to near neighbor reduction. The remainder of this section illustrates and circumvents these difficulties.

\subsection{Constructing the Data Structure}
\label{ssec:cons_aann}

In this subsection, we describe how the data structure for our adaptive nearest neighbor algorithm is constructed. The procedure is outlined in \cref{alg:cons_aann}. The procedure takes as input a partition tree produced by $\cref{alg:cons_tree}$ satisfying the conclusion of \cref{lem:cons_tree}, a failure probability $\delta \in (0, 1)$ and the total number of points $n$.

\setcounter{AlgoLine}{0}
\begin{algorithm}[H]
        \KwIn{Partition Tree $\mc{T}$, Failure Probability $\delta$, Total Number of points $n$}
         $\delta^\dagger \gets \frac{\cprob \delta}{(nd)^{10}}, \gamma \gets \stepsize$\;
        \For {$\mc{T}^\prime = \lprp{Z, \lbrb{\mc{T}_C}_{C \in \clow}, \mc{T}_{\mathrm{rep}}, \clow, \chigh, C_{\mathrm{rep}}, \rhat} \in \mc{T}$}{
             $\rlow \gets \frac{\rhat}{\rlhc\cdot (nd)^{10}}, \rhigh \gets \rlhc \cdot (nd)^{10} \rhat$\;
             $l \gets \ceil*{\frac{\log \rhigh / \rlow}{\gamma}}$, $s \gets \crep (d + \log l / \delta^\dagger) \log (nd)$\;
             For $i \in \{0\} \cup [l], j \in [s]$, let $\mc{D}_{i,j}$ be i.i.d ANN data structures with $(Z, (1 + \gamma)^i \rlow)$\;
             Let $\mc{D}_{\mc{T}^\prime} = \{\mc{D}_{i, j}\}_{i \in \{0\} \cup [l], j \in [s]}$\;
        }
        \Return{$\{\mc{D}_{\mc{T}^\prime}\}_{\mc{T}^\prime \in \mc{T}}$}
    \caption{$\mathrm{ConstructAANN} (\mc{T}, \delta, n)$}
    \label{alg:cons_aann}
  \end{algorithm}



To state the correctness guarantees on the data structure, we will require most of the ANN data structures to be accurate for an appropriately chosen discretization of $\mb{R}^d$. The reason for this will become clear when defining the operation of the query procedure on the data structure produced by the algorithm. To construct the discrete set of points, consider a particular node $\mc{T}^\prime \in \mc{T}$. Let $\mc{T}^\prime = \lprp{Z, \lbrb{\mc{T}_C}_{C \in \clow}, \mc{T}_{\mathrm{rep}}, \clow, \chigh, C_{\mathrm{rep}}, \rhat} \in \mc{T}$ and $\mc{G}(\nu)$ be the discrete subset of $\mb{R}^d$ whose coordinates are integral multiples of $\nu = \frac{\gamma}{1000(nd)^{20}} \cdot \rlow$ where $\rlow = \frac{\rhat}{\rlhc \cdot (nd)^{10}}$ and $\gamma = \stepsize$ as in \cref{alg:cons_aann}. Again, letting $\rhigh = \rlhc \cdot (nd)^{10} \rhat$ as in \cref{alg:cons_aann}, we define the grid of points corresponding to $\mc{T}^\prime$ as follows:
\begin{equation*}
    \mc{H} \lprp{\mc{T}^\prime} = \bigcup_{x \in Z} \mb{B} (x, 10^6 \cdot (nd)^{20} \cdot \rhigh) \cap \mc{G} (\nu).
\end{equation*}
That is $\mc{H}$ corresponds to the set of points in $\mc{G}$ within $10^6 \cdot (nd)^{20} \rhigh$ of some point in $Z$. Finally, the set of points where we would like to ensure the correctness of our procedure is given by:
\begin{equation}
    \mc{J} = \bigcup_{\mc{T}^\prime \in \mc{T}} \mc{H} (\mc{T}^\prime). \tag{AANN-Grid} \label{eq:aann_grid}
\end{equation}
We now state the main result concerning the correctness of \cref{alg:cons_aann}:
\begin{lemma}
    \label{lem:cons_aann}
    Let $X = \lbrb{x_i}_{i = 1}^n \subset \mb{R}^d$, $\delta \in (0, 1)$ and $\mc{T}$ be a Partition Tree of $X$ satisfying the conclusion of \cref{lem:cons_tree}. Then, \cref{alg:cons_aann} when run with input $X$, $\delta$ and $n$, returns $\mc{D} = \lbrb{\mc{D}_{\mc{T}^\prime}}_{\mc{T}^\prime \in \mc{T}}$ satisfying:
    \begin{equation*}
        \forall \mc{D}_{\mc{T}^\prime} = \{\mc{D}_{i,j}\}_{i \in \{0\} \cup [l], j \in [s]} \in \mc{D},\, i \in \{0\} \cup [l],\, z \in \mc{J}: \sum_{j = 1}^s \bm{1} \lbrb{\mc{D}_{i,j} \text{ answers } z \text{ correctly}} \geq 0.95s
    \end{equation*}
    with probability at least $1 - \delta$. Furthermore, the space complexity of $\mc{D}$ is $O^*(n^{1 + \ru} (d + \log 1 / \delta))$.
\end{lemma}
\begin{proof}
    We start by bounding the amount of space occupied by the data structure. For the nearest neighbor data structures instantiated, note that the space utilization of a single data structure with $n$ points scales as $O^*(n^{1 + \ru + f(n)})$ where $f(n) = o(1)$. Let $T$ be such that $f(n) \leq 1$ for all $n \geq T$ and let $M = \max_{i \in [T]} f(i)$. To account for the space occupied by the data structure, first consider nodes in the tree with less than $\log n$ points. Let $\mc{D} = \lprp{Z, \lbrb{\mc{T}_C}_{C \in \clow}, \mc{T}_{\mathrm{rep}}, \clow, \chigh, \crep, \rhat}$ be such a node. For this node, the space occupied by the nearest neighbor data structures is at most $O^*(\abs{Z}^{1 + \ru + M} (d + \log 1 / \delta))$ as we instantiate $O^*(d + \log 1 / \delta)$ nearest neighbor data structures in each node. Since, $\abs{Z} \leq \log n$ and there are at most $O^*(n\log n)$ many of these nodes (\cref{lem:cons_tree}), the total amount of space occupied by these nodes is at most $O^*(n(d + \log 1 / \delta))$. Now, consider the alternate case where $\abs{Z} \geq \log n$. Through similar reasoning, the total amount of space occupied by such a node is $O^*(\abs{Z}^{1 + \ru + f(\abs{Z})} (d + \log 1 / \delta))$. Now, summing over all the nodes in the tree and by using the convexity of the function $f(x) = x^{1 + \rho}$ for $\rho > 0$, we get that the total amount of space occupied by nodes with more than $\log n$ points is at most $O^*(n^{1 + \ru}(d + \log 1 / \delta))$. This completes the bound on the space complexity of the data structure produced by \cref{alg:cons_aann}.

    We will now establish the correctness guarantees required of $\mc{D}$. To bound the size of $\mc{J}$, first consider a single $\mc{H} (\mc{T}^\prime)$ in the definition of $\mc{J}$. For a single term in the definition of $\mc{H} (\mc{T}^\prime)$, $\mc{V} = \mb{B} \lprp{x, 10^6 \cdot (nd)^{20} \cdot \rhigh} \cap \mc{G}(\nu)$, note that $\mc{V}$ is a $\nu$ packing of $\mb{B} \lprp{x, (10^6 + 1) \cdot (nd)^{20} \cdot \rhigh}$. Therefore, from standard bounds on packing and covering numbers and the definitions of $\nu$ and $\rhigh$ in terms of $\rhat$, we get that $\abs{\mc{V}} \leq (nd)^{O(d)}$ \cite[Section 4.2]{vershynin}. By taking a union bound over the $O^*(n\log n)$ nodes in the tree and the at most $n$ points in each node, we get that $\abs{\mc{J}} \leq (nd)^{O(d)}$. Now, for a particular $z \in \mc{J}$ and a particular node $\mc{T}^\prime \in \mc{T}$ with data structure $\mc{D}_{\mc{T}^\prime} = \{\mc{D}_{i,j}\}_{i \in \{0\} \cup [l], j \in [s]}$, the probability that $\mc{D}_{i, j}$ incorrectly answers the Approximate Near Neighbor query is at most $0.01$. Therefore, we have by Hoeffding's inequality that the probability that more than $0.05 s$ of the $\mc{D}_{i,j}$ answer $z$ incorrectly is at most $\delta / (10 \cdot \abs{\mc{T}} \cdot \abs{\mc{J}})$ from our setting of $s$ and our bounds on $\abs{\mc{J}}$ and $\abs{\mc{T}}$. A union bound over $\mc{J}$ and the nodes of the tree establishes the lemma.
\end{proof}

\subsection{Querying the Data Structure and Proof of Theorem \ref{thm:aann_main}}
\label{ssec:query_aann}

The procedure to query the data structure returned by \cref{alg:cons_aann} is described in \cref{alg:query_aann}. The procedure takes as inputs the partition tree, $\mc{T}$, and the data structure output by \cref{alg:cons_aann} which has a separate data structure for each $\mc{T}^\prime \in \mc{T}$. The procedure recursively explores the nodes of $\mc{T}$ starting at the root and moving down the tree and stops when the approximate nearest neighbor found is within $poly(n) \rhat$ of the nodes in the tree. In addition, in anticipation of its application towards computing terminal embeddings, we show that it is sufficient to construct a terminal embedding for the set of points in the node the algorithm terminates in. We now state the main lemma of this subsection which shows both that the data point $x$ returned by \cref{alg:query_aann} is an approximate nearest neighbor of $q$ as well as establishing that it suffices to construct a terminal embedding for data points in the node of the Partition tree the algorithm terminates in.

\begin{lemma}
    \label{lem:query_aann}
    Let $X = \{x_i\}_{i = 1}^n \subset \mb{R}^d$, $\mc{T}$ be a valid Partition Tree of $X$ and $\mc{D} = \lbrb{\mc{T}, \lbrb{\mc{D}_{\mc{T}^\prime}}_{\mc{T}^\prime \in \mc{T}}}$ be a data structure satisfying the conclusion of \cref{lem:cons_aann}. Then, \cref{alg:query_aann}, when run with inputs $\mc{D}$ and any $q \in \mb{R}^d$ returns a tuple $(\xhat, \tres)$ satisfying:
    \begin{equation*}
        \norm{q - \xhat} \leq (1 + o^*(1)) c \min_{x \in X} \norm{q - x} \text{ and } \xhat \in \tres.
    \end{equation*}
    Furthermore, let $\mc{Y} = \{y_i\}_{i = 1}^n \subset \mb{R}^k$ satisfying for some $\eps^\dagger \in \lprp{\frac{1}{\sqrt{d}}, 1}$:
    \begin{equation*}
        \forall i, j \in [n]: (1 - \eps^\dagger) \norm{x_i - x_j} \leq \norm{y_i - y_j} \leq (1 + \eps^\dagger) \norm{x_i - x_j}
    \end{equation*}
    and for $\tres = \lprp{Z, \lbrb{\mc{T}_C}_{C \in \clow}, \mc{T}_{\mathrm{rep}}, \clow, \chigh, \crep, \rhat}$, let $y \in \mb{R}^k$ satisfy for some $\eps^\ddagger \in [\eps^\dagger, 1)$:
    \begin{equation*}
        \forall x_i \in Z: (1 - \eps^\ddagger) \norm{q - x_i} \leq \norm{y - y_i} \leq (1 + \eps^\ddagger) \norm{q - x_i}.
    \end{equation*}
    Then, we have that:
    \begin{equation*}
        \forall x_i \in X: \lprp{1 - \lprp{1 + o^*(1)}\eps^\ddagger} \norm{q - x_i} \leq \norm{y - y_i} \leq \lprp{1 + \lprp{1 + o^*(1)} \eps^\ddagger} \norm{q - x_i}
    \end{equation*}
    and if $\abs{Z} > 1$:
    \begin{equation*}
        \Omega \lprp{\frac{1}{(nd)^{10}} \rhat} \leq \norm{q - \xhat} \leq O((nd)^{10} \rhat).
    \end{equation*}
    Additionally, \cref{alg:query_aann} runs in time $O^*(n^{\rc} (d + \log 1 / \delta))$. 
\end{lemma}

\begin{proof}
    We first set up some notation. Note that \cref{alg:query_aann} traverses the partition tree, $\mc{T}$, using the data structure $\mc{D}$ defined in \cref{alg:cons_aann}. Let $\mc{T}^{(0)}, \dots, \mc{T}^{(K)}$ denote the sequence of nodes traversed by the algorithm with $\mc{T}^{(i)} = \lprp{Z^{(i)}, \lbrb{\mc{T}_C}_{C \in \clow^{(i)}}^{(i)}, \mc{T}_{\mathrm{rep}}^{(i)}, \clow^{(i)}, \chigh^{(i)}, C_{\mathrm{rep}}^{(i)}, \rhat^{(i)}}$. Note that $K \leq \ceil{\log n} + 1$ as the number of data points drops by at least a factor of $2$ each time the algorithm explores a new node. To prove the first claim regarding the correctness of $\xhat$, let $r^* = \min_{x \in X} \norm{q - x}$. We now have the following claim:
\begin{claim}
    \label{clm:query_depth}
    For all $i \in \{0, \dots, K\}$, we have:
    \begin{equation*}
        \exists z \in Z^{(i)}: \norm{q - z} \leq \lprp{1 + \frac{1}{(10nd)^5}}^{i} r^*.
    \end{equation*}
\end{claim}
\begin{subproof}
    We will prove the claim via induction on $i$. The base case where $i = 0$ is trivially true. Now, suppose that the claim holds true for all nodes up to $\mc{T}^{(k)}$ and we wish to establish the claim from $\mc{T}^{(k + 1)}$. For the algorithm to reach node $\mc{T}^{(k + 1)}$ from node $\mc{T}^{(k)}$, one of the following two cases must have occurred:
    \begin{enumerate}
        \item Either $\mc{T}^{(k + 1)} = \mc{T}^{(k)}_{\mrm{rep}}$ or 
        \item $\mc{T}^{(k + 1)} = \mc{T}_C$ for some $C \in \clow^{(k)}$
    \end{enumerate}
    Now, let $\wt{q}^{(k)}$ be the discretization of $q$ when $\mc{T}^{(k)}$ is being processed. We first handle the case where $\wt{q}^{(k)} \notin \mc{J}$, we have by the triangle inequality:
    \begin{equation*}
        r^* \geq \min_{x \in X} \norm{\wt{q}^{(k)} - x} - \norm{q - \wt{q}^{(k)}} \geq 5\cdot 10^5 \cdot (nd)^{20} \cdot \rhigh
    \end{equation*}
    and the first case occurs. Now assume that $\wt{q}^{(k)} \in \mc{J}$. If the first case occurs, we again have by the triangle inequality:
    \begin{equation*}
        \min_{z \in Z^{(k)}} \norm{q - z} \geq \min_{z \in Z^{(k)}} \norm{\wt{q}^{(k)} - z} - \norm{q - \wt{q}^{(k)}} \geq 0.9 \rhigh
    \end{equation*}
    by the fact that \cref{alg:query_aann} recurses on $\crep^{(k)}$ and the conclusion of \cref{lem:cons_aann} which ensures that $\wt{q}^{(k)}$ does not have a neighbor in $Z^{(k)}$ within a distance of $\rhigh^{(k)}$. Now, let $z^*$ be the closest neighbor to $q$ in $Z^{(k)}$; that is, $z^* = \argmin_{z \in Z^{(k)}} \norm{q - z}$. We know that there exists $C \in \chigh$ such that $z^* \in C$. Furthermore, we have by the triangle inequality and the fact that $\cc (Z^{(k)}, 1000n^2\rhat^{(k)}) \sqsubseteq \chigh^{(k)}$ (\cref{lem:cons_tree}) that $\norm{z - z^*} \leq 1000n^3 \rhat^{(k)}$ for all $z \in C$. Note that $C$ has representative $\hat{z}$ in the construction of $\crep^{(k)}$. For $\hat{z}$, we have
    \begin{equation*}
        \frac{\norm{q - \hat{z}}}{\norm{q - z^*}} \leq 1 + \frac{\norm{\hat{z} - z^*}}{\norm{q - z^*}} \leq 1 + \frac{1000n^3 \rhat}{0.9 \rhigh} \leq \lprp{1 + \frac{1}{(10nd)^5}}.
    \end{equation*}
    Along with the induction hypothesis, the above fact concludes the proof of the claim in this case. Finally, for the second case, again let $z^* = \argmin_{z \in Z^{(k)}} \norm{q - z}$ and $C \in \clow^{(k)}$ such that $z^* \in C$. From the definition of $\clow^{(k)}$ and \cref{lem:cons_tree}, we also know for all $z \in Z^{(k)} \setminus C$, $\norm{z^* - z} \geq \frac{\rhat}{(1000 n^3)}$ as $\clow \sqsubseteq \cc \lprp{Z^{(k)}, \frac{\rhat}{1000n^3}}$. We have by the triangle inequality for all $z \in Z^{(k)} \setminus C$:
    \begin{align*}
        \norm{\wt{q}^{(k)} - z} &\geq \norm{z - z^*} - \norm{q - z^*} - \norm{q - \wt{q}^{(k)}} = \norm{z - z^*} - \min_{z \in Z^{(k)}} \norm{q - z} - \norm{q - \wt{q}^{(k)}} \\
        &\geq \norm{z - z^*} - \min_{z \in Z^{(k)}} (\norm{\wt{q}^{(k)} - z} + \norm{\wt{q}^{(k)} - q}) - \norm{q - \wt{q}^{(k)}} \\
        &\geq \frac{\rhat^{(k)}}{1000 n^3} - c \rlow^{(k)} - 2\sqrt{d} \nu^{(k)} \geq c(10nd)^7 \rlow^{(k)}.
    \end{align*}
    Therefore, all $x \in Z^{(k)}$ such that $\norm{\wt{q} - x} \leq c \rlow^{(k)}$ must belong to $C$. Consequently, we recurse on $\mc{T}_C$ with $z^* \in C \in \clow$ which establishes the inductive hypothesis in this case as well.
\end{subproof}

To finish the proof of the first claim, let $\mc{T}^{(K)} = \lprp{Z, \lbrb{\mc{T}_C}_{C \in \clow}, \mc{T}_{\mathrm{rep}}, \clow, \chigh, C_{\mathrm{rep}}, \rhat}$ with associated data structure $\mc{D}_{\mc{T}^{(K)}} = \{\mc{D}_{i,j}\}_{i \in \{0\} \cup [l], j \in [s]}$. If $\abs{Z} = 1$, a direct application of \cref{clm:query_depth} establishes the lemma. Alternatively, we must have $0 < i^* \leq l$ when $\mc{T}^{(K)}$ is being processed by \cref{alg:query_aann}. From the guarantees on $\mc{D}_{i,j}$ (\cref{lem:cons_aann}), we get $\min_{z \in Z} \norm{z - \wt{q}} > (1 + \gamma)^{i^* - 1} \rlow$. Letting $x \in Z$  such that $\norm{x - \wt{q}} \leq c(1 + \gamma)^{i^*} \rlow$ returned by \cref{alg:query_aann}, we get by another application of the triangle inequality:
\begin{align*}
    \frac{\norm{q - x}}{\min_{z \in Z} \norm{q - z}} &\leq \frac{\norm{\wt{q} - x} + \norm{\wt{q} - q}}{\min_{z \in Z} \norm{\wt{q} - z} - \norm{\wt{q} - q}} \leq \frac{c (1 + \gamma)^{i^*} \rlow + \sqrt{d} \nu}{(1 + \gamma)^{i^* - 1} \rlow - \sqrt{d} \nu} \\
    &\leq c(1 + \gamma) + \frac{2 c \sqrt{d} \nu}{(1 + \gamma)^{i^* - 1} \rlow - \sqrt{d} \nu} \leq c(1 + \gamma)^2
\end{align*}
where the first inequality is valid as $\min_{z \in Z} \norm{\wt{q} - z} - \norm{\wt{q} - q} > 0$ from our setting of $\nu$ and the condition on $\min_{z \in Z} \norm{q - z}$ and the final inequality similarly follows from our setting of $\nu$. Another application of \cref{clm:query_depth} with the fact that $K \leq \ceil*{\log n} + 1$, establishes the first claim of the lemma. 

To prove the second claim of the lemma, we prove an analogous claim to \cref{clm:query_depth} for terminal embeddings:
\begin{claim}
    \label{clm:qdepth_term}
    We have for all $i \in \{0, \dots, K\}$:
    \begin{equation*}
        \forall z_j \in Z^{{i}}: \lprp{1 - \lprp{1 + \frac{1}{(10nd)^{5}}}^{K - i} \eps^{\ddagger}} \norm{q - x_j} \leq \norm{y - y_j} \leq \lprp{1 + \lprp{1 + \frac{1}{(10nd)^{5}}}^{K - i} \eps^{\ddagger}} \norm{q - x_j}.
    \end{equation*}
\end{claim}
\begin{subproof}
    We will prove the claim by reverse induction on $i$. For $i = K$, the claim is implied by the assumptions on $y$. Now, suppose the claim holds for $i = k + 1$ and we wish to establish the claim for $i = k$. As in the proof of \cref{clm:query_depth}, we have two cases when $\mc{T}^{(k)}$ is being processed:
    \begin{enumerate}
        \item Either $\mc{T}^{(k + 1)} = \mc{T}^{(k)}_{\mrm{rep}}$ or 
        \item $\mc{T}^{(k + 1)} = \mc{T}^{(k)}_C$ for some $C \in \clow^{(k)}$.
    \end{enumerate}
    As in \cref{clm:query_depth}, when the first case occurs, we have $\min_{z \in Z^{(k)}} \norm{q - z} \geq 0.9\rhigh^{(k)}$. Now, for any $z_i \in Z^{(k)}$, let $z_j \in \crep^{(k)}$ such that $z_i, z_j \in C \in \chigh^{(k)}$. Note that we must have by the triangle inequality and the fact that $\cc(Z^{(k)}, 1000n^2\rhat^{(k)}) \sqsubseteq \chigh$ that $\norm{z_i - z_j} \leq 1000n^3\rhat^{(k)}$. From the fact that $\crep^{(k)} = Z^{(k + 1)}$, the inductive hypothesis and the assumption on $y_i, y_j$ from the lemma, we get:
    \begin{align*}
        \norm{y - y_i} &\geq \norm{y - y_j} - \norm{y_j - y_i} \geq \lprp{1 - \lprp{1 + \frac{1}{(10nd)^5}}^{K - (k + 1)}\eps^\ddagger} \norm{q - z_j} - 2000n^3\rhat^{(k)} \\
        &\geq \lprp{1 - \lprp{1 + \frac{1}{(10nd)^5}}^{K - (k + 1)}\eps^\ddagger} (\norm{q - z_i} - \norm{z_i - z_j}) - 2000n^3\rhat^{(k)} \\
        &\geq \lprp{1 - \lprp{1 + \frac{1}{(10nd)^5}}^{K - (k + 1)}\eps^\ddagger} \norm{q - z_i} - 4000 n^3 \rhat^{(k)} \\
        &\geq \lprp{1 - \lprp{1 + \frac{1}{(10nd)^5}}^{K - (k + 1)}\eps^\ddagger} \norm{q - z_i} - \frac{1}{(10nd)^5} \eps^{\ddagger} \norm{q - z_i} \\
        &\geq \lprp{1 - \lprp{1 + \frac{1}{(10nd)^5}}^{K - k}\eps^\ddagger} \norm{q - z_i}
    \end{align*}
    where the last inequality follows from the fact that $(1 + a)^b \geq (1 + a)^{b - 1} + a$ for $a \geq 0, b \geq 1$. For the other direction, we have by a similar calculation:
    \begin{align*}
        \norm{y - y_i} &\leq \norm{y - y_j} + \norm{y_j - y_i} \leq \lprp{1 + \lprp{1 + \frac{1}{(10nd)^5}}^{K - (k + 1)}\eps^\ddagger} \norm{q - z_j} + 2000n^3\rhat^{(k)} \\
        &\leq \lprp{1 + \lprp{1 + \frac{1}{(10nd)^5}}^{K - (k + 1)}\eps^\ddagger} (\norm{q - z_i} + \norm{z_i - z_j}) + 2000n^3\rhat^{(k)} \\
        &\leq \lprp{1 + \lprp{1 + \frac{1}{(10nd)^5}}^{K - (k + 1)}\eps^\ddagger} \norm{q - z_i} + 4000 n^3 \rhat^{(k)} \\
        &\leq \lprp{1 + \lprp{1 + \frac{1}{(10nd)^5}}^{K - (k + 1)}\eps^\ddagger} \norm{q - z_i} + \frac{1}{(10nd)^5} \eps^{\ddagger} \norm{q - z_i} \\
        &\leq \lprp{1 + \lprp{1 + \frac{1}{(10nd)^5}}^{K - k}\eps^\ddagger} \norm{q - z_i}.
    \end{align*}
    This establishes the claim in the first case. For the second case, let $z^* = \argmin_{z \in Z^{(k)}} \norm{q - z}$. As in the proof of \cref{clm:qdepth_term}, we have $z^* \in Z^{(k + 1)} = C \in \clow^{(k)}$ and for all $z \in Z^{(k)} \setminus C$:
    \begin{align*}
        \norm{q - z} &\geq \norm{z^* - z} - \norm{q - z^*} = \norm{z^* - z} - \min_{z \in Z^{(k)}}\norm{q - z} \\
        &\geq \norm{z^* - z} - \min_{z \in Z^{(k)}} \norm{\wt{q}^{(k)} - z} - \norm{\wt{q}^{(k)} - q} \\
        &\geq \frac{\rhat^{(k)}}{1000 n^3} - c \rlow^{(k)} - \sqrt{d} \nu^{(k)} \geq c(10nd)^7 \rlow^{(k)}
    \end{align*}
    and furthermore, we have:
    \begin{equation*}
        \norm{q - z^*} = \min_{z \in Z^{(k)}} \norm{q - z^*} \leq \min_{z \in Z^{(k)}} \norm{\wt{q}^{(k)} - z^*} + \norm{\wt{q}^{(k)} - q} \leq c\rlow^{(k)} + \rlow^{(k)} \leq 2c\rlow^{(k)}.
    \end{equation*}
    Now, the claim is already established for all $x_i \in Z^{(k + 1)}$. For $z_i \in Z^{(k)} \setminus Z^{(k + 1)}$, letting $y^* = y_j$ for $x_j = z^*$:
    \begin{align*}
        \norm{y_i - y} &\leq \norm{y_i - y_j} + \norm{y - y_j} \leq (1 + \eps^\dagger) \norm{x_i - x_j} + \lprp{1 + \lprp{1 + \frac{1}{(10nd)^{5}}}^{K - (k + 1)} \eps^{\ddagger}} \norm{q - z^*} \\
        &\leq (1 + \eps^\dagger) \norm{x_i - x_j} + 3c\rlow^{(k)} \leq (1 + \eps^\dagger) \norm{x_i - q} + (1 + \eps^\dagger) \norm{q - x_j} + 3 c\rlow^{(k)} \\
        &\leq \lprp{1 + \lprp{1 + \frac{6c\rlow^{(k)}}{\eps^\dagger \cdot \norm{x_i - q}}}\eps^\dagger} \norm{x_i - q} \leq \lprp{1 + \lprp{1 + \frac{1}{(10nd)^5}} \eps^\dagger} \norm{x_i - q}
    \end{align*}
    where the third inequality is due to the fact that $K \leq \ceil{\log n} + 1$ and the fact that $(1 + a)^b \leq e^{ab}$ for $a, b \geq 0$ and the final two inequalities follow from the upper bound on $\norm{q - z^*}$ and the lower bound on $\norm{x_i - q}$ established previously. For the lower bound, we have by a similar computation:
    \begin{align*}
        \norm{y_i - y} &\geq \norm{y_i - y_j} - \norm{y - y_j} \geq (1 - \eps^\dagger) \norm{x_i - x_j} - \lprp{1 + \lprp{1 + \frac{1}{(10nd)^{5}}}^{K - (k + 1)} \eps^{\ddagger}} \norm{q - z^*} \\
        &\geq (1 - \eps^\dagger) \norm{x_i - x_j} - 3c\rlow^{(k)} \geq (1 - \eps^\dagger) \norm{x_i - q} - (1 - \eps^\dagger) \norm{q - x_j} - 3 c\rlow^{(k)} \\
        &\geq \lprp{1 - \lprp{1 + \frac{6c\rlow^{(k)}}{\eps^\dagger \cdot \norm{x_i - q}}}\eps^\dagger} \norm{x_i - q} \geq \lprp{1 - \lprp{1 + \frac{1}{(10nd)^5}} \eps^\dagger} \norm{x_i - q}.
    \end{align*}
    The assumption that $\eps^{\ddagger} \geq \eps^\dagger$ finishes the proof of the claim by induction.
\end{subproof}
\cref{clm:qdepth_term} establishes the second conclusion of lemma from the fact that $K \leq \ceil{\log n} + 1$. When $\abs{Z} > 1$ in $\tres$, we get by the triangle inequality, the fact that $\norm{\wt{q} - \xhat} \geq \rlow$ and the definition of~$\wt{q}$ in \cref{alg:query_aann} that:
\begin{equation*}
    \Omega \lprp{\frac{1}{(nd)^{10}} \rhat} \leq \norm{q - \xhat} \leq O((nd)^{10} \rhat).
\end{equation*}
Finally, for the bound on the runtime, note that \cref{alg:query_aann} queries at most $K l (s + 1)$ near neighbor data structures built with at most $n$ points. Therefore, each query takes time at most $O^*(n^{\rc})$. This proves our bound on the runtime of \cref{alg:query_aann}. 
\end{proof}

\paragraph{Proof of \cref{thm:aann_main}:} \cref{lem:query_aann,lem:cons_aann,lem:cons_tree} now establish \cref{thm:aann_main} barring the decoupling of $d$ and the $n^{\rc}$ term in the runtime guarantee. To do this, note that we may simply use the Median-JL data structure (see \cref{thm:jl_rep} in \cref{sec:med_jl}) by instantiating $l = \Ot(d)$ many JL-sketches and for each of them, instantiate an adaptive approximate nearest neighbor data structure in the low dimensional space. At query time, we simply pick $\Omega (\log n)$ of these sketches uniformly at random, query each of the corresponding nearest neighbor data structures with the projection of $q$ and return the best answer. This yields the improved runtimes from \cref{thm:aann_main}.
\qed 

\setcounter{AlgoLine}{0}
\begin{algorithm}
        \KwIn{Data Structure $\mc{D} = \lbrb{\mc{T}, \lbrb{\mc{D}_{\mc{T}^\prime}}_{\mc{T}^\prime \in \mc{T}}}$, Query $q$}
         $\mc{T}_{\mrm{cur}} \gets \mc{T}$\;
        \While {$\mrm{TRUE}$}{
             Let $\mc{T}_{\mrm{cur}} = \lprp{Z, \lbrb{\mc{T}_C}_{C \in \clow}, \mc{T}_{\mathrm{rep}}, \clow, \chigh, C_{\mathrm{rep}}, \rhat} \in \mc{T}$\;
            \If {$\abs{Z} = 1$}{
                 \Return{$(x, \mc{T}_{\mrm{cur}})$ for $x \in Z$}
            }
            
             $\nu \gets \frac{\gamma}{1000 (nd)^{20}} \cdot \rlow$\;
             $\wt{q}_i \gets \floor*{\frac{q_i}{\nu}} \nu$\;
             $i^* \gets \min\lbrb{i: \exists j \text{ s.t }\mc{D}_{i, j} (\wt{q}) \text{ returns } x \in Z \text{ satisfying } \norm{\wt{q} - x} \leq c (1 + \gamma)^i \rlow}$\;
            
            \If {$0 < i^* \leq l$}{
                 Let $x \in Z$ be such that $\norm{x - \wt{q}} \leq c(1 + \gamma)^{i^*} \rlow$\;
                 \textbf{Return:} $(x, \mc{T}_{\mrm{cur}})$\;
                }
            \ElseIf{$i^* = 0$}{
                 Let $x \in Z$ be such that $\norm{x - \wt{q}} \leq c \rlow$\;
                 Let $C \in \clow$ be such that $x \in C$\;
                 $\mc{T}_{\mrm{cur}} \gets \mc{T}_C$\;
                }
            \ElseIf{$i^* = \infty$}{
                 $\mc{T}_{\mrm{cur}} \gets \mc{T}_{\mrm{rep}}$\;
            }
        }
    \caption{$\mathrm{QueryAANN} (\mc{D}, q)$}
    \label{alg:query_aann}
  \end{algorithm}

\section{Median - JL}
\label{sec:med_jl}

In this section, we will prove the following lemma which enabled speeding up our algorithms by effectively projecting onto a low dimensional subspace essentially decoupling the terms that depend on $d$ and $n$.

\RestateGo{\restatejlrep}
\jlrep*

By setting $x = y$ in \ref{eq:ap_ip}, we see that the above theorem is a generalization of the standard Johnson-Lindenstrauss condition where in addition to maintaining distances between points, $\Pi$ is also required to approximately maintain relative inner products between the points in the augmented dataset. To begin the proof, we start by recalling the standard Johnson-Lindenstrauss lemma (see, for example, \cite{vershynin}). 

\begin{lemma}
    \label{lem:jl}
    Let $\Pi \in \mb{R}^{k \times d}$ be distributed according to $\Pi_{i, j} \overset{iid}{\thicksim} \mc{N} (0, 1 / k)$, $\delta \in (0, 1)$ and $k \geq C \frac{\log 1 / \delta}{\eps^2}$ for some absolute constant $C > 0$. Then, for any $v \in \mb{R}^d$, we have:
    \begin{equation*}
        \mb{P} \lbrb{(1 - \eps) \norm{v}^2 \leq \norm{\Pi v}^2 \leq (1 + \eps) \norm{v}^2} \geq 1 - \delta.
    \end{equation*}
\end{lemma}

We obtain via a union bound over all pairs of points in the dataset, $X$, the JL guarantee:

\begin{corollary}
    \label{cor:jl}
    Let $X = \{x_i\}_{i = 1}^n$, $\Pi \in \mb{R}^{k \times d}$ with $\Pi_{i,j} \overset{iid}{\thicksim} \mc{N}(0, 1 / k)$ with $k \geq C \frac{(\log n + \log 1 / \delta)}{\eps^2}$. Then, we have:
    \begin{equation*}
        \forall x_i, x_j \in X: (1 - \eps) \norm{x_i - x_j} \leq \norm{\Pi x_i - \Pi x_j} \leq (1 + \eps) \norm{x_i - x_j}
    \end{equation*}
    with probability at least $1 - \delta$.
\end{corollary}

A second corollary we will make frequent use of is the following where we show that $\Pi$ also approximately preserves inner products.

\begin{corollary}
    \label{cor:ip_pres}
    Let $\Pi \in \mb{R}^{k \times d}$ be distributed according to $\Pi_{i, j} \overset{iid}{\thicksim} \mc{N} (0, 1 / k)$ and $k \geq C \frac{\log 1 / \delta}{\eps^2}$. Then, for any $x,y \in \mb{R}^d$, we have:
    \begin{equation*}
        \mb{P} \lbrb{\abs*{\inp{\Pi x}{\Pi y} - \inp{x}{y}} \leq \eps \norm{x}\norm{y}} \geq 1 - \delta.
    \end{equation*}
\end{corollary}
\begin{proof}
    If either $x$ or $y$ are $0$, the conclusion follows trivially. Assume, $x,y \neq 0$. By scaling both sides by $\norm{x} \norm{y}$, we may assume that $\norm{x} = \norm{y} = 1$. We now have the following:
    \begin{gather*}
        \inp{\Pi x}{\Pi y} = \frac{1}{4} \lprp{\norm{\Pi (x + y)}^2 - \norm{\Pi (x - y)}^2} \\
        \inp{x}{y} = \frac{1}{4} \lprp{\norm{(x + y)}^2 - \norm{(x - y)}^2}.
    \end{gather*}
    By subtracting both equations, we get:
    \begin{align*}
        \abs{\inp{\Pi x}{\Pi y} - \inp{x}{y}} &\leq \frac{1}{4} \lprp{\abs*{\norm{\Pi (x + y)}^2 - \norm{x + y}^2} + \abs*{\norm{\Pi (x - y)}^2 - \norm{x - y}^2}}
    \end{align*}
    By the union bound, the triangle inequality and \cref{lem:jl}, we get with probability at least $1 - \delta$:
    \begin{gather*}
        \abs{\norm{\Pi (x + y)}^2 - \norm{x + y}^2} \leq \frac{\eps}{4} \norm{x + y}^2 \leq \eps \text{ and }\\
        \abs{\norm{\Pi (x - y)}^2 - \norm{x - y}^2} \leq \frac{\eps}{4} \norm{x - y}^2 \leq \eps.
    \end{gather*}
    The inequalities in the previous display imply the lemma.
\end{proof}

We start by establishing a simple lemma on the norms of the matrices $\Pi^i$.
\begin{lemma}
    \label{lem:norm_bounds}
    Assume the setting of \cref{thm:jl_rep}. Then, we have:
    \begin{equation*}
        \sum_{i = 1}^m \bm{1} \lbrb{\norm{\Pi^i}_F \leq O(\sqrt{d})} \geq 0.99m
    \end{equation*}
    with probability at least $1 - \delta / 4$.
\end{lemma}
\begin{proof}
    Let $W_i = \bm{1} \lbrb{\norm{\Pi^i}_F \leq O(\sqrt{d})}$. We have by an application of Bernstein's inequality that $\mb{P} (W_i = 1) \geq 0.995$ as $\sum_{j \in [k], l \in [d]} (\Pi^i_{j, l})^2$ is a $\chi^2$ random variable with mean $d$. Since, the $W_i$ are iid, the conclusion follows by an application of Hoeffding's inequality applied to the random variable $W = \sum_{i = 1}^m W_i$. 
\end{proof}

\begin{lemma}
    \label{lem:unit_rep}
    Assume the setting of \cref{thm:jl_rep}. Then, we have:
    \begin{equation*}
        \forall \norm{v} = 1: \sum_{i = 1}^m \bm{1} \lbrb{
                \begin{gathered}
                    (1 - \eps / 128) \leq \norm{\Pi^i v}^2 \leq (1 + \eps / 128) \\
                    \forall x, y \in X: \abs{\inp{\Pi^i v}{\Pi^i (x - y)} - \inp{v}{x - y}} \leq \frac{\eps}{128} \cdot \norm{x - y} \\
                    \norm{\Pi^i}_F \leq O(\sqrt{d}) 
                \end{gathered}
            }
            \geq 0.98m
    \end{equation*}
    with probability at least $1 - \delta / 2$. 
\end{lemma}
\begin{proof}
    Let $\mc{G}$ be a $\gamma$-net over $\mb{S}^{d - 1}$ with $\gamma = \frac{c}{(nd)^{10}}$ for some small enough constant $c$. Furthermore, we may assume $\abs{\mc{G}} \leq (nd)^{O(d)}$. Now, for $u \in \mc{G}$, let:
    \begin{equation*}
        W_i(u) = \bm{1} \lbrb{
            \begin{gathered}
                (1 - \eps / 256) \leq \norm{\Pi^i u}^2 \leq (1 + \eps / 256) \\ 
                \forall x, y \in X: \abs{\inp{\Pi^i u}{\Pi^i (x - y)} - \inp{u}{x - y}} \leq \frac{\eps}{256} \cdot \norm{x - y}
            \end{gathered}    
        }.
    \end{equation*}
    We have from \cref{cor:ip_pres}, that $\mb{P} (W_i(u) = 1) \geq 0.995$. Therefore, we have by an application of Hoeffding's inequality and a union bound over $\mc{G}$ that:
    \begin{equation*}
        \mb{P} \lbrb{\forall u \in \mc{G}: \sum_{i = 1}^m W_i (u) \geq 0.99m} \geq 1 - \delta / 4. 
    \end{equation*}
    We now condition on the event from the previous equation and the conclusion of \cref{lem:norm_bounds}. To extend from the net $\mc{G}$ to the whole sphere, consider $v \in \mb{S}^{d - 1}$ and its nearest neighbor $u \in \mc{G}$. Note that $\norm{v - u} \leq \gamma$. Let $i \in [m]$ be such that $W_i (u) = 1$ and $\norm{\Pi^i}_F \leq O(\sqrt{d})$. We have:
    \begin{gather*}
        \abs{\norm{\Pi^i v}^2 - \norm{\Pi^i u}^2} = \abs{(v + u)^\top (\Pi^i)^\top \Pi^i (v - u)} \leq 2 \cdot \norm{\Pi^i}_F^2 \cdot \gamma \leq \eps / 256.
    \end{gather*}
    Furthermore, we have for all $x, y \in X$:
    \begin{align*}
        &\abs{\inp{\Pi^i v}{\Pi^i (x - y)} - \inp{v}{x - y}} \\
        &\qquad \leq \abs{\inp{\Pi^i u}{\Pi^i (x - y)} - \inp{u}{x - y}} + \abs{\inp{\Pi^i (v - u)}{\Pi^i (x - y)}} + \abs{\inp{v - u}{x - y}} \\
        &\qquad \leq \frac{\eps}{256}\cdot \norm{x - y} + \gamma \cdot \norm{\Pi^i}_F^2 \cdot \norm{x - y} + \gamma \cdot \norm{x - y} \leq \frac{\eps}{128} \cdot \norm{x - y}.
    \end{align*}
    Since, for any $u \in \mc{G}$, at least $0.98m$ of the $\Pi^i$ satisfy $W_i (u) = 1$ and $\norm{\Pi^i}_F \leq O(\sqrt{d})$ with probability at least $1 - \delta / 2$, the conclusion of the lemma follows. 
\end{proof}

\paragraph{Proof of \cref{thm:jl_rep}:} Finally, to establish \cref{thm:jl_rep}, we will need to use a more intricate multi-scale gridding argument than the one used to prove \cref{lem:unit_rep}. Using a single grid of resolution $\gamma$ does not suffice as the dataset, $X$, may contain pairs of points separated by much less than $\gamma$. Bounding the error of the embedding of $q$ in terms of its nearest neighbor in the net does not suffice in such situations. On the other hand, using a finer net whose resolution is comparable to the minimum distance between the points in the dataset leads to a choice of $m$ dependent on the aspect ratio of $X$. The multi-scale argument presented here allows us to circumvent these difficulties.

To define the grid, let $r_{ij} = \norm{x_i - x_j}$ for $x_i, x_j \in X$ and $\mc{G}_{ij}$ be a $\gamma$-net of $\mb{B}(x_i, 2(Cnd)^{10}r_{ij})$ with $\gamma = (Cnd)^{-10}\cdot r_{ij}$ for some large enough constant $C$. The grid in our argument will consist of the union of all the $\mc{G}_{ij}$; that is $\mc{G} = \bigcup_{i,j \in [n]} \mc{G}_{ij}$. Now, for $u \in \mc{G}$, define $W_i (u)$ as follows:
\begin{equation*}
    W_i (u) = \bm{1} \lbrb{\Pi^i \text{ satisfies } \ref{eq:ap_ip}(\eps / 256, X \cup \{u\})}.
\end{equation*}
From \cref{cor:ip_pres}, we have $\mb{P}(W_i (u) = 1) \geq 0.995$. Noting that $\abs{\mc{G}} \leq (2nd)^{O(d)}$, we have by Hoeffding's Inequality and the union bound that with probability at least $1 - \delta / 4$, we have for all $u \in \mc{G}$: $\sum_{i = 1}^m W_i (u) \geq 0.99m$. For the rest of the argument, we will also condition on the conclusions of \cref{lem:unit_rep}. Note, that this event occurs with probability at least $1 - \delta$ from the union bound. Therefore, we have by the union bound with probability at least $1 - \delta$:
\begin{equation*}
    \forall u \in \mc{G}, \forall \norm{v} = 1: \sum_{i = 1}^m \bm{1} \lbrb{
        \begin{gathered}
            (1 - \eps / 128) \leq \norm{\Pi^i v}^2 \leq (1 + \eps / 128) \\
            \forall x, y \in X: \abs{\inp{\Pi^i v}{\Pi^i (x - y)} - \inp{v}{x - y}} \leq \frac{\eps}{128} \cdot \norm{x - y} \\
            \Pi^i \text{ satisfies } \ref{eq:ap_ip} (\eps / 256, X \cup \{u\}) \\
            \norm{\Pi^i}_F \leq O(\sqrt{d}) 
        \end{gathered}
    }
    \geq 0.95m.
\end{equation*}

Letting $Y_i (u, v)$ denote the indicator in the above expression, we now condition on the above event for the rest of the proof. Let $q \in \mb{R}^d$ and $x_q = \argmin_{x \in X} \norm{q - x}$ (its closest neighbor in $X$). Note that the case where $q = x_q$ is already covered by the condition on the $W_i$. Therefore, we assume $q \neq x_q$. With $v_q = \frac{(q - x_q)}{\norm{q - x_q}}$ and $\wt{q} = \argmin_{u \in \mc{G}} \norm{q - u}$, let $\mc{J} (q) = \{i: Y_i (\wt{q}, v_q) = 1\}$. We will now prove for all $i \in \mc{J} (q)$:
\begin{equation*}
    \forall \wt{x}, \wt{y}, \wt{z} \in X \cup \{q\}: \abs{\inp{\Pi^i (\wt{x} - \wt{z})}{\Pi^i (\wt{y} - \wt{z})} - \inp{\wt{x} - \wt{z}}{\wt{y} - \wt{z}}} \leq \eps \norm{\wt{x} - \wt{z}} \norm{\wt{y} - \wt{z}}.
\end{equation*}
When $\wt{x} = \wt{z}$ or $\wt{y} = \wt{z}$, the conclusion is trivial. Furthermore, for $\wt{x}, \wt{y}, \wt{z} \in X$, the conclusion follows from the definition of $\mc{J}$. Hence, we may restrict ourselves to cases where $\wt{x}, \wt{y}$ are distinct from $\wt{z}$ and at least one of $\wt{x}, \wt{y}, \wt{z}$ are $q$. We first tackle the cases where $\wt{x}, \wt{y}, \wt{z}$ are distinct and we have the following subcases:

\begin{enumerate}
    \item[] \textbf{Case 1: } $\wt{x} = q$ and $\wt{y}, \wt{z} \in X$. In this case, we have from the definition of $\mc{J}$:
    \begin{align*}
        &\abs{\inp{\Pi^i (q - \wt{z})}{\Pi^i (\wt{y} - \wt{z})} - \inp{q - \wt{z}}{\wt{y} - \wt{z}}} \\
        &\quad \leq \abs{\inp{\Pi^i (q - x_q)}{\Pi^i (\wt{y} - \wt{z})} - \inp{q - x_q}{\wt{y} - \wt{z}}} + \abs{\inp{\Pi^i (x_q - \wt{z})}{\Pi^i (\wt{y} - \wt{z})} - \inp{x_q - \wt{z}}{\wt{y} - \wt{z}}} \\
        &\quad \leq \frac{\eps}{128} \cdot \norm{q - x_q}\cdot \norm{\wt{y} - \wt{z}} + \frac{\eps}{256} \cdot \norm{x_q - \wt{z}} \cdot \norm{\wt{y} - \wt{z}} \\
        &\quad \leq \frac{\eps}{128} \cdot \norm{q - \wt{z}}\cdot \norm{\wt{y} - \wt{z}} + \frac{\eps}{256} \cdot 2 \norm{q - \wt{z}} \cdot \norm{\wt{y} - \wt{z}} \leq \frac{\eps}{64} \cdot \norm{q - \wt{z}} \cdot \norm{\wt{y} - \wt{z}}
    \end{align*}
    concluding the proof in this case.
    \item[] \textbf{Case 2: } $\wt{z} = q$ and $\wt{x}, \wt{y} \in X$. We have by algebraic expansion:
    \begin{align*}
        &\abs{\inp{\Pi^i (\wt{x} - q)}{\Pi^i (\wt{y} - q)} - \inp{\wt{x} - q}{\wt{y} - q}} = \abs{(\wt{x} - q)^\top ((\Pi^i)^\top \Pi^i - I) (\wt{y} - q)} \\
        &\leq \abs{(\wt{x} - x_q)^\top ((\Pi^i)^\top \Pi^i - I) (\wt{y} - x_q)} +  \abs{(x_q - q)^\top ((\Pi^i)^\top \Pi^i - I) (\wt{y} - x_q)} + \\
        &\quad\ \,\abs{(\wt{x} - x_q)^\top ((\Pi^i)^\top \Pi^i - I) (x_q - q)} + \abs{(x_q - q)^\top ((\Pi^i)^\top \Pi^i - I) (x_q - q)} \\
        &\leq \frac{\eps}{256} \cdot \norm{x_q - \wt{x}} \norm{x_q - \wt{y}} + \frac{\eps}{128}\cdot \norm{x_q - q} \norm{\wt{y} - x_q} \\
        &\qquad + \frac{\eps}{128} \cdot \norm{\wt{x} - x_q} \norm{q - x_q} + \frac{\eps}{128}\cdot \norm{x_q - q}^2 \\
        &\leq \frac{\eps}{256} \cdot 2\norm{\wt{x} - q} \cdot 2\norm{\wt{y} - q} + \frac{\eps}{128}\cdot \norm{\wt{x} - q} \cdot 2 \norm{\wt{y} - q} \\
        &\qquad + \frac{\eps}{128} \cdot 2\norm{\wt{x} - q}\cdot \norm{\wt{y} - q} + \frac{\eps}{128}\cdot \norm{\wt{x} - q}\norm{\wt{y} - q} \\
        &\leq \frac{\eps}{16} \cdot \norm{\wt{x} - q} \cdot \norm{\wt{y} - q}
    \end{align*}
    establishing the statement in this case as well.
\end{enumerate}
We now move on to the case where $\wt{x}, \wt{y}, \wt{z}$ are not distinct. As remarked before, it suffices to consider $x = \wt{x} = \wt{y} \neq \wt{z}$ and one of $x, \wt{z}$ are $q$. Without loss of generality, we may assume $\wt{z} = q$. When $x = x_q$ the conclusion follows from the definition of $\mc{J}(q)$. As a consequence, our goal simplifies to establishing for all $i \in \mc{J} (q)$ and $x \neq x_q$:
\begin{equation*}
    \abs{\norm{\Pi^i (q - x)}^2 - \norm{q - x}^2} \leq \eps \norm{q - x}^2.
\end{equation*}
Let $r = \norm{x - x_q}$ and here again, we have two cases:
\begin{enumerate}
    \item[] \textbf{Case 1: } $\norm{q - x_q} \leq 2(Cnd)^{10} r$. From the construction of $\mc{G}$, we have $\norm{\wt{q} - q} \leq (Cnd)^{-10} r$. Furthermore, we have from the fact that $r \leq 2 \norm{q - x}$:
    \begin{align*}
        \abs{\norm{q - x}^2 - \norm{\wt{q} - x}^2} &= \abs{\inp{q - \wt{q}}{(q - x) + (\wt{q} - x)}} \\
        &\leq (Cnd)^{-10} \cdot r \cdot (\norm{q - x} + \norm{\wt{q} - x}) \leq \frac{\eps}{1024} \cdot \norm{q - x}^2.
    \end{align*}
    Additionally, we have:
    \begin{align*}
        &\abs{(\norm{\Pi^i (x - q)}^2 - \norm{x - q}^2) - (\norm{\Pi^i (x - \wt{q})}^2 - \norm{x - \wt{q}}^2)} \\
        &\qquad = \abs{(x - q)^\top ((\Pi^i)^\top \Pi^i - I) (x - q) - (x - \wt{q})^\top ((\Pi^i)^\top \Pi^i - I) (x - \wt{q})} \\
        &\qquad = \abs{(\wt{q} - q)^\top ((\Pi^i)^\top \Pi^i - I) ((x - q) + (x - \wt{q}))} \\
        &\qquad \leq \norm{\wt{q} - q} (\norm{\Pi^i}_F^2 + 1) (\norm{x - q} + \norm{x - \wt{q}}) \leq \frac{\eps}{1024} \cdot \norm{x - q}^2.
    \end{align*}
    The conclusion now follows from the inequalities in the previous two displays and the following condition on $\Pi^i$ in the definition of $\mc{J}$:
    \begin{equation*}
        \abs{\norm{\Pi^i (\wt{q} - x)}^2 - \norm{\wt{q} - x}^2} \leq \frac{\eps}{256} \cdot \norm{\wt{q} - x}^2 \leq \frac{\eps}{128} \norm{q - x}^2.
    \end{equation*}
    \item[] \textbf{Case 2: } $\norm{q - x_q} \geq 2(Cnd)^{10} r$. We have similarly to the previous case:
    \begin{align*}
        \abs{\norm{q - x}^2 - \norm{q - x_q}^2} &= \abs{\inp{x_q - x}{(q - x) + (q - x_q)}} \\
        &\leq r \cdot (\norm{q - x} + \norm{q - x_q}) \leq \frac{\eps}{1024} \cdot \norm{q - x}^2.
    \end{align*}
    Now, we have:
    \begin{align*}
        &\abs{(\norm{\Pi^i (q - x)}^2 - \norm{q - x}^2) - (\norm{\Pi^i (q - x_q)}^2 - \norm{q - x_q}^2)} \\
        &\qquad = \abs{(q - x)^\top ((\Pi^i)^\top \Pi^i - I) (q - x) - (q - x_q)^\top ((\Pi^i)^\top \Pi^i - I) (q - x_q)} \\
        &\qquad = \abs{(x_q - x)^\top ((\Pi^i)^\top \Pi^i - I) ((q - x) + (q - x_q))} \\
        &\qquad \leq \norm{x_q - x} (\norm{\Pi^i}_F^2 + 1) (\norm{q - x} + \norm{q - x_q}) \leq \frac{\eps}{1024} \cdot \norm{q - x}^2.
    \end{align*}
    Similarly, we conclude from the following inequality implied by the definition of $\mc{J}$:
    \begin{equation*}
        \abs{\norm{\Pi^i (q - x_q)}^2 - \norm{q - x_q}^2} \leq \frac{\eps}{128} \cdot \norm{q - x_q}^2 \leq \frac{\eps}{128}\cdot \norm{q - x}^2.
    \end{equation*}
\end{enumerate}

The four cases just enumerated establish the theorem assuming $\abs{\mc{J} (q)} \geq 0.95m$ for all $q \in \mb{R}^d$. As shown before, this occurs with probability at least $1 - \delta$ concluding the proof of the theorem.
\qed

\section{Acknowledgements}
\label{sec:ack}

The authors would like to thank Sidhanth Mohanty for enlightening conversations in the course of this project and helpful feedback in the preparation of this manuscript.

\newpage 
\printbibliography

@inproceedings{hp,
  author    = {Sariel Har{-}Peled},
  title     = {A Replacement for Voronoi Diagrams of Near Linear Size},
  booktitle = {42nd Annual Symposium on Foundations of Computer Science, {FOCS} 2001,
               14-17 October 2001, Las Vegas, Nevada, {USA}},
  pages     = {94--103},
  year      = {2001},
  crossref  = {DBLP:conf/focs/2001},
  url       = {https://doi.org/10.1109/SFCS.2001.959884},
  doi       = {10.1109/SFCS.2001.959884},
  bibsource = {dblp computer science bibliography, https://dblp.org}
}

@proceedings{DBLP:conf/focs/2001,
  title     = {42nd Annual Symposium on Foundations of Computer Science, {FOCS} 2001,
               14-17 October 2001, Las Vegas, Nevada, {USA}},
  publisher = {{IEEE} Computer Society},
  year      = {2001},
  url       = {https://ieeexplore.ieee.org/xpl/conhome/7601/proceeding},
  isbn      = {0-7695-1390-5},
  bibsource = {dblp computer science bibliography, https://dblp.org}
}

@inproceedings{adaptiveds,
  author    = {Yeshwanth Cherapanamjeri and
               Jelani Nelson},
  title     = {On Adaptive Distance Estimation},
  booktitle = {Advances in Neural Information Processing Systems 33: Annual Conference
               on Neural Information Processing Systems 2020, NeurIPS 2020, December
               6-12, 2020, virtual},
  year      = {2020},
  crossref  = {DBLP:conf/nips/2020},
  url       = {https://proceedings.neurips.cc/paper/2020/hash/803ef56843860e4a48fc4cdb3065e8ce-Abstract.html},
  bibsource = {dblp computer science bibliography, https://dblp.org}
}

@proceedings{DBLP:conf/nips/2020,
  editor    = {Hugo Larochelle and
               Marc'Aurelio Ranzato and
               Raia Hadsell and
               Maria{-}Florina Balcan and
               Hsuan{-}Tien Lin},
  title     = {Advances in Neural Information Processing Systems 33: Annual Conference
               on Neural Information Processing Systems 2020, NeurIPS 2020, December
               6-12, 2020, virtual},
  year      = {2020},
  url       = {https://proceedings.neurips.cc/paper/2020},
  bibsource = {dblp computer science bibliography, https://dblp.org}
}

@article{ElkinFN17,
  author       = {Michael Elkin and
                  Arnold Filtser and
                  Ofer Neiman},
  title        = {Terminal embeddings},
  journal      = {Theor. Comput. Sci.},
  volume       = {697},
  pages        = {1--36},
  year         = {2017},
  url          = {https://doi.org/10.1016/j.tcs.2017.06.021},
  doi          = {10.1016/J.TCS.2017.06.021},
  bibsource    = {dblp computer science bibliography, https://dblp.org}
}

@book {ny83,
    AUTHOR = {Nemirovsky, A. S. and Yudin, D. B. and },
     TITLE = {Problem complexity and method efficiency in optimization},
    SERIES = {A Wiley-Interscience Publication},
      NOTE = {Translated from the Russian and with a preface by E. R. Dawson,
              Wiley-Interscience Series in Discrete Mathematics},
 PUBLISHER = {John Wiley \& Sons, Inc., New York},
      YEAR = {1983},
     PAGES = {xv+388},
      ISBN = {0-471-10345-4},
   MRCLASS = {90C25 (68C25)},
  MRNUMBER = {702836},
}

@book {bertsekas,
    AUTHOR = {Bertsekas, Dimitri P.},
     TITLE = {Nonlinear programming},
    SERIES = {Athena Scientific Optimization and Computation Series},
   EDITION = {Third},
 PUBLISHER = {Athena Scientific, Belmont, MA},
      YEAR = {2016},
     PAGES = {xviii+861},
      ISBN = {978-1-886529-05-2; 1-886529-05-1},
   MRCLASS = {49-01 (90-01)},
  MRNUMBER = {3587371},
MRREVIEWER = {Stephan Dempe},
}

@inproceedings{mmmr18,
  author    = {Sepideh Mahabadi and
               Konstantin Makarychev and
               Yury Makarychev and
               Ilya P. Razenshteyn},
  title     = {Nonlinear dimension reduction via outer Bi-Lipschitz extensions},
  booktitle = {Proceedings of the 50th Annual {ACM} {SIGACT} Symposium on Theory
               of Computing, {STOC} 2018, Los Angeles, CA, USA, June 25-29, 2018},
  pages     = {1088--1101},
  year      = {2018},
  crossref  = {DBLP:conf/stoc/2018},
  url       = {https://doi.org/10.1145/3188745.3188828},
  doi       = {10.1145/3188745.3188828},
  bibsource = {dblp computer science bibliography, https://dblp.org}
}

@proceedings{DBLP:conf/stoc/2018,
  editor    = {Ilias Diakonikolas and
               David Kempe and
               Monika Henzinger},
  title     = {Proceedings of the 50th Annual {ACM} {SIGACT} Symposium on Theory
               of Computing, {STOC} 2018, Los Angeles, CA, USA, June 25-29, 2018},
  publisher = {{ACM}},
  year      = {2018},
  url       = {http://dl.acm.org/citation.cfm?id=3188745},
  bibsource = {dblp computer science bibliography, https://dblp.org}
}

@inproceedings{nn19,
  author    = {Shyam Narayanan and
               Jelani Nelson},
  title     = {Optimal terminal dimensionality reduction in Euclidean space},
  booktitle = {Proceedings of the 51st Annual {ACM} {SIGACT} Symposium on Theory
               of Computing, {STOC} 2019, Phoenix, AZ, USA, June 23-26, 2019},
  pages     = {1064--1069},
  year      = {2019},
  crossref  = {DBLP:conf/stoc/2019},
  url       = {https://doi.org/10.1145/3313276.3316307},
  doi       = {10.1145/3313276.3316307},
  bibsource = {dblp computer science bibliography, https://dblp.org}
}

@proceedings{DBLP:conf/stoc/2019,
  editor    = {Moses Charikar and
               Edith Cohen},
  title     = {Proceedings of the 51st Annual {ACM} {SIGACT} Symposium on Theory
               of Computing, {STOC} 2019, Phoenix, AZ, USA, June 23-26, 2019},
  publisher = {{ACM}},
  year      = {2019},
  url       = {https://dl.acm.org/citation.cfm?id=3313276},
  isbn      = {978-1-4503-6705-9},
  bibsource = {dblp computer science bibliography, https://dblp.org}
}

@incollection {jl,
    AUTHOR = {Johnson, William B. and Lindenstrauss, Joram},
     TITLE = {Extensions of {L}ipschitz mappings into a {H}ilbert space},
 BOOKTITLE = {Conference in modern analysis and probability ({N}ew {H}aven,
              {C}onn., 1982)},
    SERIES = {Contemp. Math.},
    VOLUME = {26},
     PAGES = {189--206},
 PUBLISHER = {Amer. Math. Soc., Providence, RI},
      YEAR = {1984},
   MRCLASS = {46B20},
  MRNUMBER = {737400},
MRREVIEWER = {Yehoram Gordon},
       DOI = {10.1090/conm/026/737400},
       URL = {https://doi.org/10.1090/conm/026/737400},
}

@article{him,
  author    = {Sariel Har{-}Peled and
               Piotr Indyk and
               Rajeev Motwani},
  title     = {Approximate Nearest Neighbor: Towards Removing the Curse of Dimensionality},
  journal   = {Theory of Computing},
  volume    = {8},
  number    = {1},
  pages     = {321--350},
  year      = {2012},
  url       = {https://doi.org/10.4086/toc.2012.v008a014},
  doi       = {10.4086/toc.2012.v008a014},
  bibsource = {dblp computer science bibliography, https://dblp.org}
}

@inproceedings{alrw,
  author    = {Alexandr Andoni and
               Thijs Laarhoven and
               Ilya P. Razenshteyn and
               Erik Waingarten},
  title     = {Optimal Hashing-based Time-Space Trade-offs for Approximate Near Neighbors},
  booktitle = {Proceedings of the Twenty-Eighth Annual {ACM-SIAM} Symposium on Discrete
               Algorithms, {SODA} 2017, Barcelona, Spain, Hotel Porta Fira, January
               16-19},
  pages     = {47--66},
  year      = {2017},
  crossref  = {DBLP:conf/soda/2017},
  url       = {https://doi.org/10.1137/1.9781611974782.4},
  doi       = {10.1137/1.9781611974782.4},
  bibsource = {dblp computer science bibliography, https://dblp.org}
}

@proceedings{DBLP:conf/soda/2017,
  editor    = {Philip N. Klein},
  title     = {Proceedings of the Twenty-Eighth Annual {ACM-SIAM} Symposium on Discrete
               Algorithms, {SODA} 2017, Barcelona, Spain, Hotel Porta Fira, January
               16-19},
  publisher = {{SIAM}},
  year      = {2017},
  url       = {https://doi.org/10.1137/1.9781611974782},
  doi       = {10.1137/1.9781611974782},
  isbn      = {978-1-61197-478-2},
  bibsource = {dblp computer science bibliography, https://dblp.org}
}

@article{nearoptHashHighDim,
  author    = {Alexandr Andoni and
               Piotr Indyk},
  title     = {Near-optimal hashing algorithms for approximate nearest neighbor in
               high dimensions},
  journal   = {Commun. {ACM}},
  volume    = {51},
  number    = {1},
  pages     = {117--122},
  year      = {2008},
  url       = {https://doi.org/10.1145/1327452.1327494},
  doi       = {10.1145/1327452.1327494},
  bibsource = {dblp computer science bibliography, https://dblp.org}
}

@book {vershynin,
    AUTHOR = {Vershynin, Roman},
     TITLE = {High-dimensional probability},
    SERIES = {Cambridge Series in Statistical and Probabilistic Mathematics},
    VOLUME = {47},
      NOTE = {An introduction with applications in data science,
              With a foreword by Sara van de Geer},
 PUBLISHER = {Cambridge University Press, Cambridge},
      YEAR = {2018},
     PAGES = {xiv+284},
      ISBN = {978-1-108-41519-4},
   MRCLASS = {60-01 (60B05 60B20 60E15 60Fxx 62H25)},
  MRNUMBER = {3837109},
MRREVIEWER = {Sasha Sodin},
       DOI = {10.1017/9781108231596},
       URL = {https://doi.org/10.1017/9781108231596},
}

\appendix
\section{Construct Partition Tree}
\label{sec:cons_tree}

In this section, we discuss our construction of a Partition Tree (\cref{def:part_tree}) and prove \cref{lem:cons_tree} which we restate below:

\RestateGo{\restateconstree}
\constree*

The algorithm works by recursively constructing the nodes of the tree starting from the root and then partitioning the point set for that node to construct its children and so on. In \cref{ssec:ap_part} we show how to partition points at a single node and in \cref{ssec:cons_part_tree}, we use this to construct the full tree. 

\subsection{Partitioning at a Single Node}
\label{ssec:ap_part}

In this subsection, we describe the partitioning procedure at a single node. This will then be incorporated into a recursive procedure to construct the entire tree. To begin, recall \cref{def:gp_cc,def:part_ref} and the definition of $\rmed$ from \cref{sec:techniques}. While it is possible to compute $\rmed(X)$ in time $\Ot(n^2d)$, obtaining a crude estimate is sufficient for our purposes \cite{him}. Furthermore, we will also not require computing $\cc (X, r)$ exactly but appropriate refinements and coarsenings suffice. We first restate a simple lemma from \cite{him}:

\begin{lemma}
    \label{lem:rmed_apx}
    Given $X = \{x_i\}_{i = 1}^n \subset \mb{R}^d$ and $\delta \in (0, 1)$, there is a randomized algorithm, $\comprmed$, that computes in time $O(nd \log 1 / \delta)$ and outputs an estimate $\rhat$ satisfying:
    \begin{equation*}
        \mb{P} \lbrb{\rhat \geq \rmed(X)} = 1 \text{ and } \mb{P} \lbrb{\rhat \leq n \rmed(X)} \geq 1 - \delta.
    \end{equation*}
\end{lemma}
\begin{proof}
    Let $C \in \cc (X, \rmed(X))$ be such that $\abs{C} \geq n / 2$. Picking a point, $x$, uniformly at random from $X$ picks a point in $C$ with probability at least $1 / 2$. Now, we compute distances $\{\norm{x_i - x}\}_{i = 1}^n$ and output their median, $\rhat$. Conditioned on $x \in C$, we have by the triangle inequality, that $\rhat \leq n \rmed (X)$ which proves the second claim with probability at least $1/2$. For the first, note that $x$ belongs to a connected component in $\cc (X, \rhat)$ of size at least $n / 2$. This establishes the first claim of the lemma. By repeating this procedure $\Omega (\log 1 / \delta)$ times and taking the minimum of the returned estimates establishes the lemma by an application of Hoeffding's inequality. 
\end{proof}

\begin{lemma}
    \label{lem:ap_part}
    Let $X = \{x_i\}_{i = 1}^n \subset \mb{R}^d$, $r > 0$ and $\delta \in (0, 1)$. Then, there is a randomized algorithm, $\conpart$ that outputs a partitioning of $X$, $\mc{C}$, satisfying:
    \begin{equation*}
        \cc (X, 1000n^2 r) \sqsubseteq \mc{C} \sqsubseteq \cc (X, r)
    \end{equation*}
    with probability at least $1 - \delta$. Furthermore, the algorithm runs in time $O(nd \log (n / \delta))$.
\end{lemma}
\begin{proof}
  The randomized algorithm is detailed in the following pseudocode.
  
  \setcounter{AlgoLine}{0}
  \begin{algorithm}
            \KwIn{Point set $X = \{x_i\}_{i = 1}^n \subset \mb{R}^d$, Resolution $r$, \newline Failure Probability $\delta$}
             $K \gets 10 \log (n / \delta), \tau \gets 10 r, \nu \gets 1000 n^2r$\;
             $V \gets X, E \gets \phi$\;
            \For {$k = 1:K$}{
                 $g_k \thicksim \mc{N}(0, I)$\;
                 For all $i \in [n]$, let $v^{(k)}_i = \inp{x_i}{g_k}$\;
                 Let $x^{(k)}_1, \dots, x^{(k)}_n$ be an ordering of the $x_i$ increasing in $v^{(k)}_i$\;
                 $i \gets 1$\;
                \While {$i < n$}{
                     $i_n = \max \{i < j \leq n: v^{(k)}_j \leq v^{(k)}_i + \tau\}$\;
                     Add $(x^{(k)}_i, x^{(k)}_j)$ to $E$ for $j \in \{i + 1, \dots, i_n\}$ if $\norm{x^{(k)} - x^{(k)}} \leq \nu$\;
                     $i \gets \max(i + 1, i_n)$\;
                }
            }
             $\mc{C} \gets \mathrm{ConnectedComponents}(V, E)$\;
            \Return{$\mc{C}$}
        \caption{$\conpart (X, r, \delta)$}
        \label{alg:con_part}
    \end{algorithm}
    By the definition of \cref{alg:con_part}, we see that for every $(x, y) \in E$, we must have $\norm{x - y} \leq \nu$. Therefore, we obtain $\mc{C}$ refines $\cc (X, 1000 n^2r)$. We now show that $\cc (X, r)$ refines $\mc{C}$. To do this, we will need the following claim:
    \begin{claim}      
        \label{clm:prod}
        For $x, y \in \mb{R}^d$ and $g \thicksim \mc{N}(0, I)$, we have:
        \begin{gather*}
            \mb{P} \lbrb{\abs{\inp{g}{x - y}} \leq \tau} \geq \frac{999}{1000} \text{ if } \norm{x - y} \leq r \\
            \mb{P} \lbrb{\abs{\inp{g}{x - y}} \leq \tau} \leq \frac{1}{100 n^2} \text{ if } \norm{x - y} \geq \nu.
        \end{gather*}
    \end{claim}
    \begin{subproof}
        The first inequality follows from the observation that $\inp{g}{x - y}$ is a Gaussian with variance ${\norm{x - y}^2}$ and the definition of $\tau$. The second follows from the same fact and the fact that pdf of the standard normal distribution takes maximum value $1 / \sqrt{2\pi}$.
    \end{subproof}
    Now, fix an outer iteration $k$. From \cref{clm:prod}, we have with probability at least $1 / 100$, that $\abs{v_i^{(k)} - v_j^{(k)}} > \tau$ for all $\norm{x_i - x_j} \geq \nu$. Let $x_i, x_j \in X$ with $\norm{x_i - x_j} \leq r$. Then with probability at least $0.999$, we have that $\abs{v^{(k)}_i - v^{(k)}_j} \leq \tau$. For the rest of the argument, we condition on the previous two events and we assume, without loss of generality, that $v^{(k)}_i \leq v^{(k)}_j$. Note that since $v^{(k)}_j - v^{(k)}_i \leq \tau$, one of the following cases must occur:
    \begin{enumerate}
        \item[] \textbf{Case 1:} $(x_i, x_j) \in E$,
        \item[] \textbf{Case 2:} $(z, x_i), (z, x_j) \in E$ for some $z \in X$ or
        \item[] \textbf{Case 3:} $(z, x_i), (z, w), (w, x_j) \in E$ for some $z, w \in X$.
    \end{enumerate}
    In all three cases, we see that $x_i$ and $x_j$ are in the same connected component with respect to $E$. Note that $x_i, x_j$ are in the same connected component if our good event occurs for at least $1$ of the $K$ outer iterations of the algorithm. Therefore, the probability that $x_i, x_j$ are in the same connected component over all $K$ runs is at least $1 - \delta / n^2$. By a union bound, this establishes that with probability at least $1 - \delta$, $x_i, x_j$ are in the same connected component in $E$ for all $\norm{x_i - x_j} \leq r$.

    The runtime of the algorithm follows from the fact that we add at most $O(n \log (n / \delta))$ edges to $E$ over the entire run of the algorithm and we do at most $O(nd)$ amount of work in each run of the while loop. 
\end{proof}

\subsection{Constructing the Partition Tree - Proof of \cref{lem:cons_tree}}
\label{ssec:cons_part_tree}
We will now use the results of \cref{ssec:ap_part} to construct the whole tree, proving \cref{lem:cons_tree}.

    We will first assume that the functions $\comprmed$ and $\conpart$ run successfully (that is, they satisfy the conclusions of \cref{lem:rmed_apx,lem:ap_part} respectively) in every recursive call of \cref{alg:cons_tree} and then finally bound the probability of this event. Note that when $\comprmed$ runs successfully, $\rhat$ always satisfies $\rmed \leq \rhat \leq n\rmed$ for every recursive call of \cref{alg:cons_tree} (\cref{lem:rmed_apx}). Together with the correctness of $\conpart$ (\cref{lem:ap_part}), we get that:
    \begin{equation*}
        \cc (Z, 1000n^2\rhat) \sqsubseteq \chigh \sqsubseteq \cc (Z, \rhat) \sqsubseteq \cc (Z, \rmed) \sqsubseteq \cc \lprp{Z, \frac{\rhat}{10n}} \sqsubseteq \clow \sqsubseteq \cc \lprp{Z, \frac{\rhat}{1000n^3}}
    \end{equation*}
    for every node $\mc{T}^\prime = \lprp{Z, \lbrb{\mc{T}_C}_{C \in \clow}, \mc{T}_{\mathrm{rep}}, \clow, \chigh, \crep, \rhat} \in \mc{T}$. Furthermore, for any such $\mc{T}^\prime$, we get from the fact that $\cc (Z, \rmed) \sqsubseteq \clow$, that all $C \in \clow$ satisfy $\abs{C} \leq \abs{Z} / 2$. In addition, the definition of $\rmed$ and the fact that $\chigh \sqsubseteq \cc (Z, \rmed)$ yield $\abs{C_{\mathrm{rep}}} \leq \abs{Z} / 2$. To bound, first note that $\abs{C_{\mathrm{rep}}} \leq \abs{\clow}$ from the fact that $\chigh \sqsubseteq \clow$ and the construction of $C_{\mathrm{rep}}$. To bound $\Size (\mc{T})$, define $B(n)$ as follows:
    \begin{equation*}
        B(n) = n + \max_{\substack{n_1, \dots, n_k, k \\ \sum_{i = 1}^k n_i = n \\ k \leq n / 2, \forall i \in [k] : 1 \leq n_i \leq n / 2}} B(k) + \sum_{i = 1}^k B(n_i) \text{ and } B(1) = 1.
    \end{equation*}
    From the definition of $B(n)$, we see that $B(n)$ is monotonic in $n$ and from this, we get that $B(n)$ is an upper bound on $\Size (\mc{T})$. We now recall the following claim from \cite{him}:
    \begin{claim}[\cite{him}]
        \label{clm:him_bb}
        For all $n \geq 3$, $B(n) \leq C n \log n$.
    \end{claim}
    The above claim establishes the bound on $\Size (\mc{T})$. 

    Finally, we bound the probability that any execution of $\comprmed$ and $\conpart$ fail. We start by bounding the probability that any of the first $5 B(n)$ runs of $\comprmed$ and $\conpart$ fail. From the definition of $\delta^{\dagger}$, the probability that any of the $5 B(n)$ runs of $\comprmed$ and $\conpart$ fail is at most $1 - \delta$ by the union bound. However, the preceding argument shows that the algorithm terminates with fewer than $B(n)$ recursive calls if none of the executions of $\comprmed$ and $\conpart$ fail. Therefore, the probability that any of the executions of $\comprmed$ and $\conpart$ fail in the running of the algorithm is at most $1 - \delta$. This yields the previously derived conclusions with probability at least $1 - \delta$. 
  \qed

  \setcounter{AlgoLine}{0}
  \begin{algorithm}[H]
        \KwIn{Point set $Z = \{x_i\}_{i = 1}^m \subset \mb{R}^d$, Failure Probability $\delta$, Total Number of points $n$}
        \If {$\abs{X} = 1$}{
             \Return{$(X, \phi, \phi, \phi, \phi, \phi)$}
        }
      $\delta^\dagger \gets \frac{\cprob \delta}{n^2}$\;
       $\rhat \gets \comprmed(Z, \delta^\dagger)$\;
       $\clow \gets \conpart(Z, \rhat / (1000n^3), \delta^\dagger)$, $\chigh \gets \conpart(Z, \rhat, \delta^\dagger)$\;
      For $C \in \clow$, let $\mc{T}_C \gets \mathrm{ConstructPartitionTree}(C, n, \delta)$\;
      For $C \in \chigh$, pick representative $x \in C$ and add it to $C_{\mathrm{rep}}$\;
      $\mc{T}_{\mathrm{rep}} \gets \mathrm{ConstructPartitionTree}(C_{\mathrm{rep}}, n, \delta)$\;
        \Return{$\lprp{Z, \lbrb{\mc{T}_C}_{C \in \clow}, \mc{T}_{\mathrm{rep}}, \clow, \chigh, C_{\mathrm{rep}}, \rhat}$}
    \caption{$\mathrm{ConstructPartitionTree} (X, n, \delta)$}
    \label{alg:cons_tree}
\end{algorithm}


\section{Miscellaneous Results}
\label{sec:misc_res}
In this section, we develop some standard tools needed for our constructions. In \cref{ssec:ellips}, we recall some basic facts about the Ellipsoid algorithm for convex optimization \cite{ny83,bertsekas} and analyze it when it's instantiated with a weak oracle as in our terminal embedding construction. 

\subsection{Ellipsoid Algorithm Preliminaries}
\label{ssec:ellips}

We recall some basic facts regarding the operation of the Ellipsoid algorithm for convex optimization. We will instead use a weaker version where your goal is simply to output a feasible point in a convex set. In what follows, our goal is to find a point in a closed convex set, $K \subset \mb{R}^d$ and we are given a starting $x$ and $R \geq 0$ such that for all $K \subseteq \mb{B} (x, R)$. Furthermore, the algorithm assumes access to an oracle $\mc{O}$ which when given a point $x \in \mb{R}^d$ either:
\begin{enumerate}
    \item Outputs $v \neq 0$ such that $\forall y \in K$, $\inp{y - x}{v} \geq 0$ or
    \item Outputs $\mrm{FAIL}$.
\end{enumerate}
Given access to such an oracle the Ellipsoid algorithm proceeds as follows:

\setcounter{AlgoLine}{0}
\begin{algorithm}[H]
  \SetKwInOut{Input}{Input}
 
       \KwIn{Initialization $x \in \mb{R}^d$, Initial Distance $R > 0$, Separating Oracle $\mc{O}$}
         $\xt{0} \gets x, \At{0} \gets R^2 \cdot I, t \gets 0$\;
        \While{$\vt{t} = \mc{O}(\xt{t}) \neq \mrm{FAIL}$}{
             $\ut{t + 1} \gets \frac{\At{t} \vt{t}}{\sqrt{(\vt{t})^\top\At{t} \vt{t}}}$\;
           $\xt{t + 1} \gets \xt{t} + \frac{1}{(d + 1)} \ut{t + 1}$\;
            $\At{t + 1} \gets \frac{d^2}{d^2 - 1} \lprp{\At{t} - \frac{2}{d + 1} \ut{t}(\ut{t})^\top}$\;
            $t \gets t + 1$\;
            }
            \Return{$\xt{t}$}
    \caption{$\mathrm{Ellipsoid} (x, R, \mc{O})$}
    \label{alg:ellipsoid}
  \end{algorithm}
  

We recall some classical facts regarding the operation of \cref{alg:ellipsoid} where $\mc{E}(x, A)$ denotes the ellipsoid $\{y \in \mb{R}^d: (y - x)^\top A^{-1} (y - x) \leq 1\}$.
\begin{lemma}[\cite{ny83,bertsekas}]
    \label{lem:ellipsoid_misc}
    Let $x \in \mb{R}^d$, $A \succ 0$. Then for any $v \neq 0$, let:
    \begin{equation*}
        u = \frac{A v}{\sqrt{v^\top A v}}, \qquad \wt{x} = x + \frac{1}{d + 1} u \quad \text{and}\quad \wt{A} = \frac{d^2}{d^2 - 1} \lprp{A - \frac{2}{d + 1} uu^\top}.
    \end{equation*}
    Then, we have:
    \begin{equation*}
        \mc{E} (x, A) \cap \{y: \inp{y - x}{v} \geq 0\} \subseteq \mc{E} (\wt{x}, \wt{A}) \quad \text{and}\quad \vol (\mc{E}(\wt{x}, \wt{A})) \leq \exp \lprp{ - \frac{1}{2(d + 1)}} \vol (\mc{E}(x, A)).
    \end{equation*}
\end{lemma}

We now use \cref{lem:ellipsoid_misc} to establish a slightly weaker guarantee for the \cref{alg:ellipsoid} corresponding to our weaker oracle.

\RestateGo{\restateweakellipsoid}
\weakellipsoid*

\begin{proof}
    Suppose \cref{alg:ellipsoid} ran for $T$ iterations with iterates, $(\xt{t}, \At{t})_{t = 0}^T$. From \cref{lem:ellipsoid_misc}, we have for any $t \in \{0, \dots, T - 1\}$, $\vol (\mc{E}(\xt{t + 1}, \At{t + 1})) \leq \exp\lprp{- \frac{1}{2(d + 1)}} \vol (\mc{E}(\xt{t}, \At{t}))$. Furthermore, \cref{lem:ellipsoid_misc} along with our assumption on $\mc{O}$ also implies that for all $t \in \{0, \dots, T\}$, $K \subseteq \mc{E}(\xt{t}, \At{t})$. From these facts, we have:
    \begin{equation*}
        \exp \lprp{- \frac{T}{2(d + 1)}} \vol (\mb{B}(x, R)) = \exp \lprp{- \frac{T}{2(d + 1)}} \vol (\mc{E}(x, R^2\cdot I)) \geq \vol (K) \geq \vol (\mb{B}(x^*, \eps)).
    \end{equation*}
    By taking rearranging the above inequality, taking logarithms on both sides and by using the homogeneity properties of Euclidean volume, we get the desired bound on the number of iterations. The bound on the computational complexity of the algorithm follows from the fact that each iteration takes time $O(d^2)$ to compute. 
\end{proof}

\subsection{Miscellaneous Technical Results}
\label{ssec:misc_tech}

Here, we present miscellaneous technical results needed in other parts of our proof. 

\RestateGo{\restateippres}
\ippres*

\begin{proof}
    We have by the definition of the inner product:
    \begin{equation*}
        \inp{x}{y} = \frac{1}{4} \lprp{\norm{x + y}^2 - \norm{x - y}^2}.
    \end{equation*}
    Now, let $x, y \in \conv(T)$. If $x = y$, the result is true from the fact that $\Pi$ has $\eps$-convex hull distortion for $X$. Therefore, assume $x \neq y$. We have:
    \begin{align*}
        \abs{\inp{\Pi x}{\Pi y} - \inp{x}{y}} &= \frac{1}{4} \lprp{\abs{\norm{\Pi (x + y)}^2 - \norm{\Pi (x - y)^2} - \norm{x + y}^2 + \norm{x - y}^2}} \\
        &\leq \frac{1}{4} \lprp{\abs{\norm{\Pi (x + y)}^2 - \norm{x + y}^2} + \abs{\norm{\Pi (x - y)^2} - \norm{x - y}^2}}.
    \end{align*}
    For the first term, we have:
    \begin{equation*}
        \abs{\norm{\Pi (x + y)}^2 - \norm{x + y}^2} = 4 \lprp{\norm*{\Pi \lprp{\frac{x + y}{2}}} + \norm*{\frac{x + y}{2}}} \abs*{\norm*{\Pi \lprp{\frac{x + y}{2}}} - \norm*{\frac{x + y}{2}}} \leq 12 \eps
    \end{equation*}
    where the final inequality follows from the fact that $\frac{x + y}{2} \in \conv (T)$, the fact that $\norm{x},\norm{y} \leq 1$ and the assumption that $\Pi$ has $\eps$-convex hull distortion for $X$. A similar inequality for the second term yields:
    \begin{equation*}
        \abs{\inp{\Pi x}{\Pi y} - \inp{x}{y}} \leq \frac{1}{4} \lprp{12\eps + 12\eps} = 6\eps
    \end{equation*}
    concluding the proof of the lemma.
\end{proof}

\end{document}